\documentclass[11pt]{article}
\usepackage[letterpaper,margin=1in]{geometry}
\usepackage{cmap}
\usepackage[T1]{fontenc}
\usepackage[utf8]{inputenc}
\usepackage[english]{babel}
\usepackage{lmodern}
\usepackage{amsmath,amssymb,amsthm,mathtools}
\usepackage{microtype}
\usepackage{booktabs}
\usepackage{multirow}
\usepackage{enumitem}
\usepackage{xcolor}

\usepackage{xspace}

\usepackage{interval}
\intervalconfig{soft open fences}

\newcommand{\braket}[2]{\left< #1 \vphantom{#2} \middle| #2 \vphantom{#1} \right>} 
\newcommand{\ketbra}[2]{\ensuremath{\ket{#1}\bra{#2}}}
\DeclarePairedDelimiter\rbra{\lparen}{\rparen}
\DeclarePairedDelimiter\sbra{\lbrack}{\rbrack}
\DeclarePairedDelimiter\cbra{\{}{\}}
\DeclarePairedDelimiter\abs{\lvert}{\rvert}

\DeclarePairedDelimiter\ceil{\lceil}{\rceil}
\DeclarePairedDelimiter\floor{\lfloor}{\rfloor}
\let\ket\relax
\DeclarePairedDelimiter\ket{\lvert}{\rangle}
\let\bra\relax
\DeclarePairedDelimiter\bra{\langle}{\rvert}

\newcommand{\supp} {\operatorname{supp}}
\newcommand{\im}{\operatorname{image}}

\newcommand{\Inf}{\operatorname{Inf}}

\newcommand{\Real} {\operatorname{Re}}

\usepackage{amsfonts,latexsym,mathdots,bbm,graphicx,float}
\usepackage{caption,subcaption,ellipsis,textcomp,thmtools,thm-restate,mleftright}
\usepackage{algorithm}
\usepackage{algpseudocode}
\definecolor{blueviolet}{rgb}{0.17,0.17,0.78}
\definecolor{webgreen}{rgb}{0,.5,0}
\definecolor{webbrown}{rgb}{.6,0,0}
\usepackage[pdftex,
  bookmarks=false,
  colorlinks=true, 
  urlcolor=webbrown,
  linkcolor=blueviolet, 
  citecolor=webgreen,
  pdfstartpage=1,
  pdfstartview={FitH},  
  bookmarksopen=false,
  pagebackref
  ]{hyperref}
\usepackage[capitalize,nameinlink]{cleveref}
\crefname{ineq}{inequality}{inequalities}
\creflabelformat{ineq}{#2{\upshape(#1)}#3}
\crefname{algorithm}{Algorithm}{Algorithms}

\declaretheorem[numberwithin=section]{theorem}
\declaretheorem[sibling=theorem]{lemma}

\declaretheorem[sibling=theorem]{proposition}
\declaretheorem[sibling=theorem]{fact}
\declaretheorem[sibling=theorem]{corollary}
\declaretheorem[style=definition,sibling=theorem]{definition}
\declaretheorem[style=definition,sibling=theorem]{problem}

\makeatletter

\def\theHALG@line{\thealgorithm.\arabic{ALG@line}}
\makeatother

\newcommand{\ignore}[1]{}

\DeclareMathOperator*{\Ex}{\mathbb{E}}

\DeclareMathOperator{\spn}{\mathrm{span}}
\DeclareMathOperator{\Tr}{\mathrm{Tr}}

\newcommand{\R}{\mathbb R}
\newcommand{\C}{\mathbb C}

\newcommand{\Z}{\mathbb Z}

\newcommand{\U}{\mathbb U}

\newcommand{\eps}{\varepsilon}
\newcommand{\lam}{\lambda}

\newcommand{\wt}[1]{\widetilde{#1}}
\newcommand{\wh}[1]{\widehat{#1}}

\newcommand{\ot}{\otimes}

\newcommand{\bdry}{\partial}

\newcommand{\Id}{I}

\newcommand{\calA}{\mathcal{A}}
\newcommand{\calB}{\mathcal{B}}
\newcommand{\calC}{\mathcal{C}}
\newcommand{\calD}{\mathcal{D}}
\newcommand{\calE}{\mathcal{E}}

\newcommand{\calH}{\mathcal{H}}
\newcommand{\calI}{\mathcal{I}}

\newcommand{\calK}{\mathcal{K}}

\newcommand{\calN}{\mathcal{N}}
\newcommand{\calO}{\mathcal{O}}

\newcommand{\calQ}{\mathcal{Q}}
\newcommand{\calR}{\mathcal{R}}
\newcommand{\calS}{\mathcal{S}}

\newcommand{\calU}{\mathcal{U}}
\newcommand{\calV}{\mathcal{V}}
\newcommand{\calW}{\mathcal{W}}
\newcommand{\calX}{\mathcal{X}}

\newcommand{\calZ}{\mathcal{Z}}

\DeclarePairedDelimiter\norm{\lVert}{\rVert}

\DeclarePairedDelimiterX\diverg[2]{(}{)}{#1 \,\|\, #2}
\newcommand{\dtv}[2]{\mathrm{d}_{\mathrm{TV}}\rbra{#1,#2}}
\newcommand{\distB}{\mathrm{dist}_{\mathcal B}}
\newcommand{\distU}{\mathrm{dist}_{\mathcal U}}
\newcommand{\distC}{\mathrm{dist}_{\mathcal C}}
\newcommand{\mass}{\operatorname{mass}}

\newif\ifanon
\anonfalse

\title{Optimal Quantum Junta Testing}
\ifanon
\author{}
\else
\author{
  Jinge Bao\thanks{School of Informatics, University of Edinburgh, Edinburgh EH8 9AB, UK. \href{mailto:jinge.bao@ed.ac.uk}{jinge.bao@ed.ac.uk}}
  \and
  Minbo Gao\thanks{Institute of Software, Chinese Academy of Sciences, Beijing 100190, China. \href{mailto:gmb17@tsinghua.org.cn}{gmb17@tsinghua.org.cn}}
  \and
  Penghui Yao\thanks{State Key Laboratory for Novel Software Technology, Nanjing University, Nanjing 210023, China. }~\thanks{Hefei National Laboratory, Hefei 230088, China. \href{mailto:pyao@nju.edu.cn}{pyao@nju.edu.cn}}
}
\fi

\date{}

\newcommand{\logplus}{\log_{2,+}}
\newcommand{\SMALL}{\ensuremath{\textnormal{\textsc{SMALL}}}}
\newcommand{\LargeCase}{\ensuremath{\textnormal{\textsc{LARGE}}}}
\newcommand{\Accept}{\ensuremath{\textnormal{\textsc{Accept}}}}
\newcommand{\Reject}{\ensuremath{\textnormal{\textsc{Reject}}}}
\newcommand{\Overfull}{\ensuremath{\textnormal{\textsc{Overfull}}}}

\newcommand{\GapTest}{\ensuremath{\textnormal{\textsc{GapTest}}}}
\newcommand{\MinMarg}{\ensuremath{\textnormal{\textsc{MinMarg}}}}
\newcommand{\UnionSample}{\ensuremath{\textnormal{\textsc{UnionSample}}}}
\newcommand{\RankFind}{\ensuremath{\textnormal{\textsc{RankFind}}}}
\newcommand{\ScalarEstimate}{\ensuremath{\textnormal{\textsc{ScalarEstimate}}}}
\newcommand{\PartialRecover}{\ensuremath{\textnormal{\textsc{PartialRecover}}}}
\newcommand{\InfluenceSample}{\ensuremath{\textnormal{\textsc{Influence-Sample}}}}
\newcommand{\LocalRecover}{\ensuremath{\textnormal{\textsc{LocalRecover}}}}

\begin{document}

\maketitle

\begin{abstract}
We determine the quantum query complexity of testing
$k$-juntas for Boolean functions, unitaries, and quantum channels.
Let $n$ denote the number of input bits or qubits and
$\varepsilon$ the distance parameter.
For all sufficiently large $k$, $n \ge 2k$, and
$2^{-k/16} \le \varepsilon \le 1/16$,
we establish the following tight bounds.

\begin{itemize}
    \item \textbf{Junta functions:}
    The query complexity is $\Theta(\sqrt{k/\varepsilon})$, closing the gap left by
    \hyperlink{cite.ABRW16}{\textit{Ambainis, Belovs, Regev, and de Wolf
    (SODA 2016)}} and
    \hyperlink{cite.BKT20}{\textit{Bun, Kothari, and Thaler
    (STOC 2018, ToC 2020)}}.

    \item \textbf{Junta unitaries:}
    The query complexity is
    $\Theta(\sqrt{k}/\varepsilon)$, closing the gap left by
    \hyperlink{cite.CNY23}{\textit{Chen, Nadimpalli, and Yuen
    (SODA 2023)}}.
    With forward-only access, the query complexity is
    $\Theta(k/(\varepsilon^2 \log k))$.
    At constant $\varepsilon$, these bounds establish a separation between $\Theta(\sqrt{k})$ queries with forward and inverse access and $\Theta(k/\log k)$ queries with forward-only access.

    \item \textbf{Junta channels:}
    The query complexity is
    $\Theta(k/(\varepsilon^2 \log k))$, closing the gap left by
    \hyperlink{cite.BY25}{\textit{Bao and Yao
    (COLT 2023, TPAMI 2025)}}.
\end{itemize}
The upper bounds for channels and forward-only unitaries
are achieved by nonadaptive, ancilla-free testers
using product inputs and single-qubit measurements.

Our upper bounds follow from a reduction to
\emph{union size testing}, a generalization of
support size testing to distributions over subsets
of a ground set.
We develop optimal testers for this problem
both from independent samples and from a reversible sampler.
Our matching lower bounds use moment-matching constructions
and a phase-masked unitary ensemble.
For this ensemble, the average acceptance probability
of every forward-only quantum algorithm admits
an exact representation in terms of independent
classical samples.
\end{abstract}

\clearpage

\setcounter{tocdepth}{2}

\tableofcontents

\clearpage

\section{Introduction}\label{sec:introduction}

Property testing asks whether an object has a specified property or is
far from every object with that property~\cite{Gol17}.
The aim is to answer this question using few queries, without
reconstructing the object.
Quantum property testing studies how quantum algorithms test properties
of classical and quantum objects~\cite{BFNR08,MdW16}.
For quantum states and operations, testing can require far fewer copies or queries
than full reconstruction.
A central question is which properties admit tests whose complexity is
independent of the system size.

Junta testing provides a standard example. Boolean juntas are functions that depend on only a small subset of their input variables. Testing whether a function is a junta has been studied in property testing and learning theory under several access models~\cite{MOS03,PRS02,FKR+04,CG04,Bla08,Bla09,STW15,Sag18,Bsh19}.

Quantum juntas extend the notion of Boolean juntas. A junta unitary acts nontrivially only on a small set of qubits and as the identity on the remaining qubits. A junta channel has the same form, but its action on the selected
qubits may include noise.
The tester for junta unitaries in~\cite{Wan11} has query complexity
independent of the total number of qubits.
An improved tester combines influence estimation with quantum group
testing~\cite{CNY23}.
Testing and learning algorithms for junta channels are given in~\cite{BY25}.
Further work studies tolerant and nonadaptive quantum junta
testing~\cite{CLL24,ADEG25,BLY+25}.

These results establish that junta testing can be independent of the system
size, but leave open its optimal dependence on the junta size and
distance parameter. For unitaries, the answer may also depend on which
oracle operations are available.
Standard amplitude amplification and amplitude estimation use both a
unitary and its inverse~\cite{BHMT02}.
The junta tester of~\cite{CNY23} uses such inverse queries.
In the worst case, approximating the inverse of an unknown unitary to a
fixed small diamond-norm error requires a number of forward queries that
grows with the dimension~\cite{CYZ26}.
On some instances, no algorithm without inverse queries achieves the quadratic speedups of amplitude amplification and amplitude estimation~\cite{TW25}.
In the \emph{standard model}, a tester may apply the
unknown unitary and its inverse.
In the \emph{forward-only model}, a tester may apply the unknown unitary
but is given no inverse oracle.
Neither model supplies controlled versions of the unknown unitary or its inverse.
These distinctions raise the following question.
\begin{center}
    \emph{What is the optimal quantum query complexity of junta testing across different oracle models?}
\end{center}

In this work, we study the quantum query complexity of junta testing for Boolean functions, for unitaries in the standard and forward-only models, and for quantum channels. We introduce union size testing, a statistical problem that unifies these testing tasks. It asks whether random subsets are confined to a common small set of coordinates or fall outside every such set with probability at least a given threshold. We give optimal union size testers from independent samples and from a reversible sampler. These testers give upper bounds for all three junta classes. Partial support recovery further improves the Boolean bound at distances exponentially small in $k$. We also prove matching quantum lower bounds for junta testing. Together, these results determine the optimal query complexity in all four settings in the parameter regimes stated below.

\subsection{Main results}

Let $n$ be the number of input bits or qubits, $k$ the junta-size bound, and $\eps$ the distance parameter.
The tester must distinguish, with success probability at least $2/3$, a $k$-junta from an object whose distance from every $k$-junta of the same type is at least $\eps$.
Our lower bounds allow arbitrary ancillas and adaptive operations.

For Boolean functions, a $k$-junta depends on at most $k$ input bits.
We measure distance by uniform Hamming distance (\cref{def:boolean-distance}).
Since the evaluation oracle is self-inverse, the two access models coincide.

\begin{theorem}[Testing junta Boolean functions, informal version of \cref{thm:boolean-junta-upper-bound,thm:boolean-junta-lower-bound}]
\label{thm:main-boolean}
For integers $k\ge1$ and $n\ge2k$, and for $2^{-k/16}\le\eps\le1/8$,
testing whether an $n$-bit Boolean function is a
$k$-junta or is $\eps$-far from every $k$-junta has quantum query
complexity $\Theta(\sqrt{k/\eps})$.
\end{theorem}

The tester in~\cite{ABRW16} uses $O(\sqrt{k/\eps}\log k)$ queries.
Our upper bound removes a factor of $\log k$.
At constant distance, the previous lower bound of~\cite{BKT20} matches our bound up to logarithmic factors.
The formal theorems establish matching bounds with refined distance
dependence for all $k\ge1$, $n\ge2k$, and $0<\eps\le1/8$.

For unitaries, a $k$-junta acts nontrivially on at most $k$ qubits.
We measure distance by normalized Hilbert--Schmidt distance up to global phase (\cref{def:unitary-distance}).
We first consider the standard model.

\begin{theorem}[Testing junta unitaries in the standard model, informal version of \cref{thm:unitary-junta-upper-bound,thm:unitary-junta-lower-bound}]
\label{thm:main-unitary-inverse}
For integers $k\ge1$ and $n\ge2k$, and for $0<\eps\le1/4$,
testing whether an $n$-qubit
unitary is a $k$-junta or is $\eps$-far from every $k$-junta
has quantum query complexity $\Theta(\sqrt{k}/\eps)$ using both forward and inverse access.
The lower bound holds even if the tester may also apply controlled versions of the unitary and its inverse.
\end{theorem}

The tester in~\cite{CNY23} uses
$O\rbra{\sqrt{k}\log^{3/2}k/\eps}$ queries in this model.
Our upper bound removes a factor of $\log^{3/2}k$.
At constant distance, the previous lower bound of~\cite{CNY23} matches our bound up to logarithmic factors.
The formal upper bound holds for all $1\le k\le n$ and $0<\eps\le1$.

We next consider the forward-only model.

\begin{theorem}[Testing junta unitaries in the forward-only model, informal version of \cref{thm:unitary-junta-upper-bound-forward,thm:unitary-junta-lower-bound-forward}]
\label{thm:main-unitary}
For sufficiently large $k$, $n\ge2k$, and $2^{-k/16}\le\eps\le1/16$,
testing whether an $n$-qubit
unitary is a $k$-junta or is $\eps$-far from every $k$-junta
has quantum query complexity $\Theta(k/(\eps^2\log k))$
with forward-only access.
Moreover, a nonadaptive tester without ancillas achieves
this optimal query complexity.
\end{theorem}

The tester in~\cite{Wan11}, under the refined analysis in~\cite{MdW16},
uses $O\rbra{k/\eps^2}$
forward queries.
Our upper bound removes a factor of $\log k$.
At constant distance, our lower bound raises the previous $\widetilde\Omega(\sqrt{k})$ bound of~\cite{CNY23} to $\Omega(k/\log k)$.
The formal upper bound holds for all $1\le k\le n$ and $0<\eps\le1$, with $\log(ek)$ in place of $\log k$.
At constant distance, \cref{thm:main-unitary-inverse,thm:main-unitary} give a
separation between $\Theta(\sqrt{k})$ queries in the standard model and
$\Theta(k/\log k)$ queries in the forward-only model.
Inverse access thus improves the query complexity quadratically, up to a logarithmic factor.

We finally consider channels, which extend the forward-access setting
to operations that may include noise.
A $k$-junta channel also acts nontrivially on at most $k$ qubits.
We measure distance by normalized Hilbert--Schmidt distance between Choi matrices (\cref{def:channel-distance}).
Each query applies the unknown channel once, without access to its environment.

\begin{theorem}[Testing junta channels, informal version of \cref{thm:channel-junta-upper-bound,thm:channel-junta-lower-bound}]
\label{thm:main-channel}
For sufficiently large $k$, $n\ge2k$, and $2^{-k/16}\le\eps\le1/16$,
testing whether an $n$-qubit
channel is a $k$-junta or is $\eps$-far from every $k$-junta
has quantum query complexity $\Theta(k/(\eps^2\log k))$
with channel access.
Moreover, a nonadaptive tester without ancillas achieves
this optimal query complexity.
\end{theorem}

This bound equals the forward-only unitary bound of \cref{thm:main-unitary}.
The tester in~\cite{BY25} uses $O\rbra{k/\eps^2}$ channel queries.
Our upper bound removes a factor of $\log k$ with the same product input states and single-qubit measurements.
At constant distance, our lower bound raises the previous $\Omega(\sqrt{k})$ bound of~\cite{BY25} to $\Omega(k/\log k)$.
This lower bound already holds for unitary channels, with distance measured against all junta channels, including noisy ones.
The formal upper bound holds for all $1\le k\le n$ and $0<\eps\le1$, with $\log(ek)$ in place of $\log k$.

\Cref{tab:junta-comparison} compares the quantum query bounds at a
constant distance.\footnote{For Boolean junta testing and unitary junta
testing in both access models, we have derived the more explicit lower
bound $\Omega(\sqrt{k}/\log^2 k)$ with the assistance of AI tools. The idea is to make the dual-polynomial parameters
in~\cite{BKT20} explicit and express symbol-query acceptance probabilities
as polynomials in symbol counts by symmetrizing over input positions.
This preserves polynomial degree and avoids an extra logarithmic loss
from binary encoding. The reductions in~\cite{ABRW16,CNY23} then give the
junta bounds with only constant query overhead.}

\begin{table}[!t]
\centering
\small
\setlength{\tabcolsep}{4pt}
\renewcommand{\arraystretch}{1.1}
\begin{tabular}{@{}llccc@{}}
\toprule
\multirow{2}{*}{Problem} &
& \multicolumn{2}{c}{Prior work} & This work \\
\cmidrule(lr){3-4}\cmidrule(l){5-5}
& & Upper bound & Lower bound & Tight bound \\
\midrule
\multirow{2}{*}{Junta Boolean functions}
&
& $O(\sqrt{k}\log k)$ & $\widetilde{\Omega}(\sqrt{k})$ & $\Theta(\sqrt{k})$ \\
& & \cite{ABRW16} & \cite{BKT20}
& \cref{thm:boolean-junta-upper-bound,thm:boolean-junta-lower-bound} \\
\midrule
\multirow{4}{*}[-1.5pt]{Junta unitaries}
& \multirow{2}{*}{Standard}
& $O(\sqrt{k}\log^{3/2}k)$ & $\widetilde{\Omega}(\sqrt{k})$ & $\Theta(\sqrt{k})$ \\
& & \cite{CNY23} & \cite{CNY23}
& \cref{thm:unitary-junta-upper-bound,thm:unitary-junta-lower-bound} \\
\cmidrule(l){2-5}
& \multirow{2}{*}{Forward-only}
& $O(k)$ & $\widetilde{\Omega}(\sqrt{k})$ & $\Theta(k/\log k)$ \\
& & \cite{Wan11,MdW16} & \cite{CNY23}
& \cref{thm:unitary-junta-upper-bound-forward,thm:unitary-junta-lower-bound-forward} \\
\midrule
\multirow{2}{*}{Junta channels} &
& $O(k)$ & $\Omega(\sqrt{k})$ & $\Theta(k/\log k)$ \\
& & \cite{BY25} & \cite{BY25}
& \cref{thm:channel-junta-upper-bound,thm:channel-junta-lower-bound} \\
\bottomrule
\end{tabular}
\caption{Quantum query complexity of testing $k$-juntas at a fixed sufficiently small distance.}
\label{tab:junta-comparison}
\end{table}

Our junta upper bounds rest on union size testing (\cref{prob:set-testing}).

\begin{theorem}[Union size testing, informal version of \cref{thm:set-testing,cor:cover-lower-independent,thm:reversible-subset-upper,cor:cover-lower}]
\label{thm:main-union-size}
Fix an integer $k\ge1$, a finite ground set of size at least $2k$, and $0<\delta<1/32$.
Testing whether some set of at most $k$ elements contains the sampled subset with probability one or every such set fails to contain it with probability at least $\delta$ has quantum query complexity $\Theta(\sqrt{k/\delta})$ with a reversible sampler and, for sufficiently large $k$, sample complexity $\Theta(k/(\delta\log k))$ with independent samples.
\end{theorem}

Support size testing is the special case with singleton samples.
Here, being far in total variation from every distribution with small support is equivalent to every small set of symbols missing a large probability mass.
The size bound thus becomes a support size bound, and the threshold becomes a total-variation distance.

\begin{theorem}[Support size testing, informal version of \cref{cor:classical-support-upper-bound,thm:classical-support-lb}]
\label{thm:main-support-size}
Fix a sufficiently large integer $s$, a finite alphabet of size at least $2s$, and $0<\delta<1/32$.
Testing whether a distribution on this alphabet has support size at most $s$ or is $\delta$-far in total-variation distance from every such distribution has sample complexity $\Theta(s/(\delta\log s))$ with independent samples.
\end{theorem}

The upper bound in \cref{thm:main-support-size} is the independent-sample bound of \cref{thm:main-union-size} for singleton samples.
It improves the tester of~\cite{PH26} by a logarithmic factor when the distance is polynomially small in the support size bound, and matches it at constant distance.
The lower bound is stated in~\cite{PH26,Wang26} as a consequence of~\cite{VV11,WY19}.
Our self-contained proof adapts the moment-matching argument of~\cite{WY19}. Since it uses only singleton samples, it also gives the independent-sample lower bound in \cref{thm:main-union-size}.

\subsection{Technical overview}

The common problem behind our results is \emph{union size testing}.
Samples are subsets of a known finite ground set. Given a size bound
$k\ge1$ and a threshold $0<\delta\le1$, the task (\cref{prob:set-testing})
is to distinguish whether some set of at most $k$ coordinates contains
the sampled subset with probability one, or every such set fails to
contain it with probability at least $\delta$.
We refer to these as the containment and separation cases, respectively.
In the former case, the union of all subsets occurring with
positive probability has at most $k$ elements.
When every sample is a singleton, this is support size testing.
In general, the inclusion events of coordinates within a sample may be
arbitrarily dependent.

The complexity depends on how the distribution is accessed.
Independent samples provide independent draws from the distribution.
Reversible sampler access provides a unitary whose subset register has
that distribution when measured, together with the inverse unitary.
Each query is an ordinary call to the sampler or its inverse.
Controlled versions of these unitaries are not supplied.
\Cref{thm:main-union-size} gives the optimal bounds in both models.

We first reduce the junta problems to union size testing.
Independent samples yield the channel and forward-only unitary upper
bounds, whereas reversible sampling yields the standard-model unitary
bound and the basic Boolean bound. Partial support recovery then improves
the Boolean bound at distances exponentially small in $k$. To prove optimality, we construct
moment-matched hard families for standard-model unitaries and Boolean
functions. For forward-only queries, we construct a phase-masked ensemble
whose ensemble-averaged acceptance probabilities can be reproduced exactly
from independent classical samples.

\subsubsection{From junta testing to union size testing}

Let $f$ be an $n$-bit Boolean function and $U$ an
$n$-qubit unitary. Fourier sampling of the sign function
$(-1)^f$~\cite{AS07} and Pauli sampling of $U$~\cite{Wan11,MO10}
draw labels according to squared magnitudes of their respective expansion
coefficients. The resulting subset $X$ consists of the variables appearing
in the sampled Fourier parity or the qubits carrying nonidentity factors
in the sampled Pauli string. For a candidate junta support
$T$, the uncovered probability $\Pr\sbra{X\nsubseteq T}$ equals the influence
outside $T$. Consequently, a $k$-junta yields samples confined to some $k$-set.
If the input is at distance at least $\eps$ from every $k$-junta, then
every candidate $k$-set has uncovered probability at least a constant
multiple of $\eps$ for Boolean functions and of $\eps^2$ for unitaries.

For an $n$-qubit channel, prepare its normalized Choi state by applying
the channel to one half of each of $n$ Bell pairs. A Bell-basis
measurement produces Pauli labels. The channel influence outside a
candidate support $T$ is the probability that a label is not supported
within $T$. For a channel that is $\eps$-far from every $k$-junta,
this influence is at least $\eps^2$ outside every candidate $k$-set.
The tester does not prepare this Choi state. Instead, the \InfluenceSample\ procedure of~\cite{BY25} uses one channel
query on a product input, followed by single-qubit measurements, to
produce a subset $X$. Its uncovered probability lies
between two thirds of this influence and the influence itself.
A forward query to $U$ implements the same procedure for its induced
channel, whose influence equals that of $U$. Thus channel testing and
forward-only unitary testing reduce to union size testing from
independent samples with $\delta=\Theta(\eps^2)$.

With access to $U$ and its inverse $U^{\dagger}=U^{-1}$,
the Pauli sampler uses one query to $U$, and its inverse uses one query
to $U^{\dagger}$. This gives reversible sampling access, again with $\delta=\Theta(\eps^2)$.
The Boolean Fourier sampler is also reversible because its evaluation oracle is
self-inverse. For Boolean functions, the threshold is $\delta=\Theta(\eps)$,
as shown in \cref{subsec:influence}.
The thresholds of order $\eps$ and $\eps^2$ account for the different
distance dependence of the Boolean and unitary upper bounds.

\subsubsection{Union size testing from independent samples}

A naive union size tester takes unions of observed sets until more than
$k$ coordinates appear. It may need $\Omega(k/\delta)$ samples because
excess support can be spread over many rare coordinates. Our improvement
is to detect that support statistically, before observing more than $k$
distinct coordinates. The difficulty is that the marginal inclusion
probabilities may be arbitrarily small, and coordinates within a sample
need not be independent.

For sufficiently large $k$, we first combine $\Theta(1/\delta)$ original
observations into a union, amplifying the uncovered probability in
the separation case to a constant. From $\Theta(k/\log k)$ such amplified observations, we
form an initial union and reject if it has more than $k$ coordinates.
We then consider the parts of subsequent observations outside this union.
We replace each residual observation with more than $\ceil{\log k}$ coordinates by the empty
set, so one observation affects at most $\ceil{\log k}$ coordinate counts.
With probability at least $19/20$ over the initial observations, this
truncation preserves a constant uncovered probability in the separation case.
The remaining size bound is the number of coordinates still allowed
beyond the initial union. When it is at most $k/\log k$, further union
sampling suffices. Otherwise, we test the residual observations statistically.

We take a Poisson number of residual observations and count each
coordinate's occurrences. Each count has a Poisson distribution,
although different counts may be dependent. Adapting polynomial
estimators for support size estimation~\cite{WY19}, we use a Fej\'er
polynomial, built from Chebyshev polynomials, to assign a real weight
to each count. Its degree is logarithmic in the remaining size bound.
A count of zero receives weight zero, and each coordinate's expected
weight lies between zero and one. It is close to one for sufficiently
frequent coordinates. For rare coordinates, it is bounded below by a
constant multiple of the marginal times the remaining size bound.
In the containment case, the expected sum of weights is at
most the remaining size bound. In the separation case, the uncovered
probability forces a constant total marginal mass beyond the most
frequent coordinates allowed by this bound. The contribution of this
tail exceeds the possible deficit of the leading coordinates from
weight one. This gives a detectable expectation gap without any lower
bound on positive marginals. The logarithmic polynomial degree allows us to
obtain this gap with $O(k/\log k)$ residual observations.

Truncation also makes it possible to control concentration despite the
dependence within an observation. We treat each entire residual
observation as one Poisson point. Adding a point changes only
$O(\log k)$ occurrence counts. Bounds on the changes in the polynomial weights then let a
Poisson Poincar\'e inequality control the variance even though the counts
are dependent. The variance is small compared with
the squared expectation gap. If the proposed Poisson number of observations
exceeds a fixed constant multiple of its mean, the tester accepts.
This cap adds only a small constant error and bounds the sample count on
every execution. \Cref{thm:set-testing} shows that the tester uses
$O(k/(\delta\log(ek)))$ original samples. This gives the optimal
$O(k/(\eps^2\log(ek)))$ channel and forward-only unitary upper
bounds. Each sample is generated by the same procedure, independently
of earlier outcomes. The fixed sample bound therefore permits all
quantum queries and measurements to be completed before classical
post-processing. The resulting quantum testers are nonadaptive and use
product inputs, single-qubit measurements, and no ancillas.

\subsubsection{Union size testing from a reversible sampler}

Reversible sampling permits coherent tests of collections of coordinates.
Earlier group-testing approaches~\cite{ABRW16,CNY23} search over influence
thresholds, incurring a logarithmic loss when marginal probabilities
occur at many scales. We avoid searching through all scales by using a
certificate based on the aggregate reciprocal cost of a set of coordinates.
For a set of coordinates with positive marginal inclusion probabilities,
its reciprocal cost is the sum of the reciprocals of these probabilities.
For an integer $1\le d\le k$, suppose some set of $k+d$ coordinates has
reciprocal cost $O(kd/\delta)$. Given a known upper bound of this order,
our reciprocal gap test distinguishes distributions in the containment
case from those admitting such a certificate. It adapts the group-testing
certificate of~\cite{ABRW16} and uses $O(\sqrt{k/\delta})$ queries to the
sampler and its inverse. The query bound depends on the square root of the
reciprocal-cost bound per excess coordinate. The certificate set itself
may remain unknown.

Selecting $d$ is the main difficulty. We first test for too many
coordinates with large marginals, and for uncovered probability dispersed among
very rare coordinates. In the remaining separation case, the analysis
isolates a core of $O(k)$ coordinates. The algorithm need not identify
this core.
The marginals outside the $k$ most frequent core coordinates have total
mass $\Omega(\delta)$. Our good-scale lemma turns
this tail-mass condition into an interval of marginal estimates.
Any estimate in that interval determines an excess size $d$ admitting a
certificate of the required reciprocal cost.

The interval is unknown. Enumerating candidate thresholds would
reintroduce the logarithmic factor, so we use union sampling instead.
For the remaining distributions in the separation case, with constant
probability, union sampling either produces an oversized union or finds a
nonempty observed union whose reciprocal cost is $O(k/\delta)$ and whose
minimum marginal falls in the good-scale interval.
An oversized union certifies that more than $k$ coordinates have been observed.
In the latter case, our minimum-marginal estimator uses the known reciprocal-cost
bound to approximate the smallest marginal in the nonempty observed union within a constant
factor using $O(\sqrt{k/\delta})$ queries, even when coordinate-inclusion
events overlap. The estimator adapts the minimum-finding analysis
of~\cite{DH96}. Each search for a smaller marginal uses a rank-sensitive
finder. Its cost is governed by the reciprocal cost of the union rather
than by the current threshold, so the expected costs of successive phases
form a convergent series. The estimate identifies a suitable $d$, and a
single reciprocal gap test then suffices. The resulting tester uses
$O(\sqrt{k/\delta})$ queries (\cref{thm:reversible-subset-upper}). With reversible Pauli sampling and
$\delta=\Theta(\eps^2)$, this yields the
$O(\sqrt{k}/\eps)$ unitary tester.
It uses $U$ and $U^{\dagger}$ but not their controlled versions.
In the gap test, controlled walk steps apply the sampler and its inverse
unconditionally and control only known reflections.

\subsubsection{Partial support recovery for Boolean juntas}

For Boolean functions, the Fourier sampler is reversible at constant query cost, so the
generic union size tester gives $O(\sqrt{k/\eps})$ queries.
This bound is already optimal in the regime of \cref{thm:main-boolean}.
When $\eps$ is exponentially small in $k$, Boolean functions admit a further improvement,
because their Fourier supports are constrained by the combinatorial structure
of Boolean level sets. We exploit this through \emph{partial support
recovery} before testing the residual distribution.

For a promised $k$-junta, our recovery procedure returns a set $J$ containing only relevant
coordinates. Let $X$ be a Fourier-support sample.
With probability at least $19/20$ over the recovery procedure, either
the residual Fourier mass $\Pr\sbra{X\nsubseteq J}$ is less than
$\eps/32$, or only $O(\log(2+\eps2^k))$ relevant coordinates remain
outside $J$.

The improvement comes from recovering several coordinates at once.
Consider a uniformly sampled pair of inputs, conditioned on agreeing on $J$
and having different function values.
A uniformly chosen endpoint lies with constant probability in the
minority level set among inputs sharing their values on $J$.
Call a coordinate outside $J$ sensitive at an input if flipping it changes the function value.
The edge-isoperimetric inequality~\cite{Har64,Ber67,Har76}
forces many sensitive coordinates on average within that minority set.
Random restrictions followed by Fourier sampling recover these coordinates
simultaneously.

To bound the recovery cost, we group rounds by residual Fourier mass
and bound the expected number of rounds at each scale. Summing the
corresponding search costs bounds the expected number of queries needed
to reach either recovery outcome. A fixed query cap then gives the
stated success guarantee.

If recovery finds more than $k$ relevant coordinates, the tester rejects.
Otherwise, quantum amplitude estimation tests for small residual Fourier
mass, while a reversible union size test uses the reduced residual size
bound, capped at $k-\abs{J}$. We accept if either test accepts.
For a far input, every extension of $J$ within this allowance has size
at most $k$, so both tests reject with bounded error.
The reduced size bound replaces $k$ in the residual testing cost by
$O(\log(2+\eps2^k))$. In the small-distance regime where recovery is
used, its cost is absorbed by the resulting testing bound.

Partial support recovery also gives exact support recovery of a promised $k$-junta using $O(2^{k/2})$ queries (\cref{cor:exact-support}).
Since every $n$-bit function is an $n$-junta, exact recovery with junta-size bound $n$ gives a tester that uses $O(2^{n/2})$ queries and accepts exactly when the recovered support has at most $k$ elements.
Combining this finite-domain branch with the generic tester and the
recovery branch gives the query bound
$O\rbra{\min\cbra{2^{n/2},\sqrt{\log(2+\eps2^k)/\eps}}}$
of \cref{thm:boolean-junta-upper-bound}.

\subsubsection{Lower bounds via moment matching}

To match the upper bounds, we construct sparse and dense priors that
cannot be distinguished with too few samples or queries. The common
idea is to match moments of hidden parameters, using the form of the
algorithm's acceptance probability to control indistinguishability.
The relevant moments differ across the three settings.

For independent samples, low-order moment matching controls the
distributions of occurrence counts. We adapt the approach of~\cite{WY19}
to prove a matching $\Omega(k/(\delta\log k))$ lower bound for union size
testing already with singleton samples (\cref{cor:cover-lower-independent}).

In the standard unitary model, we apply the polynomial method~\cite{BBC+01} to the product rotations
$U_{\theta}=\bigotimes_{j=1}^{2k}e^{i\theta_j Z}$ on $2k$ qubits,
where the $\theta_j$ are real angles and $Z$ is the Pauli $Z$ matrix.
For a candidate support $T$ among these $2k$ qubits, the squared normalized
Hilbert--Schmidt distance, up to global phase, from $U_\theta$ to the
unitaries supported on $T$ equals
$1-\prod_{j\notin T}\abs{\cos\theta_j}$.
This relates the number and magnitudes of nonzero angles to junta
membership and distance.

The acceptance probability of any $q$-query quantum algorithm, even
with inverse queries, ancillas, and adaptivity, is a trigonometric
polynomial in the $\theta_j$ whose frequency in each coordinate is at
most $2q$ in magnitude. We construct two
one-dimensional angle distributions with equal expectations for all
Fourier modes of frequency at most $2q$ in magnitude. For each
distribution, we draw the coordinate angles independently. The sparse
distribution puts most of its mass at zero, so its product unitary is a
$k$-junta with probability at least $7/8$. In the dense distribution, every angle
has magnitude between a suitable constant multiple of $\eps/\sqrt{k}$
and $\pi/4$, so every product unitary is $\eps$-far from $k$-juntas.
Chebyshev-based finite moment matching supplies
these distributions when $q$ is at most a sufficiently small constant
multiple of $\sqrt{k}/\eps$. By independence and the coordinatewise
frequency bound, every such algorithm has the same mean acceptance
probability under the two product distributions.
A tester with success probability at least $2/3$ would accept more often
on average under the sparse prior than under the dense prior,
contradicting this equality. This gives the matching
$\Omega(\sqrt{k}/\eps)$ lower bound of \cref{thm:unitary-junta-lower-bound} in the standard unitary model.

Conjugating the product rotations by Hadamard gates gives reversible
subset samplers that include each coordinate $j$ independently with
probability $\sin^2\theta_j$. Choosing the angle scale of order
$\sqrt{\delta/k}$ gives the matching $\Omega(\sqrt{k/\delta})$
reversible union size lower bound of \cref{cor:cover-lower}.

For Boolean juntas, we use a parity-address construction for the matching lower bound,
adapting the addressing functions of~\cite{AS07,ABRW16}.
Address bits are forced to be relevant, while the data dependence is
controlled by a random Boolean matrix with independent latent
column probabilities. Here, acceptance is a low-degree polynomial in
the truth-table entries. Averaging over the matrix expresses it in
terms of moments of the latent column probabilities. Matching these moments makes
the sparse and dense ensembles indistinguishable to low-query algorithms.
A unique-search reduction covers the remaining regimes of small
distance and of small $n$. Together, these arguments prove \cref{thm:boolean-junta-lower-bound}, which matches the refined upper bound
whenever $n\ge2k$ and $\eps\le1/8$.
The same search reduction proves a matching
$\Omega(2^{k/2})$ lower bound for exact support recovery (\cref{cor:exact-support-lower}).

\subsubsection{Forward-only lower bounds via classicalization}

The classical support size lower bound need not constrain an arbitrary quantum
algorithm making coherent forward queries. Our reduction must establish
two facts. The hard ensemble must map the support size containment and
separation cases to junta membership and distance, respectively. Every
forward-query algorithm must also admit a classical decision rule,
independent of the unknown distribution, that uses independent samples
to reproduce its ensemble-averaged acceptance probability exactly.
We construct a phase-masked addressing ensemble to meet these requirements.
Let $s$ be the classical support size bound. We use an address register
and one data qubit for each symbol of the known alphabet.
For each computational-basis address, we independently draw a label
from the unknown distribution. We also draw independent uniform phases
on the complex unit circle, independently of all labels.
The oracle multiplies the address state by its phase
and applies Pauli $Z$ to the data qubit selected by the label.
These choices are sampled once and kept fixed throughout the computation.
The junta-size bound $k$ is the number of address qubits plus $s$.
If the distribution has support size at most $s$, every resulting
unitary is a $k$-junta.

We choose $\delta=\Theta(\eps^2)$ and an address register
large enough for the required concentration bounds. The stated range
of $\eps$ allows this choice while keeping $s=\Theta(k)$.
If every set of at most $s$ symbols misses probability mass at least
$\delta$, concentration forces nontrivial influence outside every
small set of data qubits. The random phases separately force influence
on omitted address qubits. With probability at least $19/20$, the
resulting unitary is $\eps$-far from every $k$-junta unitary, and its
induced channel is $\eps$-far from every $k$-junta channel, including
nonunitary ones.

The second fact is our exact classicalization theorem, \cref{thm:classicalization}.
For every forward-only quantum algorithm that is independent of the
underlying distribution and makes at most $q$ queries, there is a
classical randomized decision rule that uses $q$ independent samples
and exactly reproduces the algorithm's acceptance probability averaged
over the hard ensemble. The rule depends on the algorithm and the known
ensemble parameters, and is independent of the unknown distribution.

To prove this representation, we adapt the phase-averaged path expansion
in~\cite{TW25}, based on compressed oracles~\cite{Zha19}
(see also~\cite{FLMNW26}). We expand the quantum computation into query
paths and group them by the number of visits to each address, calling
each group a sector. Averaging over the independent phases removes
interference between different visit histograms. Every surviving sector
depends on the labels of at most $q$ addresses, hence on at most $q$
independent draws from the underlying distribution. This alone is
insufficient, because the sectors involve different address sets and do
not immediately define a normalized classical acceptance rule.
We therefore symmetrize over permutations of the $q$ sample positions.
For each sample tuple, the accepting sector weight lies between zero and
the total sector weight. The latter is permutation invariant and has
expectation one for every underlying distribution. A homogeneous-polynomial
identity then forces it to equal one for every tuple. The accepting
weight thus defines the required classical randomized decision rule.

The classical $\Omega(s/(\delta\log s))$ support size lower bound therefore
yields $\Omega(k/(\eps^2\log k))$ forward-only unitary queries in the
stated parameter range.
Each channel query on these unitary instances is simulated by one forward
query, giving the same lower bound for channel testing.
\Cref{thm:unitary-junta-lower-bound-forward,thm:channel-junta-lower-bound} state these two bounds formally.
The argument requires forward-only access.
With inverse queries, positive and negative phase exponents
can cancel, so a fixed signed-visit sector may depend on more than $q$
distinct addresses. The classicalization argument then no longer
applies.

\subsection{Related work}\label{subsec:related-work}

\paragraph{Quantum junta testing.}
Pauli analysis and influence for quantum Boolean functions~\cite{MO10}
underlie testers for junta unitaries~\cite{Wan11,CNY23} and
junta channels~\cite{BY25}.
Our forward-only and channel testers use the product-state
\InfluenceSample\ procedure of~\cite{BY25} with a new classical
decision rule. Our reversible union size tester uses the group-testing
certificate of~\cite{ABRW16}, with the additional marginal-scale
selection and reciprocal-cost analysis proved here.
Our channel lower bound also uses a Choi-matrix projection bound from
agnostic process tomography~\cite{WLKD25}.
Further work studies junta approximation~\cite{RWZ24},
junta states~\cite{BEG24}, and experimental influence sampling~\cite{ZBY+25}.

\paragraph{Boolean junta testing.}
Dimension-independent testing is known for dictatorships~\cite{PRS02}
and general Boolean juntas~\cite{FKR+04}.
For Boolean juntas, adaptive~\cite{Bla09,CG04,Sag18} and
nonadaptive~\cite{Bla08,CSTWX17} query complexities can differ~\cite{STW15}.
Quantum algorithms and lower bounds for Boolean junta testing appear
in~\cite{AS07,ABRW16,BKT20}.
Our bounds close the logarithmic gap at constant distance and determine
the refined small-distance dependence
(see~\cref{thm:boolean-junta-upper-bound,thm:boolean-junta-lower-bound}).
Other work considers learning from uniform examples~\cite{MOS03},
quantum identification of juntas with a prescribed symmetric
function~\cite{Bel15}, distribution-free testing with
classical~\cite{LCSSX18,Bsh19,Zhang19DF} or quantum
queries~\cite{Bel19}, and relative-error testing~\cite{CPPS25}.
Junta distributions have been studied under sample~\cite{ABR16,BEG24}
and subcube-conditioning access~\cite{CJLW21}.
Learning them relates to learning parities with noise~\cite{Ber25}.

\paragraph{Tolerant junta testing.}
Tolerant testing distinguishes inputs close to the property from inputs
sufficiently far from it~\cite{PRR06}.
For Boolean juntas, the size bound may differ between these two
conditions~\cite{BCELR19}.
For a common size bound, results address
nonadaptive complexity~\cite{LW19,PRW22,CDLNS24,NP24},
dimension-independent testing~\cite{DMN19,ITW21,CPS26},
adaptive lower bounds~\cite{CP23}, quantum algorithms~\cite{BLY+25},
and quantum advantages~\cite{TY26}.
Tolerant testing has also been studied for junta
unitaries~\cite{CLL24,ADEG25,BLY+25}.
In contrast, our main testing promises distinguish exact junta membership
from being sufficiently far from the junta class.

\paragraph{Distributional testing.}
Quantum algorithms test classical distributions under several access
models~\cite{BHH11,GL20,Bel19Distributions}.
Moment methods give lower bounds for support size
estimation~\cite{RRSS09}.
For support size estimation, sample-optimal estimators~\cite{VV17,WY19} and
quantum algorithms~\cite{LiWu19Entropy} assume a lower bound on nonzero
probabilities.
For the total-variation promise used here, the tester of~\cite{PH26}
loses a logarithmic factor, which independent and simultaneous
work~\cite{Wang26} removes on countable domains.
\Cref{thm:main-support-size} attains this rate on a known finite
alphabet, and \cref{thm:main-union-size} extends it to subset samples.
Our reversible tester also tests support size in the purified access
model of~\cite{GL20}, but its matching lower bound uses nonsingleton
samples.
Copy access to subset states can leave support size testing
hard~\cite{JMSW24}.

\paragraph{Quantum query lower bounds.}
Our lower bounds for unitary testing in the standard model and for Boolean
junta testing use the polynomial method~\cite{BBC+01} with finite
moment-matched priors. Related polynomial and invariant-theoretic lower
bounds for unitary property testing appear in~\cite{SY23}.
Our forward-query lower bound adapts path and histogram
expansions~\cite{TW25}, which are related to compressed oracles~\cite{Zha19} and
multiset representations of random diagonal queries~\cite{FLMNW26}.
It gives an exact representation of ensemble-averaged acceptance probabilities
using independent classical samples.
Sample-to-query lifting relates copy and query complexities of
quantum-state testing~\cite{WZ25}, and other work develops approximate
simulations across oracle models~\cite{Zha12,TWZ25,GZ25} and control
removal~\cite{TW25control}.

\subsection{Organization}\label{subsec:organization}
\Cref{sec:preliminaries} fixes notation and collects Fourier expansions, junta classes, influence, access models, and quantum subroutines.
\Cref{sec:subset-cover} gives optimal union size testers with matching lower bounds, first from independent samples and then from a reversible sampler. For independent samples, the singleton case gives support size testing.
\Cref{sec:boolean} proves the Boolean upper bound by combining the reversible tester with partial support recovery, and the lower bound by a parity-address construction and a search reduction.
\Cref{sec:unitary} proves the unitary bounds in the standard and forward-only models.
\Cref{sec:channel} treats channels with the same influence sampling and hard ensemble as the forward-only model.

\section{Preliminaries}\label{sec:preliminaries}

\subsection{Notation}\label{subsec:notation}

Throughout, $n\geq1$ denotes the number of input bits or qubits.
For a positive integer $m$, write $\sbra{m}=\cbra{1,\ldots,m}$.
For a finite set $S$, its cardinality is denoted by $\abs{S}$, and $2^S$ denotes the set of all subsets of $S$.

For a vector $x=\rbra{x_1,\ldots,x_m}$, write
$\supp\rbra{x}=\cbra{i\in\sbra{m}\colon x_i\neq0}$.
For $x\in\C^m$, write
$\norm{x}_1=\sum_{i=1}^m\abs{x_i}$.
For $z\in\C$, write $\overline z$ for its complex conjugate.
For $T\subseteq\sbra{m}$, let $x_T$ denote the restriction of $x$
to coordinates in $T$, ordered by their indices.

For a finite set $\calX$, let $\Delta\rbra{\calX}$
denote the set of probability distributions on $\calX$.
For $p\in\Delta\rbra{\calX}$ and an integer $t\geq0$, write $p^{\ot t}\in\Delta\rbra{\calX^t}$ for the law of $t$ independent draws from $p$.
We use $\Pr$, $\Ex$, and $\operatorname{Var}$ for probability, expectation, and variance, respectively.
We write $\mathrm{Bern}$, $\operatorname{Bin}$, $\operatorname{Poi}$, $\operatorname{Multinomial}$, and $\operatorname{Unif}$ for the Bernoulli, binomial, Poisson, multinomial, and uniform distributions.
For a positive integer $M$, $p\in\Delta\rbra{\sbra{M}}$, and $S\subseteq\sbra{M}$, write
$p\rbra{S}=\sum_{j\in S}p\rbra{j}$ and
$\supp\rbra{p}=\cbra{j\in\sbra{M}\colon p\rbra{j}>0}$.
For distributions $p,q$ on a countable set $\calX$, the total-variation distance is
\begin{equation*}
  \dtv{p}{q}=\frac12\sum_{x\in\calX}\abs*{p\rbra*{x}-q\rbra*{x}}=\sup_{S\subseteq\calX}\abs*{p\rbra*{S}-q\rbra*{S}}.
\end{equation*}
It satisfies the triangle inequality and does not increase when the same randomized map is applied to both distributions.

We write $\log$ for the base-two logarithm and $\ln$ for the natural logarithm.
We also write $\logplus x=\max\cbra{0,\log x}$.
For vectors, an unsubscripted norm is the Euclidean norm.
The notation $\widetilde O$ suppresses polylogarithmic factors in the
stated parameters.
Unless stated otherwise, dot products of binary vectors are taken modulo two.

We use the following standard probability bounds.

\begin{fact}[Markov's inequality]
\label{fact:markov}
For a nonnegative random variable $W$ and every $t>0$, 
\begin{equation*}
\Pr\sbra{W\geq t}\leq\Ex\sbra{W}/t.
\end{equation*}
\end{fact}

\begin{fact}[Chebyshev's inequality]
\label{fact:chebyshev}
For a random variable $W$ with finite variance and every $t>0$, \begin{equation*}\Pr\sbra{\abs{W-\Ex\sbra{W}}\geq t}\leq\operatorname{Var}\sbra{W}/t^2.
\end{equation*}
\end{fact}

\begin{fact}[Chernoff bound]
\label{fact:chernoff}
Let $m\geq1$ be an integer and let $X_1,\ldots,X_m$ be independent
Bernoulli random variables. Set $W=\sum_{j=1}^m X_j$ and
$\mu=\Ex\sbra{W}$. Then, for every $0<\delta<1$,
\[
 \Pr\sbra*{W\leq(1-\delta)\mu}\leq e^{-\delta^2\mu/2}.
\]
\end{fact}

\begin{fact}[Hoeffding's inequality {\cite[Theorem~2]{Hoe63}}]
\label{fact:hoeffding}
Let $m\geq1$ be an integer and let $X_1,\ldots,X_m$ be independent
random variables taking values in $[0,1]$.
For every $t>0$,
\[
 \Pr\sbra*{\abs*{\sum_{j=1}^m\rbra*{X_j-\Ex\sbra{X_j}}}\geq t}\leq2e^{-2t^2/m}.
\]
\end{fact}

\subsection{Fourier analysis}\label{subsec:fourier-analysis}

\subsubsection{Boolean functions}

For complex-valued functions on $\cbra{0,1}^n$, use the uniform inner
product $\langle f,g\rangle=\Ex_X\sbra{\overline{f(X)}g(X)}$, where $X$ is uniform
on $\cbra{0,1}^n$.
For $S\subseteq\sbra{n}$, define the parity character
$\chi_S(x)=(-1)^{\sum_{i\in S}x_i}$.

\begin{fact}[Fourier expansion of Boolean functions
  {\cite[Proposition~8.55]{OD14}}]
\label{fact:boolean-fourier}
The parity characters $\rbra{\chi_S}_{S\subseteq\sbra{n}}$ form an
orthonormal basis for the complex-valued functions on $\cbra{0,1}^n$.
Every $f:\cbra{0,1}^n\to\C$ satisfies
\[
 f=\sum_{S\subseteq\sbra{n}}\wh f\rbra{S}\chi_S,
 \qquad
 \wh f\rbra{S}=\langle\chi_S,f\rangle.
\]
Moreover, for all $f,g:\cbra{0,1}^n\to\C$,
\[
 \langle f,g\rangle=\sum_{S\subseteq\sbra{n}}\overline{\wh f\rbra{S}}\,\wh g\rbra{S}.
\]
In particular, $\sum_{S\subseteq\sbra{n}}\abs{\wh f\rbra{S}}^2=\Ex_X\sbra{\abs{f(X)}^2}$, where $X$ is uniform on $\cbra{0,1}^n$.
\end{fact}

For Boolean $f:\cbra{0,1}^n\to\cbra{0,1}$, write $(-1)^f$ for the sign function
$x\mapsto(-1)^{f(x)}$.
Its squared Fourier coefficients sum to one by Parseval's identity (\cref{fact:boolean-fourier}).
A \emph{Fourier sample} of $(-1)^f$ is a set $S\subseteq\sbra{n}$ drawn with probability $\abs{\wh{(-1)^f}\rbra{S}}^2$.

For $i\in\sbra{n}$, write $e_i$ for the $i$th coordinate vector.
A coordinate $i$ is relevant to $f$ if $f(x)\ne f(x\oplus e_i)$ for
some $x$. The support $\supp(f)$ is the set of relevant coordinates
\cite[Definition~2.18]{OD14}.

\subsubsection{Operators}\label{subsec:operators}

An $n$-qubit system has Hilbert space $\calH=\C^N$, $N=2^n$.
Write $L(\calK)$ for linear operators on a finite-dimensional space
$\calK$, $\mathrm U(m)$ for the unitaries on $\C^m$ (in particular, $\mathrm U(N)$ acts on $\calH$), $I_N$ for the
identity operator on $\calH$, and $I$ for an identity operator or identity channel whose space is clear from context. Register subscripts
specify where an operator or channel acts.
We use $A^{\dagger}$ for the adjoint, and the unnormalized norms
\[
 \norm A_2=\sqrt{\Tr\rbra{A^{\dagger}A}},\qquad
 \norm A_1=\Tr\rbra{\sqrt{A^{\dagger}A}}.
\]
The operator norm $\norm A_\infty$ is the largest singular value of $A$.
A density operator satisfies $\rho\succeq0$ and $\Tr\rbra{\rho}=1$.

Let
\begin{equation*}
  \sigma_0=I_2,\qquad \sigma_1=X,\qquad
  \sigma_2=Y,\qquad \sigma_3=Z
\end{equation*}
be the single-qubit Pauli matrices. Write $\Z_4=\cbra{0,1,2,3}$.
For a label $x=\rbra{x_1,\ldots,x_n}\in\Z_4^n$, define the corresponding Pauli string
by
\begin{equation*}
  \sigma_x=\sigma_{x_1}\ot\cdots\ot\sigma_{x_n}.
\end{equation*}
We use the unnormalized inner product $\langle A,B\rangle=\Tr\rbra{A^{\dagger}B}$ on $L\rbra{\calH}$.

\begin{fact}[Fourier expansion of operators {\cite[Proposition~10]{MO10}}]
\label{fact:pauli-expansion}
The Pauli strings $\rbra{\sigma_x}_{x\in\Z_4^n}$ form an orthogonal basis of $L\rbra{\calH}$
with $\langle\sigma_x,\sigma_y\rangle=N$ if $x=y$ and $\langle\sigma_x,\sigma_y\rangle=0$ otherwise.
Every $A\in L\rbra{\calH}$ satisfies
\[
 A=\sum_{x\in\Z_4^n}\wh A\rbra{x}\sigma_x,
 \qquad
 \wh A\rbra{x}=\frac1N\langle\sigma_x,A\rangle.
\]
Moreover, for all $A,B\in L\rbra{\calH}$,
\[
 \frac1N\langle A,B\rangle=\sum_{x\in\Z_4^n}\overline{\wh A\rbra{x}}\,\wh B\rbra{x}.
\]
In particular, $\sum_{x\in\Z_4^n}\abs{\wh A\rbra{x}}^2=\norm{A}_2^2/N$.
\end{fact}

We use $\Z_4^n$ only as labels. For $x\in\Z_4^n$, $\supp(x)$ is the set of positions of nonidentity factors,
and $x_T\ne0$ means $\supp(x)\cap T\ne\varnothing$.
Fix a computational basis
$\cbra{\ket{a}\colon a\in\sbra{N}}$ of $\calH$.
For $A\in L\rbra{\calH}$, define its vectorization by
\begin{equation}
\label{eq:vectorization}
  \operatorname{vec}\rbra*{A}
  =\rbra*{A\ot I_N}
    \sum_{a=1}^N\ket{a}\ot\ket{a}.
\end{equation}
Then
\begin{equation}
\label{eq:vectorization-inner-product}
  \operatorname{vec}\rbra*{A}^{\dagger}\operatorname{vec}\rbra*{B}
  =\Tr\rbra*{A^{\dagger}B}
  =\langle A,B\rangle.
\end{equation}
For a unitary $U\in\mathrm U\rbra{N}$, define
\begin{equation*}
  \ket{v\rbra*{U}}
  =\frac1{\sqrt N}\operatorname{vec}\rbra*{U}.
\end{equation*}
This is a unit vector, with
\begin{equation*}
  \ket{v\rbra*{I_N}}
  =\frac1{\sqrt N}\sum_{a=1}^N\ket{a}\ot\ket{a}.
\end{equation*}

By \cref{eq:vectorization-inner-product} and the Pauli expansion (\cref{fact:pauli-expansion}),
the Pauli vectors
$\cbra{\ket{v\rbra{\sigma_x}}\colon x\in\Z_4^n}$
form an orthonormal basis of $\calH\ot\calH$.
Applying vectorization to the Pauli expansion of $U$ (\cref{fact:pauli-expansion}) gives
\begin{equation*}
  \ket{v\rbra*{U}}
  =\sum_{x\in\Z_4^n}\wh U\rbra*{x}
    \ket{v\rbra*{\sigma_x}}.
\end{equation*}
Since $\norm{U}_2^2=N$, the Pauli expansion (\cref{fact:pauli-expansion}) applied to $U$ gives
\begin{equation}
\label{eq:unitary-pauli-mass}
  \sum_{x\in\Z_4^n}\abs{\wh U\rbra{x}}^2=1.
\end{equation}
These probabilities on Pauli labels induce a subset distribution by
taking supports.

\subsubsection{Superoperators}

A superoperator on $L\rbra{\calH}$ is a linear map
from $L\rbra{\calH}$ to itself.
In the Choi representation below, the first tensor factor is the
output space and the second is a reference copy of $\calH$.

\begin{definition}[Choi matrix]
\label{def:choi-representation}
For a superoperator $\Phi$ on $L\rbra{\calH}$, define its Choi matrix by
\begin{equation*}
  J\rbra*{\Phi}=
  \sum_{a,b\in\sbra{N}}
    \Phi\rbra*{\ketbra{a}{b}}
    \ot\ketbra{a}{b}.
\end{equation*}
\end{definition}

\begin{definition}[Pauli coefficients of a superoperator]
\label{def:channel-fourier-coefficients}
For a superoperator $\Phi$ on $L\rbra{\calH}$ and
$x,y\in\Z_4^n$, define
\begin{equation*}
  \wh\Phi\rbra*{x,y}=
  \frac1N
  \bra{v\rbra*{\sigma_x}}
  J\rbra*{\Phi}\ket{v\rbra*{\sigma_y}}.
\end{equation*}
We write $\wh\Phi$ for the matrix with these entries,
using the same fixed order on $\Z_4^n$ for its rows and columns.
\end{definition}

\begin{proposition}[Fourier expansion of superoperators {\cite[Proposition~6]{BY25}}]
\label{prop:channel-fourier-representation}
For every superoperator $\Phi$ on $L\rbra{\calH}$ and
every $A\in L\rbra{\calH}$,
\begin{equation*}
  \Phi\rbra*{A}=\sum_{x,y\in\Z_4^n}\wh\Phi\rbra*{x,y}\sigma_x A\sigma_y.
\end{equation*}
\end{proposition}

\begin{proof}
In the Pauli-vector basis,
\begin{equation*}
  J\rbra*{\Phi}
  =N\sum_{x,y\in\Z_4^n}\wh\Phi\rbra*{x,y}
    \ketbra{v\rbra*{\sigma_x}}{v\rbra*{\sigma_y}}.
\end{equation*}
The map $A\mapsto\sigma_x A\sigma_y$ has Choi matrix
$\operatorname{vec}\rbra{\sigma_x}\operatorname{vec}\rbra{\sigma_y}^{\dagger}
=N\ketbra{v\rbra{\sigma_x}}{v\rbra{\sigma_y}}$.
The Choi matrix determines the map, so comparing Choi matrices completes the proof.
\end{proof}

\begin{proposition}[Plancherel identity for superoperators]
\label{prop:superoperator-plancherel}
For all superoperators $\Phi,\Psi$ on $L\rbra{\calH}$,
\[
 \frac1{N^2}\langle J\rbra*{\Phi},J\rbra*{\Psi}\rangle=\sum_{x,y\in\Z_4^n}\overline{\wh\Phi\rbra{x,y}}\,\wh\Psi\rbra{x,y}.
\]
\end{proposition}

\begin{proof}
By \cref{def:channel-fourier-coefficients}, $\wh\Phi$ and $\wh\Psi$ are the matrices of $J\rbra{\Phi}/N$ and $J\rbra{\Psi}/N$ in the orthonormal Pauli-vector basis of $\calH\ot\calH$.
The Hilbert--Schmidt inner product of two operators equals the sum of the products of the conjugated entries of the first matrix with the corresponding entries of the second in any orthonormal basis, which gives the claim.
\end{proof}

A quantum channel is a completely positive trace-preserving
(CPTP) superoperator.
For a finite-dimensional space $\calK$, we write $C\rbra{\calK}$ for the set of CPTP linear maps from $L\rbra{\calK}$ to itself, so that $C\rbra{\calH}$ is the set of such channels.

\begin{fact}[Kraus representation {\cite[Corollary~2.27]{Wat18}}]
\label{fact:kraus-representation}
A superoperator $\Phi$ on $L\rbra{\calH}$ is a quantum channel if and
only if there is a finite family of operators $K_j\in L\rbra{\calH}$
such that, for every $A\in L\rbra{\calH}$,
\begin{equation*}
  \Phi\rbra*{A}=\sum_jK_jAK_j^{\dagger},
  \qquad \sum_jK_j^{\dagger}K_j=I_N.
\end{equation*}
\end{fact}

For $U\in\mathrm U\rbra{N}$, the single Kraus operator $U$ defines
the unitary channel
\begin{equation*}
  \Phi_U\rbra*{A}=UAU^{\dagger}.
\end{equation*}

For a channel $\Phi\in C\rbra{\calH}$, its Choi state is
\begin{equation}
\label{eq:choi-state}
  v\rbra*{\Phi}=\frac{J\rbra*{\Phi}}{N}
  =\rbra*{\Phi\ot I}
    \rbra*{\ketbra{v\rbra*{I_N}}{v\rbra*{I_N}}}.
\end{equation}
Complete positivity and trace preservation make $v(\Phi)$ a density
operator.
By the Kraus representation (\cref{fact:kraus-representation}), $\Phi$ has Kraus operators $\rbra{K_j}_j$,
and by \cref{eq:choi-state,eq:vectorization} its Choi state satisfies
\begin{equation}
\label{eq:choi-kraus-vectors}
  v\rbra*{\Phi}
  =\frac1N\sum_j
    \operatorname{vec}\rbra*{K_j}
    \operatorname{vec}\rbra*{K_j}^{\dagger}.
\end{equation}
For a unitary channel,
$v\rbra{\Phi_U}=\ketbra{v\rbra{U}}{v\rbra{U}}$ is pure.

For a channel $\Phi$ with Kraus operators $\rbra{K_j}_j$, its
Pauli coefficients satisfy
\begin{equation}
\label{eq:channel-fourier-kraus}
  \wh\Phi\rbra*{x,y}
  =\sum_j\wh K_j\rbra*{x}\,\overline{\wh K_j\rbra*{y}},
  \qquad x,y\in\Z_4^n.
\end{equation}
This follows from \cref{eq:choi-kraus-vectors} and
$\bra{v\rbra{\sigma_x}}\operatorname{vec}\rbra{K_j}
=\sqrt N\wh K_j\rbra{x}$.

\begin{lemma}[Pauli coefficient matrix {\cite[Lemma~8]{BY25}}]
\label{lem:channel-fourier-normalization}
For every channel $\Phi\in C\rbra{\calH}$,
\begin{equation*}
  \wh\Phi\succeq0,
  \qquad\Tr\rbra{\wh\Phi}
    =\sum_{x\in\Z_4^n}\wh\Phi\rbra*{x,x}=1.
\end{equation*}
\end{lemma}

\begin{proof}
The matrix $\wh\Phi$ represents the density operator $v\rbra{\Phi}$ in
an orthonormal basis, so it is positive semidefinite with trace one.
\end{proof}

Thus $x\mapsto\wh\Phi\rbra{x,x}$ defines a probability distribution
on Pauli labels.

For a unitary channel $\Phi_U$, the single Kraus operator $U$ gives
\begin{equation}
\label{eq:unitary-channel-fourier}
  \wh{\Phi_U}\rbra*{x,y}
  =\wh U\rbra*{x}\,\overline{\wh U\rbra*{y}},
  \qquad \wh{\Phi_U}\rbra*{x,x}=\abs{\wh U\rbra*{x}}^2.
\end{equation}

\subsection{Junta classes}
\label{subsec:junta-classes}

For $T\subseteq\sbra{n}$, write $\overline{T}=\sbra{n}\setminus T$.
Each junta class below carries a superscript $\rbra{n}$ for the number of input bits or qubits, which we omit when it is clear.
A Boolean function, unitary, or channel is \emph{$\eps$-far} from a junta class if its distance to that class, defined below, is at least $\eps$.

\subsubsection{Junta Boolean functions}\label{subsec:junta-boolean}

\begin{definition}[Junta Boolean functions]\label{def:boolean-juntas}
For $T\subseteq\sbra{n}$, define the class of $T$-junta Boolean functions
on $n$ variables by
\begin{equation*}
  \calB_T^{\rbra*{n}}
  =\cbra*{x\mapsto f_T\rbra{x_T}\colon
    f_T:\cbra{0,1}^{\abs{T}}\to\cbra{0,1}}.
\end{equation*}
For an integer $0\leq k\leq n$, define the class of
$n$-variable $k$-junta Boolean functions by
\begin{equation*}
  \calB_k^{\rbra*{n}}
  =\bigcup_{\substack{T\subseteq\sbra*{n}\\\abs{T}\leq k}}
      \calB_T^{\rbra*{n}}.
\end{equation*}
\end{definition}

Equivalently, $f\in\calB_T^{\rbra{n}}$ if and only if
$\supp(f)\subseteq T$.

\begin{definition}[Boolean distance]\label{def:boolean-distance}
For Boolean functions $f,g:\cbra{0,1}^n\to\cbra{0,1}$, define the
uniform Hamming distance by
\begin{equation*}
  \distB\rbra*{f,g}
  =\Pr_{x\sim\operatorname{Unif}\rbra{\cbra{0,1}^n}}
    \sbra*{f\rbra{x}\neq g\rbra{x}}.
\end{equation*}
For a nonempty set $\calS$ of Boolean functions on $n$ variables, define
\begin{equation*}
  \distB\rbra*{f,\calS}
  =\min_{g\in\calS}\distB\rbra*{f,g}.
\end{equation*}
\end{definition}

\subsubsection{Junta unitaries}\label{subsec:junta-unitaries}

Let $\calH_T$ be the tensor product of the single-qubit spaces
indexed by $T$, in increasing order.
Set $\calH_\varnothing=\C$.
We write $V_T\ot I_{\overline{T}}$ for the operator that acts as
$V_T$ on $T$ and as the identity elsewhere, with tensor factors in the
original qubit order. We use the same notation for channels.
For $T\subseteq\sbra{n}$, write $\Z_4^T$ for the set of labels $z=\rbra{z_i}_{i\in T}$, and let $\sigma_z=\bigotimes_{i\in T}\sigma_{z_i}\in L\rbra{\calH_T}$, with factors in increasing order of $i$.
The all-identity label in $\Z_4^T$ is denoted by $0$, and for $x\in\Z_4^n$ we regard $x_T$ as an element of $\Z_4^T$.
We write $\Tr_{\overline{T}}$ for the partial trace over the tensor factors outside $T$.

\begin{definition}[Junta unitaries {\cite[Definition~2.2]{CNY23}}]
\label{def:quantum-juntas}
For $T\subseteq\sbra{n}$, define the class of $T$-junta unitaries by
\begin{equation*}
  \calU_T^{\rbra*{n}}
  =\cbra*{V_T\ot I_{\overline{T}}\colon
      V_T\in\mathrm U\rbra{2^{\abs{T}}}}.
\end{equation*}
For an integer $0\leq k\leq n$, define the class of
$n$-qubit $k$-junta unitaries by
\begin{equation*}
  \calU_k^{\rbra*{n}}
  =\bigcup_{\substack{T\subseteq\sbra*{n}\\\abs{T}\leq k}}
      \calU_T^{\rbra*{n}}.
\end{equation*}
\end{definition}

\begin{definition}[Unitary distance
  {\cite[Definition~2.3]{CNY23}}]
\label{def:unitary-distance}
For $U,V\in\mathrm U\rbra{N}$, define the phase-insensitive
normalized Hilbert--Schmidt distance by
\begin{equation}
\label{eq:unitary-distance}
  \distU\rbra*{U,V}
  =\min_{0\leq\theta<2\pi}
  \frac{1}{\sqrt{2N}}\norm{e^{i\theta}U-V}_2.
\end{equation}
For a nonempty set $\calS\subseteq\mathrm U\rbra{N}$, define
\begin{equation*}
  \distU\rbra*{U,\calS}
  =\inf_{V\in\calS}\distU\rbra*{U,V}.
\end{equation*}
\end{definition}

Expanding the squared norm in \cref{eq:unitary-distance}
and minimizing over the global phase gives
\begin{equation*}
  \distU\rbra*{U,V}^2
  =1-\frac{\abs{\langle U,V\rangle}}{N}
  =1-\frac{\abs{\Tr\rbra*{U^{\dagger}V}}}{N}.
\end{equation*}
In particular, $0\leq\distU(U,V)\leq1$.

\subsubsection{Junta channels}\label{subsec:junta-channels}

\begin{definition}[Junta channels {\cite[Definition~4]{BY25}}]
\label{def:junta-channels}
For $T\subseteq\sbra{n}$, define the class of $T$-junta channels by
\begin{equation*}
  \calC_T^{\rbra*{n}}
  =\cbra*{\Psi_T\ot I_{\overline{T}}\colon
    \Psi_T\in C\rbra*{\calH_T}}.
\end{equation*}
For an integer $0\leq k\leq n$, define the class of
$n$-qubit $k$-junta channels by
\begin{equation*}
  \calC_k^{\rbra*{n}}
  =\bigcup_{\substack{T\subseteq\sbra*{n}\\\abs{T}\leq k}}\calC_T^{\rbra*{n}}.
\end{equation*}
\end{definition}

Every $U\in\calU_T$ induces a channel $\Phi_U\in\calC_T$.

\begin{definition}[Channel distance {\cite[Eq.~(1)]{BY25}}]
\label{def:channel-distance}
For $\Phi,\Psi\in C\rbra{\calH}$, define the
normalized Hilbert--Schmidt distance by
\begin{equation*}
  \distC\rbra*{\Phi,\Psi}
  =
  \frac1{\sqrt2}
  \norm{v\rbra*{\Phi}-v\rbra*{\Psi}}_2
  =\frac1{N\sqrt2}
  \norm{J\rbra*{\Phi}-J\rbra*{\Psi}}_2.
\end{equation*}
For a nonempty set
$\calS\subseteq C\rbra{\calH}$, define
\begin{equation*}
  \distC\rbra*{\Phi,\calS}
  =\inf_{\Psi\in\calS}
  \distC\rbra*{\Phi,\Psi}.
\end{equation*}
\end{definition}

Since Choi states are density operators, $0\leq\distC(\Phi,\Psi)\leq1$.

The Choi states of unitary channels are pure, so for every
$U,V\in\mathrm U\rbra{N}$,
\begin{equation*}
  \distC\rbra*{\Phi_U,\Phi_V}^2
  =1-\frac{\abs{\langle U,V\rangle}^2}{N^2}
  =\distU\rbra*{U,V}^2
    \rbra*{2-\distU\rbra*{U,V}^2}.
\end{equation*}
The infimum defining $\distC\rbra{\Phi_U,\calC_k}$
ranges over all $k$-junta channels, including nonunitary ones.

\subsection{Influence}
\label{subsec:influence}

Influence measures the Fourier or Pauli mass carried by labels that
intersect a specified set of coordinates. For a candidate junta support
$T$, we use the influence of $\overline T$ to measure the mass outside
$T$. The bounds below relate this mass to distance from the junta classes.

\subsubsection{Influence of Boolean functions}
\begin{definition}[Boolean influence]\label{def:boolean-influence}
For $f:\cbra{0,1}^n\to\cbra{0,1}$ and $T\subseteq\sbra{n}$, define
\begin{equation*}
 \Inf_T\sbra*{f}
 =\sum_{S:S\cap T\ne\varnothing}
   \abs{\wh{(-1)^f}\rbra{S}}^2.
\end{equation*}
For $i\in\sbra{n}$, write
$\Inf_i\sbra{f}=\Inf_{\cbra{i}}\sbra{f}$.
\end{definition}

\begin{fact}[Coordinate influence {\cite[Theorem~2.20]{OD14}}]
\label{fact:boolean-coordinate-influence}
For every $f:\cbra{0,1}^n\to\cbra{0,1}$ and $i\in\sbra{n}$,
$\Inf_i\sbra{f}=\Pr_X\sbra{f(X)\ne f(X\oplus e_i)}$, where $X$ is uniform on $\cbra{0,1}^n$.
\end{fact}

By the coordinate-influence formula (\cref{fact:boolean-coordinate-influence}), $\supp(f)=\cbra{i\in\sbra{n}\colon\Inf_i\sbra{f}>0}$, so every Fourier sample of $(-1)^f$ is contained in $\supp(f)$.
For a set $T$, $\Inf_T\sbra{f}$ is the probability that a Fourier
sample of $(-1)^f$ intersects $T$.

\begin{lemma}[Influence controls Boolean junta distance]
\label{lem:boolean-influence-certificate}
Fix $f:\cbra{0,1}^n\to\cbra{0,1}$ and $T\subseteq\sbra{n}$.
\begin{enumerate}[label=(\roman*)]
\item $\Inf_{\overline T}\sbra{f}\geq2\distB\rbra{f,\calB_T}$.
\item For every $i\in\overline T$, $\distB\rbra{f,\calB_T}\geq\Inf_i\sbra{f}/2$.
\end{enumerate}
\end{lemma}

\begin{proof}
For $u\in\cbra{0,1}^{\abs{T}}$, let $b_T(u)=\Ex\sbra{(-1)^{f(X)}\mid X_T=u}$ for uniform $X$.
By Parseval's identity (\cref{fact:boolean-fourier}), $\Inf_{\overline T}\sbra{f}=1-\Ex_u\sbra{b_T(u)^2}$, where $u$ is uniform on $\cbra{0,1}^{\abs T}$, and majority prediction on each fiber gives $\distB\rbra{f,\calB_T}=\rbra{1-\Ex_u\sbra{\abs{b_T(u)}}}/2$.
Since $b_T(u)^2\leq\abs{b_T(u)}$, part~(i) follows.
For part~(ii), every $g\in\calB_T$ satisfies $g(x)=g(x\oplus e_i)$, so $g$ disagrees with $f$ on at least one endpoint of every pair $\cbra{x,x\oplus e_i}$ with $f(x)\neq f(x\oplus e_i)$.
By the coordinate-influence formula (\cref{fact:boolean-coordinate-influence}), these pairs cover an $\Inf_i\sbra{f}$ fraction of the inputs, so $g$ disagrees with $f$ on at least an $\Inf_i\sbra{f}/2$ fraction of the inputs.
\end{proof}

\subsubsection{Influence of unitaries}
Unless otherwise specified, Pauli labels in sums range
over $\Z_4^n$, subject to the indicated restrictions.

\begin{definition}[Unitary influence
  {\cite[Definitions~2.5 and~2.6]{CNY23}}]
\label{def:influence}
Fix $T\subseteq\sbra{n}$.
For a unitary $U\in\mathrm U\rbra{N}$, define
\begin{equation*}
  \Inf_T\sbra*{U}
  =\sum_{x:\,x_T\neq0}
  \abs{\wh U\rbra*{x}}^2.
\end{equation*}
For $j\in\sbra{n}$, write
$\Inf_j\sbra{U}=\Inf_{\cbra{j}}\sbra{U}$.
\end{definition}

By \cref{eq:unitary-pauli-mass}, the weights $\abs{\wh U\rbra{x}}^2$ form a probability distribution on $\Z_4^n$.
Under this distribution, $\Inf_T\sbra{U}$ is the probability that the support of the label intersects $T$, so it lies in $[0,1]$. Moreover, for every $T\subseteq\sbra{n}$,
\begin{equation*}
  \sum_{x:\,\supp\rbra*{x}\subseteq T}
    \abs{\wh U\rbra*{x}}^2
  =1-\Inf_{\overline{T}}\sbra*{U}.
\end{equation*}

The next lemma bounds the distance to $\calU_T$ by the influence outside $T$.

\begin{lemma}[Influence controls unitary junta distance]
\label{lem:influence-controls-distance}
For every $U\in\mathrm U\rbra{N}$ and $T\subseteq\sbra{n}$,
\begin{equation*}
  1-\sqrt{1-\Inf_{\overline{T}}\sbra*{U}}
  \leq\distU\rbra*{U,\calU_T}^2
   \leq\Inf_{\overline{T}}\sbra*{U}.
\end{equation*}
\end{lemma}

\begin{proof}
Set $d_T=2^{\abs{T}}$ and define the normalized partial trace
\begin{equation*}
  A_T=2^{-\abs{\overline{T}}}\Tr_{\overline{T}}\rbra*{U}
      \in L\rbra*{\calH_T}.
\end{equation*}
We first show, as a special case of~\cite[Eqs.~(14), (15) and~(22)]{GDKBR10}, that
\begin{equation}
\label{eq:unitary-exact-distance}
  \distU\rbra*{U,\calU_T}^2
  =1-\frac{\norm{A_T}_1}{d_T}.
\end{equation}
For every unitary $V_T$ on $\calH_T$,
\[
 \frac{\Tr\rbra*{\rbra{V_T^{\dagger}\ot\Id}U}}{N}
 =\frac{\Tr\rbra*{V_T^{\dagger}A_T}}{d_T}.
\]
Write a singular value decomposition $A_T=P\Sigma Q^{\dagger}$.
Then $\abs{\Tr\rbra*{V_T^{\dagger}A_T}}\leq\Tr\rbra{\Sigma}=\norm{A_T}_1$,
with equality at $V_T=PQ^{\dagger}$.
Minimizing the squared distance over $\calU_T$ therefore gives \cref{eq:unitary-exact-distance}.

For unit vectors $v,w\in\calH_T$
and an orthonormal basis $\rbra{b_j}_{j=1}^{2^{\abs{\overline{T}}}}$ of $\calH_{\overline{T}}$,
\begin{equation*}
  \abs{v^{\dagger}A_Tw}
  \leq2^{-\abs{\overline{T}}}\sum_{j=1}^{2^{\abs{\overline{T}}}}
    \abs{\rbra*{v\ot b_j}^{\dagger}U\rbra*{w\ot b_j}}
   \leq1.
\end{equation*}
Thus $\norm{A_T}_\infty\leq1$, so its singular values lie in $[0,1]$.
Every nonidentity Pauli matrix has trace zero, so expanding $U$ in the
Pauli basis gives
\[
 A_T=\sum_{x:\,\supp\rbra*{x}\subseteq T}
       \wh U\rbra*{x}\sigma_{x_T}.
\]
By the Pauli expansion (\cref{fact:pauli-expansion}) on $\calH_T$, this yields the partial-trace
identity~\cite[Lemma~A.1]{CNY23}
\begin{equation}
\label{eq:unitary-partial-trace-mass}
  \frac{\norm{A_T}_2^2}{d_T}
  =\sum_{x:\,\supp\rbra*{x}\subseteq T}\abs{\wh U\rbra*{x}}^2
   =1-\Inf_{\overline{T}}\sbra*{U}.
\end{equation}
Applying $\varsigma^2\leq\varsigma$ to each singular value $\varsigma\in[0,1]$ of $A_T$ and then Cauchy--Schwarz yields
\begin{equation*}
  \frac{\norm{A_T}_2^2}{d_T}
  \leq\frac{\norm{A_T}_1}{d_T}
   \leq\sqrt{\frac{\norm{A_T}_2^2}{d_T}}.
\end{equation*}
Combining these inequalities with
\cref{eq:unitary-partial-trace-mass,eq:unitary-exact-distance}
proves both bounds.
\end{proof}

The lower bound also follows from the proof of
\cite[Lemma~2.8]{CLL24}.

\begin{lemma}[Influence bounds for testing junta unitaries]
\label{lem:influence-certificate}
Fix $U\in\mathrm U\rbra{N}$, an integer $0\leq k\leq n$, and
$0\leq\eps\leq1$.
\begin{enumerate}[label=(\roman*)]
\item If $\distU\rbra{U,\calU_k}\geq\eps$, then
  $\Inf_{\overline{T}}\sbra{U}\geq\eps^2$ for every $T\subseteq\sbra{n}$
  with $\abs{T}\leq k$.
\item If $\Inf_{\overline{T}}\sbra{U}\geq2\eps^2-\eps^4$ for every
  $T\subseteq\sbra{n}$ with $\abs{T}\leq k$, then
  $\distU\rbra{U,\calU_k}\geq\eps$.
\end{enumerate}
\end{lemma}

\begin{proof}
For every $\abs{T}\le k$, \cref{lem:influence-controls-distance} gives
$\distU(U,\calU_k)^2\le\distU(U,\calU_T)^2\le\Inf_{\overline T}\sbra{U}$,
which proves (i). Under the hypothesis of (ii), the lower bound in \cref{lem:influence-controls-distance} gives, for every $T$ with $\abs{T}\le k$,
\[
 \distU(U,\calU_T)^2\ge1-\sqrt{1-(2\eps^2-\eps^4)}=\eps^2,
\]
where the equality uses $0\le\eps\le1$. Taking the infimum over $\abs{T}\le k$ proves (ii).
\end{proof}

Part~(i) improves the threshold $\eps^2/4$ of~\cite[Lemma~2.2]{CNY23}, which is due to~\cite{Wan11}, and the threshold $\eps^2/2$ of~\cite[Lemma~2.9]{CLL24} to $\eps^2$.

\subsubsection{Influence of channels}
\begin{definition}[Channel influence {\cite[Definition~9]{BY25}}]
\label{def:channel-influence}
Fix $T\subseteq\sbra{n}$.
For a channel $\Phi\in C\rbra{\calH}$, define
\begin{equation*}
  \Inf_T\sbra*{\Phi}
  =\sum_{x:\,x_T\neq0}
  \wh\Phi\rbra*{x,x}.
\end{equation*}
For $j\in\sbra{n}$, write
$\Inf_j\sbra{\Phi}=\Inf_{\cbra{j}}\sbra{\Phi}$.
\end{definition}

Since $x\mapsto\wh\Phi\rbra{x,x}$ is a probability distribution on $\Z_4^n$ by \cref{lem:channel-fourier-normalization}, we have $0\le\Inf_T\sbra{\Phi}\le1$ and, for every $T\subseteq\sbra{n}$,
\begin{equation}
\label{eq:channel-complementary-mass}
  \sum_{x:\,\supp\rbra*{x}\subseteq T}
    \wh\Phi\rbra*{x,x}
  =1-\Inf_{\overline{T}}\sbra*{\Phi}.
\end{equation}

\begin{lemma}[Support characterization for junta channels,
  cf.~{\cite[Theorem~12]{BY25}}]
\label{lem:junta-fourier-support}
For every channel $\Phi\in C\rbra{\calH}$
and $T\subseteq\sbra{n}$, $\Phi\in\calC_T$ if and only if $\Inf_{\overline{T}}\sbra{\Phi}=0$.
\end{lemma}

\begin{proof}
Tensoring Kraus operators with $I_{\overline T}$ proves the forward
implication by \cref{eq:channel-fourier-kraus}. Conversely, for Kraus operators $\rbra{K_j}_j$ of $\Phi$ (\cref{fact:kraus-representation}), that identity
expresses zero influence as a sum of nonnegative terms
$\abs{\wh K_j(x)}^2$ over labels outside $T$. Each term is zero, so every
Kraus operator has the form $K_j=L_j\ot I_{\overline T}$ for some $L_j\in L\rbra{\calH_T}$. Trace preservation gives
$\sum_jL_j^{\dagger}L_j=I_T$, defining the required channel on $T$.
\end{proof}

To project onto Pauli vectors supported on $T$, use
\begin{equation*}
  P_T=\sum_{x:\,\supp\rbra*{x}\subseteq T}
    \ketbra{v\rbra*{\sigma_x}}{v\rbra*{\sigma_x}}.
\end{equation*}

\begin{lemma}[Influence bound for testing junta channels]
\label{lem:channel-influence-certificate}
Fix a channel $\Phi\in C\rbra{\calH}$, an integer $0\leq k\leq n$, and
$0\leq\eps\leq1$.
If $\distC\rbra{\Phi,\calC_k}\geq\eps$, then
$\Inf_{\overline{T}}\sbra{\Phi}\geq\eps^2$
for every $T\subseteq\sbra{n}$ with $\abs{T}\leq k$.
\end{lemma}

\begin{proof}
Fix $T\subseteq\sbra{n}$ with $\abs{T}\leq k$.
We adapt the channel-completion construction in the proof of~\cite[Theorem~12]{BY25} and use Hilbert--Schmidt orthogonality to find $\Psi\in\calC_T$ with $\distC\rbra{\Phi,\Psi}^2\leq\Inf_{\overline{T}}\sbra{\Phi}$.
By the Kraus representation (\cref{fact:kraus-representation}), choose Kraus operators
$\rbra{K_j}_j$ for $\Phi$ and expand each of them in
the Pauli basis on $\calH_{\overline{T}}$ as
\begin{equation*}
  K_j=\sum_{z\in\Z_4^{\overline{T}}}B_{j,z}\ot\sigma_z,
  \qquad B_{j,z}\in L\rbra*{\calH_T}.
\end{equation*}
Trace preservation and the Pauli expansion (\cref{fact:pauli-expansion}) on $\calH_{\overline{T}}$ give
\begin{equation*}
  \sum_{j,z}B_{j,z}^{\dagger}B_{j,z}
  =2^{-\abs{\overline{T}}}\Tr_{\overline{T}}
    \rbra*{\sum_jK_j^{\dagger}K_j}
   =I_T.
\end{equation*}
Thus $D=(I_T-\sum_jB_{j,0}^{\dagger}B_{j,0})^{1/2}$ is well defined, and
$B_{j,0}\ot I_{\overline T}$ together with $D\ot I_{\overline T}$ define
a channel $\Psi\in\calC_T$.

The projection keeps exactly the Pauli coefficients of $K_j$ supported on $T$, so
\begin{equation}
\label{eq:channel-kraus-support-projection}
  P_T\operatorname{vec}\rbra*{K_j}
  =\operatorname{vec}\rbra*{B_{j,0}\ot I_{\overline{T}}}.
\end{equation}
Set $\Gamma=P_Tv\rbra{\Phi}P_T$.  By
\cref{eq:choi-kraus-vectors,eq:channel-kraus-support-projection},
\begin{equation*}
  v\rbra*{\Psi}=\Gamma+\kappa, \qquad
  \kappa=\frac1N\operatorname{vec}\rbra*{D\ot I_{\overline{T}}}
    \operatorname{vec}\rbra*{D\ot I_{\overline{T}}}^{\dagger}.
\end{equation*}
The matrix $\kappa$ is positive semidefinite and has rank at most one.
Since $v\rbra{\Psi}$ has trace one, the definition of $P_T$
and \cref{eq:channel-complementary-mass} give
\begin{equation*}
  \norm{\kappa}_2
  =\Tr\rbra{\kappa}
   =1-\Tr\rbra{\Gamma}
   =\Inf_{\overline{T}}\sbra*{\Phi}.
\end{equation*}

By positive semidefiniteness of $\wh\Phi$
(\cref{lem:channel-fourier-normalization}),
$\abs{\wh\Phi\rbra{x,y}}^2\leq\wh\Phi\rbra{x,x}\wh\Phi\rbra{y,y}$ for
all $x,y$. In the Pauli-vector basis, $v\rbra{\Phi}-\Gamma$ agrees with $v\rbra{\Phi}$
on the entries for which at least one label has support not contained in
$T$, and it is zero elsewhere.  Therefore \cref{prop:superoperator-plancherel}, applied to the superoperator with Choi matrix $N\rbra{v\rbra{\Phi}-\Gamma}$, gives
\begin{equation*}
  \norm{v\rbra*{\Phi}-\Gamma}_2^2
  \leq\sum_{x,y}\wh\Phi\rbra*{x,x}\wh\Phi\rbra*{y,y}
    -\sum_{x,y:\,\supp\rbra*{x}\cup\supp\rbra*{y}\subseteq T}\wh\Phi\rbra*{x,x}\wh\Phi\rbra*{y,y}
  =1-\rbra*{1-\Inf_{\overline{T}}\sbra*{\Phi}}^2.
\end{equation*}
Finally, $P_T\rbra{v\rbra{\Phi}-\Gamma}P_T=0$ and
$P_T\kappa P_T=\kappa$, so $v\rbra{\Phi}-\Gamma$ and $\kappa$ are
orthogonal in the Hilbert--Schmidt inner product.  Consequently,
\begin{equation*}
  2\distC\rbra*{\Phi,\Psi}^2
  =\norm{v\rbra*{\Phi}-\Gamma}_2^2+\norm{\kappa}_2^2
  \leq1-\rbra*{1-\Inf_{\overline{T}}\sbra*{\Phi}}^2
    +\Inf_{\overline{T}}\sbra*{\Phi}^2
   =2\Inf_{\overline{T}}\sbra*{\Phi}.
\end{equation*}
Hence $\eps^2\le\distC(\Phi,\calC_k)^2\le\distC(\Phi,\Psi)^2\le\Inf_{\overline T}\sbra{\Phi}$.
\end{proof}

This bound replaces the threshold $\eps^2/4$ in
\cite[Corollary~13]{BY25} by $\eps^2$.

The two influence definitions agree on unitary channels.
For every $U\in\mathrm U\rbra{N}$ and $T\subseteq\sbra{n}$,
\cref{eq:unitary-channel-fourier} gives
\begin{equation}
\label{eq:unitary-channel-influence}
  \Inf_T\sbra*{\Phi_U}=\Inf_T\sbra*{U}.
\end{equation}

\begin{lemma}[Distance to junta channels from unitary influence]
\label{lem:unitary-channel-influence-certificate}
Fix a unitary $U\in\mathrm U\rbra{N}$, an integer $0\leq k\leq n$, and
$0\leq\delta\leq1$.
If $\Inf_{\overline{T}}\sbra{U}\geq\delta$
for every $T\subseteq\sbra{n}$ with $\abs{T}\leq k$, then
$\distC\rbra{\Phi_U,\calC_k}\geq\sqrt{\delta-\delta^2/2}$.
\end{lemma}

\begin{proof}
Fix $\Psi\in\calC_k$ and choose $T$ with $\abs{T}\leq k$ such that
$\Psi\in\calC_T$. Set $\rho=v\rbra{\Phi_U}$.
Since $v\rbra{\Psi}=P_Tv\rbra{\Psi}P_T$, we have $\rho-v\rbra{\Psi}=\rbra{\rho-P_T\rho P_T}+P_T\rbra{\rho-v\rbra{\Psi}}P_T$.
The matrix $\rho-P_T\rho P_T$ is orthogonal in the Hilbert--Schmidt inner product to every matrix of the form $P_TYP_T$.
Applying this to the decomposition above and to $\rho=\rbra{\rho-P_T\rho P_T}+P_T\rho P_T$ gives
\[
  2\distC\rbra*{\Phi_U,\Psi}^2\geq\norm{\rho-P_T\rho P_T}_2^2=\norm{\rho}_2^2-\norm{P_T\rho P_T}_2^2=1-\rbra*{1-\Inf_{\overline{T}}\sbra*{U}}^2\geq2\delta-\delta^2,
\]
as in the block-projection bound of~\cite[Eqs.~(D101)--(D104)]{WLKD25}.
The second equality follows from the purity of $\rho$ and
\cref{eq:channel-complementary-mass,eq:unitary-channel-influence}.
The last inequality uses that $x\mapsto2x-x^2$ is increasing on
$[0,1]$.
Taking the infimum over $\Psi\in\calC_k$ proves the claim.
\end{proof}

\subsection{Access models}\label{subsec:access-models}
A tester knows the problem parameters, namely $n$, $k$, the distance parameter or threshold, and the ground set, and accesses its unknown input as specified below.
Its other operations are unrestricted. They may use ancillas, private randomness, and intermediate measurements, and may depend on earlier measurement outcomes.

A quantum query algorithm alternates oracle calls with oracle-independent operations.
For unitary and channel queries, the query register $\mathsf Q$ consists of $n$ qubits, and all other qubits form a private workspace $\mathsf W$ of unrestricted size.
Query complexity counts oracle calls, and sample complexity counts requests for independent samples.
In each case, complexity is the largest count over all inputs and all executions, including inputs outside the promise.
A tester must succeed with the stated probability on every promised input, where the probability is over the samples, its private randomness, and its measurement outcomes.
It has no correctness requirement outside the promise.
A query algorithm is \emph{nonadaptive} if, before its first query, it fixes the input state of every query and the measurement applied to the output of every query.

\subsubsection{Boolean function access}\label{subsec:boolean-access}
For $f:\cbra{0,1}^n\to\cbra{0,1}$, a query applies
\begin{equation*}
 O_f\ket{x,b}=\ket{x,b\oplus f(x)}.
\end{equation*}
This oracle is self-inverse. Phase kickback implements
$U_f\ket x=(-1)^{f(x)}\ket x$ with one query.
The algorithms of~\cite{ABRW16} use the phase oracle $U_f$.
The two oracles are equivalent for this testing task up to constant query overhead.
Indeed, fix a known $x_0$ and set $g(x)=f(x)\oplus f(x_0)$.
On a basis state $\ket{x,b}$, using known controlled swaps, apply one phase query to a register that holds $x$ when $b=1$ and $x_0$ when $b=0$.
Hadamards on $b$ turn this map into a query to $g$.
Complementing a Boolean function preserves its support and its distance to
$\calB_k$ for every $k$, so the unknown constant $f(x_0)$ does not affect the testing task.

\subsubsection{Unitary access}\label{subsec:unitary-access}

We consider two unitary access models for an unknown unitary $U$ on
$\mathsf Q$.
In the \emph{standard} model, one query applies either $U$ or $U^{\dagger}$.
In the \emph{forward-only} model, one query applies $U$.
Neither model supplies controlled-$U$ or controlled-$U^{\dagger}$.
A \emph{controlled query} applies a coherent choice among the identity, the oracle, and, in the standard model, its inverse.
Some of our lower bounds allow controlled queries, and their statements say so.
The testing promise is the same in both models.

A forward query acts on any
joint state $\rho$ of $\mathsf Q$ and $\mathsf W$ as
\begin{equation}
\label{eq:forward-query-channel}
  \calO_{\Phi_U}\rbra*{\rho}
  =\rbra*{U_{\mathsf Q}\ot I_{\mathsf W}}\rho
   \rbra*{U_{\mathsf Q}^{\dagger}\ot I_{\mathsf W}}.
\end{equation}
This is the density-operator description of a single forward query
to $U$. The factor $U^{\dagger}$ comes from conjugation and does not require
an inverse query.

By \cref{eq:forward-query-channel}, for every $\theta\in\R$, $U$ and $e^{i\theta}U$ induce the same forward query channel, and the same holds for inverse queries.
Hence every algorithm in either unitary access model has the same acceptance
probability on $U$ and on $e^{i\theta}U$.
This motivates the phase-insensitive distance in
\cref{def:unitary-distance}.

\subsubsection{Channel access}\label{subsec:channel-access}

With \emph{channel access} to an unknown $n$-qubit channel $\Phi$, one query applies
\begin{equation*}
  \calO_\Phi=\Phi\ot I_{\mathsf W}
\end{equation*}
to the joint state of $\mathsf Q$ and $\mathsf W$, so its input may be entangled with $\mathsf W$.
The oracle gives no access to an environment, a Kraus label, or a Stinespring dilation of $\Phi$.
For a unitary channel $\Phi_U$, one such query is a forward query to $U$ by \cref{eq:forward-query-channel}.

\subsubsection{Independent samples}\label{subsec:independent-samples}

Let $\calX$ be a known finite set.
A tester with \emph{independent samples} from an unknown $p\in\Delta(\calX)$ receives a new draw from $p$ at each request, independent of all earlier draws, as in~\cite[Definition~1]{GL20}.
A tester that requests at most $t$ samples may receive $Y\sim p^{\ot t}$ at the start and ignore the samples it does not request.
For $y\in\calX^t$, let $F(y)\in[0,1]$ be its acceptance probability over private randomness when $Y=y$.
Then $F$ does not depend on $p$, and the acceptance probability is $\Ex_{Y\sim p^{\ot t}}\sbra{F(Y)}$.

\subsubsection{Reversible sampler access}\label{subsec:reversible-sampler-access}

A \emph{reversible sampler} for a distribution $\calD$ on $2^{\sbra{n}}$ is a unitary $A$ on a finite-dimensional space with an $n$-qubit subset register, such that
\begin{equation*}
 A\ket0=\sum_{X\subseteq\sbra{n}}\ket X\ket{\phi_X},
 \qquad \norm{\phi_X}^2=\calD(X),
\end{equation*}
where $\ket0$ is the all-zero state of all registers, $\ket X$ is the incidence string of $X$ in the computational basis, and the auxiliary outputs $\phi_X$ are otherwise arbitrary.
One query is one ordinary call to $A$ or to $A^{\dagger}$. Controlled versions of $A$ and $A^{\dagger}$ are not supplied.
A tester must succeed for all auxiliary outputs and every unitary extension, that is, for every unitary $A$ with this action on $\ket0$.
A lower bound may therefore fix one such $A$ for each distribution.
This is the purified query access of~\cite[Definition~2]{GL20} for classical distributions, except that the unit vectors $\phi_X/\norm{\phi_X}$ need not be orthonormal.
Copying $X$ into an additional register initialized to $\ket0$ makes them orthonormal without queries, so the two access models are equivalent.

\subsection{Quantum algorithmic tools}\label{subsec:quantum-tools}
This subsection records the quantum subroutines used by the testers, together with the Fourier and Pauli samplers, each of which uses one query.

\subsubsection{Basic subroutines}\label{subsec:basic-quantum-subroutines}

\begin{lemma}[Quantum phase estimation {\cite{Kit95,CEMM98}}]
\label{lem:phase-estimation}
Let $W$ be a unitary on a finite-dimensional space and let $0<\theta\le1$.
There is a unitary circuit using $O(1/\theta)$ controlled applications of $W$ and $W^{\dagger}$ that maps $\ket\psi\ket0$ to $\ket\psi\ket{\omega_\psi}$ for every unit eigenvector $\ket\psi$ of $W$ with eigenvalue $e^{i\phi}$, where $\ket{\omega_\psi}$ is a unit state of the work registers.
Measuring the work registers gives an estimate $\widehat\phi$ satisfying
\[
 \Pr\sbra*{\min_{j\in\Z}\abs{\widehat\phi-\phi+2\pi j}\le\theta}\ge\frac9{10}.
\]
\end{lemma}

Amplitude amplification extends Grover's search algorithm~\cite{Gro96}.

\begin{lemma}[Randomized Grover iterations, adapted from {\cite[Lemma~2]{BBHT98}}]
\label{lem:randomized-grover}
Let $A$ be a finite-dimensional unitary, let $\Pi$ be a known orthogonal projector, and put $p=\norm{\Pi A\ket0}^2$ with $0<p<1$ and $\theta=\arcsin\sqrt p$.
Let $M\geq1/\sin(2\theta)$ be an integer, and let $t$ be uniform on $\cbra{0,\ldots,M-1}$.
Then applying $A$ to $\ket0$, then $t$ iterations of $-A(2\ket0\bra0-\Id)A^{\dagger}(2\Pi-\Id)$, and then measuring $\Pi$ succeeds with probability at least $1/4$.
\end{lemma}

\begin{lemma}[Quantum amplitude amplification {\cite{BHMT02}}]
\label{lem:amplitude-amplification}
Let $A$ be a finite-dimensional unitary, let $\Pi$ be a known orthogonal projector, and put $p=\norm{\Pi A\ket0}^2$.
For $0<\zeta\le1$ and $0<\beta<1/2$, there is a procedure that prepares states, measures $\Pi$ on them, and uses $O(\zeta^{-1/2}\ln(1/\beta))$ ordinary queries to $A,A^{\dagger}$ on every input.
If $p\ge\zeta$, some measurement succeeds with probability at least $1-\beta$.
If $p=0$, no measurement succeeds.
\end{lemma}

\begin{proof}
Put $M=\ceil{\zeta^{-1/2}}$ and use the Grover iterate $-A(2\ket0\bra0-\Id)A^{\dagger}(2\Pi-\Id)$.
In each trial, first measure $\Pi$ on $A\ket0$.
On failure, reset, apply $A$ followed by a uniformly random number of Grover iterations from $\cbra{0,\ldots,M-1}$, and measure $\Pi$ again.
For $p\ge1/2$, the first measurement succeeds with probability at least $1/2$.
For $\zeta\le p<1/2$, we have $M\ge\zeta^{-1/2}\ge p^{-1/2}$ and $\sin(2\arcsin\sqrt p)=2\sqrt{p(1-p)}>\sqrt p$, so \cref{lem:randomized-grover} bounds the second attempt's success probability below by $1/4$.
Repeat $O(\ln(1/\beta))$ independent trials.
Each trial uses $O(M)$ ordinary queries, and neither attempt succeeds at $p=0$.
\end{proof}

\begin{lemma}[Exact conditional sampling, adapted from {\cite[Theorem~3]{BHMT02}}]
\label{lem:exact-conditional-sampling}
Let $A$ be a finite-dimensional unitary, let $\Pi$ be a known orthogonal projector, and put $p=\norm{\Pi A\ket0}^2$.
There is a procedure using ordinary queries to $A,A^{\dagger}$ with the following properties.
If $p>0$, it terminates with probability one after $O(p^{-1/2})$ queries in expectation, and its output has the law of a computational-basis measurement of $\Pi A\ket0/\sqrt p$.
If $p=0$, it never terminates.
\end{lemma}

\begin{proof}
For $0<p<1$, use the unknown-success-probability search procedure
of~\cite[Theorem~3]{BHMT02}, with success projector $\Pi$.
The reflections $2\Pi-\Id$ and
$A(2\ket0\bra0-\Id)A^{\dagger}$
require only ordinary queries to $A,A^{\dagger}$.
Each trial stays in the span of $\Pi A\ket0$ and
$(\Id-\Pi)A\ket0$. Conditional on success, its state is
$\Pi A\ket0/\sqrt p$ up to a global phase.
Measuring in the computational basis gives the stated output law.
The procedure terminates almost surely after
$O(p^{-1/2})$ queries in expectation.
If $p=1$, the first direct trial succeeds.
If $p=0$, no trial succeeds and the procedure continues indefinitely.
\end{proof}

\begin{lemma}[Quantum amplitude estimation {\cite{BHMT02}}]
\label{lem:amplitude-estimation}
Let $A$ be a finite-dimensional unitary, let $\cbra{\Pi_j}_{j\in E}$ be a known family of orthogonal projectors indexed by a finite set $E$, and put $p_j=\norm{\Pi_jA\ket0}^2$.
For $0<\tau\le1$ and $0<\beta<1/2$, there is a unitary circuit $V$ that uses $O(\tau^{-1/2}\ln(1/\beta))$ ordinary queries to $A,A^{\dagger}$ in an order depending only on $\tau$ and $\beta$, and that satisfies $V\ket j\ket0=\ket j\ket{\xi_j}$ for every $j\in E$, where $\ket0$ is the all-zero state of the work registers.
We call a unitary of this form label-preserving.
For each $j\in E$, measuring the answer bit of $\ket{\xi_j}$ gives \SMALL\ if $p_j\le\tau/4$ and \LargeCase\ if $p_j\ge\tau$, with error at most $\beta$.
\end{lemma}

\begin{proof}
Set $R_0=2\ket0\bra0-\Id$ and $W_j=-R_0A^{\dagger}(2\Pi_j-\Id)A$.
The Grover rotation~\cite[Eqs.~(4)--(6)]{BHMT02} shows that all eigencomponents of $\ket0$ have circular eigenphase magnitude $2\arcsin\sqrt{p_j}$, including at $p_j\in\cbra{0,1}$.
Controlled calls to the conjugated reflection use two ordinary queries, since
\[
 \ketbra00\ot\Id+\ketbra11\ot A^{\dagger}(2\Pi_j-\Id)A
 =(\Id\ot A^{\dagger})\rbra*{\ketbra00\ot\Id+\ketbra11\ot(2\Pi_j-\Id)}(\Id\ot A).
\]
Only the known middle reflection depends on $j$ or on the control bit, so controlled calls to $W_j,W_j^{\dagger}$ have the same query cost for every label.

By quantum phase estimation (\cref{lem:phase-estimation}) with accuracy $\sqrt\tau/4$, the circular distance of the estimate from zero distinguishes the two cases with error at most $1/10$.
Indeed, their eigenphase magnitudes are separated by at least $\sqrt\tau$, and we compare with the midpoint $\arcsin\sqrt\tau+\arcsin(\sqrt\tau/2)$.
Repeat on independently initialized work registers $O(\ln(1/\beta))$ times and compute the majority reversibly, keeping all work registers.
Every operation preserves $j$, and the oracle-query sequence depends only on $\tau,\beta$.
\end{proof}

\subsubsection{Reversible samplers}\label{subsec:reversible-quantum-samplers}

Let $H$ be the single-qubit Hadamard gate, and identify subsets of $\sbra{n}$ with their incidence strings.

\begin{lemma}[Fourier sampling {\cite{AS07}}]
\label{lem:fourier-sampling}
Let $f:\cbra{0,1}^n\to\cbra{0,1}$, where $n\ge1$.
The unitary $A_f=H^{\ot n}U_fH^{\ot n}$ satisfies
\[
 A_f\ket0=\sum_{S\subseteq\sbra{n}}\wh{(-1)^f}\rbra{S}\ket S.
\]
Both $A_f$ and $A_f^{\dagger}$ use one Boolean oracle query.
\end{lemma}

\begin{proof}
The Hadamard transform before and after $U_f$ gives the stated Fourier coefficients.
Phase kickback implements $U_f$ with one Boolean query, and $U_f$ is self-inverse.
\end{proof}

\begin{lemma}[Pauli sampling {\cite{MO10,Wan11,CNY23}}]
\label{lem:pauli-sampling}
Let $U$ be an $n$-qubit unitary, where $n\ge1$.
There is a unitary $A_U$ satisfying
\[
 A_U\ket0=\sum_{x\in\Z_4^n}\wh U\rbra{x}\ket x.
\]
The circuits $A_U$ and $A_U^{\dagger}$ use one query to $U$ and $U^{\dagger}$, respectively.
\end{lemma}

\begin{proof}
Prepare $n$ Bell pairs and apply $U$ to one half, obtaining $\ket{v\rbra{U}}=\sum_x\wh U\rbra{x}\ket{v\rbra{\sigma_x}}$ by the Pauli expansion (\cref{fact:pauli-expansion}).
On each pair, change from the Pauli-vector basis to the computational basis, including its known phases.
This gives the required state, and reversing the circuit replaces the call to $U$ by one to $U^{\dagger}$.
\end{proof}

Writing $\supp(x)$ into an additional register initialized to $\ket0$ while keeping $x$ gives a reversible subset sampler without further queries.

\section{Union Size Testing}\label{sec:subset-cover}

This section studies union size testing, the common problem behind our results.
Its lower bounds also supply the classical and moment-matching ingredients of our quantum lower bounds.

\begin{problem}[Union Size Testing]\label{prob:set-testing}
Fix a known finite ground set $\Omega$, an integer $k\ge1$, and
$0<\delta\le1$.
Given independent samples from, or a reversible sampler for, an unknown distribution $\calD$ on $2^\Omega$, distinguish between the following two cases
with probability at least $2/3$, where $X\sim\calD$:
\begin{itemize}
 \item (YES) There exists $T\subseteq\Omega$ with $\abs{T}\le k$
 such that $\Pr\sbra{X\subseteq T}=1$.
 \item (NO) For every $T\subseteq\Omega$ with $\abs{T}\le k$,
 $\Pr\sbra{X\nsubseteq T}\ge\delta$.
\end{itemize}
\end{problem}

A set $T$ as in the YES case is called a \emph{witness}.
For $T\subseteq\Omega$, the probability $\Pr\sbra{X\nsubseteq T}$ is called the \emph{uncovered probability} of $T$.

For independent samples, \cref{subsec:cover-upper-independent} gives a tester and \cref{subsec:cover-lower-independent} proves a matching lower bound.
\Cref{subsec:cover-upper-reversible,subsec:cover-lower-reversible} do the same for a reversible sampler.

\subsection{Upper bound from independent samples}\label{subsec:cover-upper-independent}
We first consider independent samples in the model of
\cref{subsec:independent-samples}, with $\calX=2^\Omega$.

\begin{theorem}[Upper bound for Union Size Testing from independent samples]
\label{thm:set-testing}
Fix a finite set $\Omega$, an integer $k\ge1$, and $0<\delta\le1$. Given independent samples from a distribution $\calD$ on $2^\Omega$, there is a randomized algorithm for \cref{prob:set-testing} that succeeds with probability at least $2/3$ and uses $O(k/(\delta\log(ek)))$ samples.
\end{theorem}

\subsubsection{A Poisson statistic from a Fej\'er polynomial}
\label{subsec:fejer-statistic}

The tester below needs a bounded count statistic whose mean is zero for an
absent element, at most one for every element, proportional to the
inclusion probability for a rare element, and close to one for an element
with large inclusion probability. We adapt the polynomial-estimator approach
of~\cite[Eqs.~(18)--(25)]{WY19} using the Fej\'er polynomial below.

Let $T_j$ be the Chebyshev polynomial characterized by
$T_j\rbra{\cos\theta}=\cos\rbra{j\theta}$.
Fix $A>0$ and an odd integer $d\geq3$.
Define the polynomial $F_d$ by
\begin{equation}
\label{eq:fejer-polynomial}
  F_d\rbra*{x}=
  \begin{cases}
    \displaystyle\frac{1-T_d\rbra*{1-2Ax/d^2}}{2Ax},
      &x\neq0,\\
    1,&x=0.
  \end{cases}
\end{equation}
Set $\lam=2A/d$.  For $x\geq0$, define
$Q_d\rbra{x}=1-e^{-\lam x}F_d\rbra{x}$.
The notation for $F_d$ and $Q_d$ suppresses their dependence on $A$.

\begin{proposition}[Fej\'er polynomial properties]
\label{prop:fejer-properties}
For every $A>0$ and every odd integer $d\geq3$, the following hold.
\begin{enumerate}[label=(\roman*)]
  \item The function $F_d$ is a polynomial of degree $d-1$.
    If $F_d\rbra{x}=\sum_{j=0}^{d-1}f_{d,j}x^j$, then
    $f_{d,0}=1$ and
    \begin{equation*}
      \frac{j!\abs{f_{d,j}}}{\lam^j}\leq3^d
      \qquad\rbra*{0\leq j\leq d-1}.
    \end{equation*}

  \item The function $Q_d$ satisfies $Q_d\rbra{0}=0$ and
    \[
  \frac1{128}\min\cbra*{Ax,1}\leq Q_d\rbra*{x}\leq1\quad\rbra*{x\geq0},\qquad 1-Q_d\rbra*{x}\leq\min\cbra*{1,\frac1{Ax}}\quad\rbra*{x>0}.
\]
\end{enumerate}
\end{proposition}

\begin{proof}
\noindent\emph{Polynomial structure and coefficients.}
Since $T_d(1)=1$ and $T_d'(1)=d^2$, the numerator in \cref{eq:fejer-polynomial} is zero at $x=0$ and has derivative $2A$ there. Division by $2Ax$ therefore gives a polynomial $F_d$ of degree $d-1$ with $f_{d,0}=1$, and consequently $Q_d(0)=0$.

For $0<Ax\leq d^2$, set $z=Ax/d^2$.
Choose $\theta\in\interval[open left]{0}{\pi}$ so that $z=\sin^2\rbra{\theta/2}$.
Since $1-T_d\rbra{1-2z}=1-\cos\rbra{d\theta}=2\sin^2\rbra{d\theta/2}$ and
$2Ax=2d^2\sin^2\rbra{\theta/2}$, the definition of $F_d$ gives the first equality in
\begin{equation}
  F_d\rbra*{x}
  =\rbra*{\frac{\sin\rbra*{d\theta/2}}{d\sin\rbra*{\theta/2}}}^2
   =\frac1{d^2}\abs*{\sum_{j=0}^{d-1}e^{ij\theta}}^2
  =\frac1{d^2}\sbra*{d+2\sum_{j=1}^{d-1}\rbra*{d-j}\cos\rbra*{j\theta}},
  \label{eq:fejer-cosine-expansion}
\end{equation}
whose remaining equalities follow by expanding the squared geometric sum.

Since $\cos\rbra{j\theta}=T_j\rbra{1-2Ax/d^2}$, the last expression in
\cref{eq:fejer-cosine-expansion} gives, for $0<Ax\leq d^2$,
\begin{equation*}
  F_d\rbra*{x}=\frac1d+\frac2{d^2}
     \sum_{j=1}^{d-1}\rbra*{d-j}T_j\rbra*{1-2Ax/d^2}.
\end{equation*}
Both sides are polynomials, so the identity holds for every $x$.
The weights $1/d$ and $2\rbra{d-j}/d^2$ are nonnegative and sum to
one.  Let $\norm{T_j}_1$ denote the sum of the absolute values of the monomial
coefficients of $T_j$.  For $j\geq2$, the recurrence
$T_j\rbra{y}=2yT_{j-1}\rbra{y}-T_{j-2}\rbra{y}$ gives
$\norm{T_j}_1\leq2\norm{T_{j-1}}_1+\norm{T_{j-2}}_1$.  Since $\norm{T_0}_1=\norm{T_1}_1=1$, induction yields
$\norm{T_j}_1\leq3^j$.

For $0\leq j,t\leq d-1$, let $c_{j,t}$ be the coefficient of $x^t$ in $T_j\rbra{1-2Ax/d^2}$.  Expanding each monomial and using
$\binom rt\leq d^t/t!$ for $0\leq r\leq j\leq d-1$ gives
\begin{equation*}
  \abs{c_{j,t}}
  \leq\norm{T_j}_1\frac{d^t}{t!}\rbra*{\frac{2A}{d^2}}^t
   \leq3^d\frac{\lam^t}{t!},
\end{equation*}
where $\binom rt=0$ when $r<t$.  Applying this bound to every term in
the convex combination proves (i).

\noindent\emph{Upper bounds for $Q_d$.}
First let $0<Ax\leq d^2$, with $z$ and $\theta$ as above.
The triangle inequality bounds the squared sum by $d^2$.
The inequality $\abs{\sin\rbra{d\theta/2}}\leq1$ bounds the squared ratio
by $1/\rbra{d^2z}$.  Hence
\begin{equation*}
  0\leq e^{-\lam x}F_d\rbra*{x}
   \leq F_d\rbra*{x}
   \leq\min\cbra*{1,\frac1{d^2z}}
   =\min\cbra*{1,\frac1{Ax}}.
\end{equation*}

Now let $Ax\geq d^2$, set $z=Ax/d^2\geq1$, and write $v=2z-1$.
Since $v\ge1$, write $v=\cosh u$ with $u\ge0$. By the Chebyshev identity $T_d\rbra{\cosh u}=\cosh\rbra{du}$,
\begin{equation*}
  1\leq T_d\rbra*{v}
   \leq\rbra*{v+\sqrt{v^2-1}}^d
   \leq\rbra*{4z}^d.
\end{equation*}
Since $d$ is odd, $T_d\rbra{-v}=-T_d\rbra{v}$, so \cref{eq:fejer-polynomial} gives
\begin{equation*}
  0\leq F_d\rbra{x}
  =\frac{1+T_d\rbra{v}}{2d^2z}
  \leq\frac{\rbra{4z}^d}{d^2z}.
\end{equation*}
For $z\geq1$, the function $2z-\ln\rbra{4z}$ is positive at one and
increasing.  Since $\lam x=2dz$, we obtain
\begin{equation*}
  0\leq e^{-\lam x}F_d\rbra*{x}
   \leq\frac{\exp\rbra*{d\ln\rbra*{4z}-2dz}}{d^2z}
   \leq\frac1{d^2z}
   =\frac1{Ax}\leq1.
\end{equation*}
The two ranges show that $0\leq e^{-\lam x}F_d\rbra{x}\leq
\min\cbra*{1,1/\rbra{Ax}}$ for every $x>0$.  This proves the upper
bound on $Q_d$ and the stated deficit bound. The upper bound at
$x=0$ follows from $Q_d\rbra{0}=0$.

\noindent\emph{Lower bound for $Q_d$.}
If $Ax\geq2$, then
the deficit bound just proved gives $Q_d\rbra{x}\geq1/2$.
It remains to consider $0<Ax\leq2$.
With $z$ and $\theta$ as above, $0<z\leq2/d^2<1$ and
\begin{equation*}
  2\sqrt z\leq\theta=2\arcsin\sqrt z\leq\pi\sqrt z.
\end{equation*}
Thus $j\theta\leq\pi/\sqrt2<\pi$ whenever
$1\leq j\leq\floor{d/2}$.
Apply $1-\cos y\geq2y^2/\pi^2$ for $\abs{y}\leq\pi$ to these
terms in \cref{eq:fejer-cosine-expansion}.
Since every discarded term in $1-F_d\rbra{x}$ is nonnegative, $d-j\geq d/2$ for $1\leq j\leq\floor{d/2}$, $\floor{d/2}\geq d/3$, and $\sum_{j=1}^tj^2\geq t^3/3$ for every integer $t\geq1$, we obtain
\[
  1-F_d\rbra*{x}=\frac2{d^2}\sum_{j=1}^{d-1}\rbra*{d-j}\rbra*{1-\cos\rbra*{j\theta}}\geq\frac{2\theta^2}{d\pi^2}\sum_{j=1}^{\floor{d/2}}j^2\geq\frac{2d^2\theta^2}{81\pi^2}\geq\frac8{81\pi^2}d^2z\geq\frac1{128}Ax.
\]
Finally, since $F_d\rbra{x}\geq0$, we have
$Q_d\rbra{x}\geq1-F_d\rbra{x}$.  This proves
(ii) for $x>0$.  The claim at $x=0$ follows from
$Q_d\rbra{0}=0$.
\end{proof}

Using the coefficients in \cref{prop:fejer-properties}, define the
statistic $h_d:\Z_{\geq0}\to\R$ by
\begin{equation}
\label{eq:fejer-statistic-definition}
  h_d\rbra*{j}=\begin{cases}
    0,&j=0,\\
    1-j!f_{d,j}/\lam^j,&1\leq j\leq d-1,\\
    1,&j\geq d.
  \end{cases}
\end{equation}
The notation for $h_d$ also suppresses its dependence on $A$.
Set $\norm{h_d}_\infty=\sup_{j\geq0}\abs{h_d\rbra{j}}$.

\begin{lemma}[Poisson representation]
\label{lem:fejer-poisson-representation}
Fix $A>0$ and an odd integer $d\geq3$.
For every $x\geq0$ and $N_x\sim\operatorname{Poi}\rbra*{\lam x}$,
\begin{equation*}
  \Ex\sbra*{h_d\rbra*{N_x}}=Q_d\rbra*{x}.
\end{equation*}
Moreover, $\norm{h_d}_\infty\leq1+3^d$.
\end{lemma}

\begin{proof}
For every $x\geq0$, the Poisson probabilities and $f_{d,0}=1$ give
\begin{equation*}
  \Ex\sbra*{h_d\rbra{N_x}}
  =1-e^{-\lam x}\sum_{j=0}^{d-1}f_{d,j}x^j
   =1-e^{-\lam x}F_d\rbra{x}=Q_d\rbra{x}.
\end{equation*}
The coefficient bound in \cref{prop:fejer-properties}(i) gives
$\abs{h_d\rbra{j}}\leq1+3^d$ for $1\leq j\leq d-1$.
For the remaining values, $h_d\rbra{j}\in\cbra{0,1}$.
\end{proof}

\subsubsection{Reduction to bounded residual sets}
\label{subsec:initial-union-residual-sets}
We now define the tester \textsc{UnionSizeTest}$\rbra{\Omega,k,\delta}$ for \cref{thm:set-testing}.
It first reduces the problem to bounded residual sets and then runs a residual test on them.
It uses a sufficiently large absolute integer $k_0$.
If $k\geq\abs{\Omega}$, it accepts.
Otherwise, it uses amplified samples $\widetilde X_1,\widetilde X_2,\ldots$, each the union of $L=\ceil{1/\delta}$ new independent samples from $\calD$.
\begin{enumerate}[label=(\alph*)]
\item If $k<k_0$, it accepts exactly when $\abs{\widetilde X_1\cup\cdots\cup\widetilde X_{40(k+1)}}\leq k$.
\item Otherwise, it sets $b=\ceil{\log k}$ and $m_0=\ceil{128k/b}$, and forms the \emph{initial union} $S=\widetilde X_1\cup\cdots\cup\widetilde X_{m_0}$.
It rejects if $\abs{S}>k$, and otherwise sets $\ell=k-\abs{S}$.
\item If $\ell\leq k/\log k$, it accepts exactly when $\abs{\widetilde X_1\cup\cdots\cup\widetilde X_{m_0+40(\ell+1)}}\leq k$.
\item Otherwise, it runs the \emph{residual test}.
It sets $A=2048\ell$, lets $d$ be the largest odd integer with $d\leq\log_3\ell/64$, and sets $\lam=2A/d$.
It draws $N_{\mathrm P}\sim\operatorname{Poi}\rbra{\lam}$ and accepts if $N_{\mathrm P}>\ceil{100\lam}$.
Otherwise, for $j=1,\ldots,N_{\mathrm P}$, it sets $R_j=\widetilde X_{m_0+j}\setminus S$, replaced by $\varnothing$ if $\abs{R_j}>b$.
With $N_i=\abs{\cbra{j:i\in R_j}}$ for $i\in\Omega\setminus S$, it rejects exactly when $Z=\sum_{i\in\Omega\setminus S}h_d\rbra{N_i}\geq\ell+\ell/\rbra{2048b}$.
\end{enumerate}
\Cref{alg:set-tester} summarizes the tester.

\begin{algorithm}[htbp]
\caption{\textsc{UnionSizeTest}$\rbra{\Omega,k,\delta}$}
\label{alg:set-tester}
\begin{algorithmic}[1]
\Require A finite set $\Omega$, $k\geq1$, $0<\delta\leq1$, an absolute integer $k_0$, and independent samples from $\calD$ on $2^\Omega$
\Ensure \Accept\ or \Reject
\If{$k\geq\abs{\Omega}$} \State \Return \Accept \EndIf
\State $L\gets\ceil{1/\delta}$
\State When $\widetilde X_t$ is first requested,
  draw independent $X_{t,1},\ldots,X_{t,L}\sim\calD$
  and store $\widetilde X_t\gets\bigcup_{j=1}^{L}X_{t,j}$.
  Reuse this stored value whenever the same index $t$ is requested again.
\If{$k<k_0$} \State \Return \Accept\ iff $\abs*{\bigcup_{t=1}^{40\rbra{k+1}}\widetilde X_t}\leq k$ \EndIf
\State $b\gets\ceil{\log k}$
\State $m_0\gets\ceil{128k/b}$
\State $S\gets\bigcup_{t=1}^{m_0}\widetilde X_t$
\If{$\abs{S}>k$} \State \Return \Reject \EndIf
\State $\ell\gets k-\abs{S}$
\If{$\ell\leq k/\log k$}
  \State \Return \Accept\ iff $\abs*{\bigcup_{t=1}^{m_0+40\rbra{\ell+1}}\widetilde X_t}\leq k$
\EndIf
\State $A\gets2048\ell$
\State $d\gets\max\cbra{j\in2\Z+1\colon j\leq\log_3\ell/64}$
\State $\lam\gets2A/d$
\State $N_{\mathrm P}\sim\operatorname{Poi}\rbra{\lam}$
\If{$N_{\mathrm P}>\ceil{100\lam}$} \State \Return \Accept \EndIf
\For{$j=1,\ldots,N_{\mathrm P}$}
  \State $R_j\gets\widetilde X_{m_0+j}\setminus S$
  \If{$\abs{R_j}>b$} \State $R_j\gets\varnothing$ \EndIf
\EndFor
\ForAll{$i\in\Omega\setminus S$} \State $N_i\gets\abs{\cbra{j:i\in R_j}}$ \EndFor
\State $Z\gets\sum_{i\in\Omega\setminus S}h_d\rbra{N_i}$
\State \Return \Reject\ iff $Z\geq\ell+\ell/\rbra{2048b}$
\end{algorithmic}
\end{algorithm}

We first analyze the amplified samples.

\begin{proposition}[Amplification]
\label{prop:set-amplification}
Let $L=\ceil{1/\delta}$, and let $\widetilde X$ be the union of $L$ independent samples from $\calD$.
If $\calD$ is a YES instance with witness $T_*$, then $\widetilde X\subseteq T_*$ almost surely.
If $\calD$ is a NO instance, then $\Pr\sbra{\widetilde X\nsubseteq T}\geq1/2$ for every $T\subseteq\Omega$ with $\abs{T}\leq k$.
\end{proposition}

\begin{proof}
In a YES instance, every sample and hence its union lies in the witness
almost surely. In a NO instance, independence gives, for every
$T\subseteq\Omega$ with $\abs{T}\le k$,
\[
 \Pr\sbra*{\widetilde X\subseteq T}
 \le(1-\delta)^L\le e^{-\delta L}\le e^{-1}<1/2.
\]
Taking complements proves the NO bound.
\end{proof}

Call a set $T\subseteq\Omega$ with $\abs{T}\leq k$ \emph{good} if an independent amplified sample $\widetilde X$, the union of $L$ new independent samples from $\calD$, satisfies $\Pr\sbra{\abs{\widetilde X\setminus T}>b}\leq1/4$.

\begin{lemma}[Initial-union guarantee]
\label{lem:set-initial-union}
Let $k\geq2$, put $b=\ceil{\log k}$ and $m_0=\ceil{128k/b}$, and let $S$ be the union of $m_0$ independent copies of $\widetilde X$.
For every distribution $\calD$, with probability at least $19/20$,
either $\abs{S}>k$, or $S$ is good.
\end{lemma}

\begin{proof}
Let $S_0=\varnothing$ and $S_t=S_{t-1}\cup\widetilde X_t$, so $S=S_{m_0}$. Conditional on the preceding copies, put
\[
 q_t=\Pr\sbra*{\abs{\widetilde X\setminus S_{t-1}}>b\mid S_{t-1}},
\]
where $\widetilde X$ is an independent copy.
The events $\cbra{\abs{\widetilde X_t\setminus S_{t-1}}>b}$ need not be independent. We therefore generate the $t$th copy from an independent $V_t\sim\operatorname{Unif}[0,1]$, sampling from the conditional law given $\abs{\widetilde X\setminus S_{t-1}}>b$ when $V_t\leq q_t$ and from the conditional law given its complement otherwise.
Zero-probability branches can be assigned arbitrary laws.
Mixing these conditional laws with weights $q_t$ and $1-q_t$ recovers the original law at every step, so the copies still have their joint independent law and, almost surely,
\[
 \cbra{\abs{\widetilde X_t\setminus S_{t-1}}>b}=\cbra{V_t\leq q_t}.
\]

On the event $\calE$ that $\abs{S}\leq k$ and $S$ is not good, monotonicity under $S_{t-1}\subseteq S$ gives $q_t>1/4$ for every $t$. Hence every index $t$ with $V_t\leq1/4$ adds at least $b+1$ new elements to $S$. Since these increments are disjoint, every outcome in $\calE$ satisfies $(b+1)W\le\abs{S}\le k$ for
\[
 W=\abs{\cbra*{t\in\sbra{m_0}\colon V_t\leq1/4}}
 \sim\operatorname{Bin}(m_0,1/4),\qquad
 \Ex\sbra*{W}=m_0/4\geq32k/b,
\]
so $W\le k/(b+1)<\Ex\sbra{W}/2$ on $\calE$. By the Chernoff bound (\cref{fact:chernoff}),
\[
 \Pr\sbra*{\calE}
 \leq\Pr\sbra*{W\leq\Ex\sbra*{W}/2}
 \leq e^{-\Ex\sbra*{W}/8}
 \leq e^{-4k/b}\leq e^{-4}<1/20,
\]
where $b=\ceil{\log k}\leq k$ for integers $k\geq2$.
\end{proof}

\begin{proposition}[Small residual size bound]
\label{prop:set-small-cover}
Let $S'\subseteq\Omega$ be a union of copies of $\widetilde X$ with $\abs{S'}\leq k$, and put $\ell'=k-\abs{S'}$.
Then the check that accepts exactly when the union of $S'$ and $40(\ell'+1)$ further independent copies of $\widetilde X$ has at most $k$ elements accepts every YES instance with probability one and rejects every NO instance with probability at least $19/20$.
\end{proposition}

\begin{proof}
In a YES instance, every copy of $\widetilde X$ lies in a fixed set of
at most $k$ elements, so the union with $S'$ has size at most $k$.

Fix $S'$ in a NO instance. By \cref{prop:set-amplification}, as long as the current union has size at
most $k$, a new independent copy of $\widetilde X$ contains an element outside
the current union with conditional probability at least $1/2$.
Starting from $\abs{S'}=k-\ell'$, at most
$\ell'+1$ strict enlargements are needed for rejection.  Let $\tau$ be
the index of the first additional copy after which the union has more than
$k$ elements.  Since each copy is a strict enlargement with conditional
probability at least $1/2$ given the past, $\tau$ is stochastically
dominated by a sum of $\ell'+1$ independent geometric random variables
of mean two, so
\begin{equation*}
  \Ex\sbra*{\tau\mid S'} \leq 2\rbra*{\ell'+1},
  \qquad
  \Pr\sbra*{\tau>40\rbra*{\ell'+1}\mid S'}\leq\frac1{20}.
\end{equation*}
The probability bound follows from Markov's inequality (\cref{fact:markov}).
\end{proof}

On the remaining branch, we truncate the residual observations.
Fix a possible value of the initial union $S$ with $\abs{S}\leq k$, and let $\ell=k-\abs{S}$.
When $\ell>k/\log k$, the tester runs the residual test. For an independent copy of $\widetilde X$, define its residual by
\begin{equation*}
  R=\begin{cases}
    \widetilde X\setminus S,
      &\abs{\widetilde X\setminus S}\leq b,\\
    \varnothing,
      &\abs{\widetilde X\setminus S}>b.
  \end{cases}
\end{equation*}
Every outcome satisfies $R\subseteq\Omega\setminus S$ and $\abs{R}\leq b$.
The truncation replaces a large residual by $\varnothing$ instead of conditioning on $\abs{\widetilde X\setminus S}\leq b$, so for fixed $S$ the map from $\widetilde X$ to $R$ is deterministic.
Independent copies of $\widetilde X$ therefore produce independent copies of $R$.

\subsubsection{Analysis of the residual test}
\label{subsec:set-tester-analysis}

Throughout this subsubsection, $S$ is a fixed possible value of the initial union with $\abs{S}\leq k$ and $\ell=k-\abs{S}>k/\log k$.
All probabilities and expectations in the analysis of
the residual test are conditional on this fixed $S$.
The statistic $h_d$ uses the values of $A,d,\lam$ selected by the residual test,
as in \cref{eq:fejer-statistic-definition}.

We first analyze the statistic without the cap $N_{\mathrm P}\leq\ceil{100\lam}$ of the residual test, and then restore the cap.

For each $i\in\Omega\setminus S$, let
$a_i=\Pr\sbra{i\in R}$.  Since every residual set has size at most $b$,
\begin{equation}
\label{eq:set-marginal-bound}
  \sum_{i\in\Omega\setminus S}a_i
  =\Ex\sbra*{\abs{R}}\leq b.
\end{equation}

From now on, assume that $k$ is sufficiently large.
Since $\ell>k/\log k$ and $d\leq\log_3\ell/64$, we then have $d\geq3$, $3^d\leq\ell^{1/64}$,
$d=\Theta\rbra{\log\ell}$, and $\lam=\Theta\rbra{\ell/\log\ell}$.

Let $R_1,R_2,\ldots$ be independent copies of $R$, draw
$N_{\mathrm P}\sim\operatorname{Poi}\rbra{\lam}$ independently of them, and define
\[
  N_i=\abs{\cbra*{j:1\leq j\leq N_{\mathrm P},\ i\in R_j}}\quad\rbra*{i\in\Omega\setminus S},\qquad Z=\sum_{i\in\Omega\setminus S}h_d\rbra*{N_i}.
\]
Because $h_d\rbra{0}=0$, the value of $Z$ can be computed from the
elements appearing in the sampled residual sets.  Poisson thinning
gives the marginal distribution
$N_i\sim\operatorname{Poi}\rbra{\lam a_i}$.  The counts $N_i$ need not be
independent, but linearity of expectation and
\cref{lem:fejer-poisson-representation} give
\begin{equation}
\label{eq:set-statistic-expectation}
  \Ex\sbra*{Z}=\sum_{i\in\Omega\setminus S}Q_d\rbra*{a_i}.
\end{equation}

\begin{lemma}[Expectation gap]
\label{lem:set-expectation-separation}
Let $k$ be sufficiently large. In a YES instance, $\Ex\sbra{Z}\leq\ell$.  In a NO instance with good
$S$,
\begin{equation*}
  \Ex\sbra*{Z}\geq\ell+\frac{\ell}{1024b}.
\end{equation*}
\end{lemma}

\begin{proof}
In a YES instance, let $T_*$ be a witness. By \cref{prop:set-amplification}, every copy of $\widetilde X$ lies in $T_*$ almost surely, so $S\subseteq T_*$ and every residual lies in $T_*\setminus S$, a set of at most $\ell$ elements.
Hence at most $\ell$ marginals are positive. By \cref{prop:fejer-properties}(ii), $Q_d\rbra{0}=0$ and $Q_d\rbra{x}\leq1$ for $x\geq0$, so \cref{eq:set-statistic-expectation} gives $\Ex\sbra{Z}\leq\ell$.

Now consider a NO instance with good $S$.
Every $T\subseteq\Omega\setminus S$ with $\abs{T}\le\ell$ satisfies $\abs{S\cup T}\le k$, so \cref{prop:set-amplification} and the goodness of $S$ give
\begin{equation}\label{eq:set-residual-uncovered}
 \Pr\sbra*{R\nsubseteq T}
 \ge\Pr\sbra*{\widetilde X\nsubseteq S\cup T}
   -\Pr\sbra*{\abs{\widetilde X\setminus S}>b}
 \ge\frac12-\frac14=\frac14.
\end{equation}
List the elements with positive marginals as $i_1,i_2,\ldots$ with $a_{i_1}\geq a_{i_2}\geq\cdots>0$, and put $T_{\mathrm{top}}=\cbra{i_1,\ldots,i_\ell}$.
There are more than $\ell$ such elements, since otherwise their set would contain $R$ with probability one, contradicting \cref{eq:set-residual-uncovered}.
By \cref{eq:set-residual-uncovered} with $T=T_{\mathrm{top}}$, $R$ contains an element outside $T_{\mathrm{top}}$ with probability at least $1/4$. Consequently,
\begin{equation}
\label{eq:set-residual-tail-mass}
  \sum_{j>\ell}a_{i_j}
  =\Ex\sbra*{\abs{R\setminus T_{\mathrm{top}}}}
  \geq\Pr\sbra*{R\nsubseteq T_{\mathrm{top}}}
   \geq\frac14.
\end{equation}
Set $t=a_{i_{\ell+1}}>0$.  We show that the elements $i_j$ with $j>\ell$
contribute enough to exceed the total amount by which the first $\ell$
contributions may fall below one.

\noindent\emph{Case 1: $At<1$.}
For every $j>\ell$, we have $Aa_{i_j}\leq At<1$.
By the lower bound on $Q_d$ in \cref{prop:fejer-properties}(ii) and
\cref{eq:set-residual-tail-mass},
\begin{equation*}
  \Ex\sbra*{Z}
  \geq\sum_{j>\ell}Q_d\rbra*{a_{i_j}}
   \geq\frac{A}{128}\sum_{j>\ell}a_{i_j}
   \geq\frac{A}{512}=4\ell.
\end{equation*}
Since $b\geq1$, this gives the claimed expectation gap.

\noindent\emph{Case 2: $At\geq1$.}
Since $At\geq1$, every $0\leq x\leq t$ satisfies
$\min\cbra{Ax,1}\geq x/t$.  Hence the tail contributes
\begin{equation*}
  \sum_{j>\ell}Q_d\rbra*{a_{i_j}}
  \geq\frac1{128t}\sum_{j>\ell}a_{i_j}
   \geq\frac1{512t}.
\end{equation*}
Each of the first $\ell$ marginals is at least $t$.  The bound on
$1-Q_d$ in \cref{prop:fejer-properties}(ii) therefore gives
\begin{equation*}
  \sum_{j=1}^{\ell}\rbra*{1-Q_d\rbra*{a_{i_j}}}
  \leq\sum_{j=1}^{\ell}\frac1{Aa_{i_j}}
   \leq\frac{\ell}{At}.
\end{equation*}
Subtracting this deficit from $\ell$ and adding the tail contribution
yields
\begin{equation*}
  \Ex\sbra*{Z}
  \geq\ell+\frac{1/512-\ell/A}{t}
   =\ell+\frac3{2048t}.
\end{equation*}
By \cref{eq:set-marginal-bound}, $\rbra{\ell+1}t\leq b$, so
\begin{equation*}
  \Ex\sbra*{Z}
  \geq\ell+\frac{3\rbra*{\ell+1}}{2048b}
   \geq\ell+\frac{\ell}{1024b},
\end{equation*}
which proves the claimed gap.
\end{proof}

We bound the variance by treating each residual set as one Poisson point.

\begin{lemma}[Poisson Poincar\'e inequality, finite-space form
  {\cite[Theorem~18.7]{LP17}}]
\label{lem:set-poisson-variance}
Let $r\geq1$ be an integer, let $K_j\sim\operatorname{Poi}\rbra{\lam_j}$
be independent with $\lam_j\geq0$, and set $K=\rbra{K_1,\ldots,K_r}$.
For the $j$th coordinate vector $e_j$ and every bounded function
$g:\Z_{\geq0}^r\to\R$,
\[
 \operatorname{Var}\sbra*{g\rbra{K}}
 \leq\sum_{j=1}^r\lam_j
      \Ex\sbra*{\rbra*{g\rbra{K+e_j}-g\rbra{K}}^2}.
\]
\end{lemma}

\begin{lemma}[Concentration of the residual statistic]
\label{lem:set-statistic-concentration}
Let $k$ be sufficiently large, suppose that $k/\log k<\ell\leq k$, and let $Z$ be computed from all $N_{\mathrm P}$ residual sets, without the cap $N_{\mathrm P}\leq\ceil{100\lam}$.
Then the rule that rejects exactly when $Z\geq\ell+\ell/\rbra{2048b}$ errs with probability at most $1/20$ on every YES instance and on every NO instance for which $S$ is good.
\end{lemma}

\begin{proof}
By
\cref{lem:fejer-poisson-representation} and the choice of $d$,
\begin{equation*}
  \norm{h_d}_\infty
  \leq1+3^d
   \leq1+\ell^{1/64}
   \leq\ell^{1/16}
\end{equation*}
for all sufficiently large $\ell$.

Enumerate the possible residual sets as $B_1,\ldots,B_r$.
Poisson splitting gives independent multiplicities
$K_j\sim\operatorname{Poi}\rbra{\lam\Pr\sbra{R=B_j}}$,
whose means sum to $\lam$. Write $K=\rbra{K_1,\ldots,K_r}$ and $Z=g\rbra{K}$.
Adding one residual set changes
at most $b$ coordinate counts, and each summand $h_d\rbra{N_i}$ changes
by at most $2\norm{h_d}_\infty$. Thus, for every $j\in\sbra{r}$, $\abs{g\rbra{K+e_j}-g\rbra{K}}\leq2b\norm{h_d}_\infty$, where $e_j$ is the $j$th coordinate vector of $\Z_{\geq0}^r$.
Since $\abs{Z}\leq\abs{\Omega\setminus S}\norm{h_d}_\infty$,
the Poisson Poincar\'e inequality (\cref{lem:set-poisson-variance}) gives
\begin{equation}
\label{eq:set-statistic-variance}
  \operatorname{Var}\sbra*{Z}
  \leq4\lam b^2\norm{h_d}_\infty^2
   \leq4\lam b^2\ell^{1/8}.
\end{equation}

By \cref{lem:set-expectation-separation}, the stated rejection threshold is at distance at least
$\ell/\rbra{2048b}$ from the mean in each of the two cases.  Chebyshev's
inequality and \cref{eq:set-statistic-variance} therefore bound the
error probability by
\begin{equation*}
  \frac{4\lam b^2\ell^{1/8}}{\rbra{\ell/\rbra{2048b}}^2}
  =4\cdot2048^2\frac{\lam b^4}{\ell^{15/8}}.
\end{equation*}
Since $\lam=\Theta\rbra{\ell/\log\ell}$, and since $k/\log k<\ell\leq k$ gives $b=\Theta\rbra{\log\ell}$ and $\ell\to\infty$ as $k\to\infty$, this bound is $O\rbra{\log^3\ell/\ell^{7/8}}$ and hence at most $1/20$ for all sufficiently large $k$.
\end{proof}

\begin{proposition}[Guarantee of the residual test]
\label{prop:set-residual-algorithm}
There is an absolute integer $k_0$ with the following property.
Let $k\geq k_0$, let $S$ be the initial union, and suppose that $\ell=k-\abs{S}>k/\log k$.
Then, conditional on $S$, the residual test uses $O\rbra{k/\log k}$ additional copies of $\widetilde X$.
It accepts every YES instance with probability at least $19/20$ and rejects every NO instance with good $S$ with probability at least $47/50$.
\end{proposition}

\begin{proof}
Choose an absolute $k_0\geq16$ such that $d\geq3$ and
\cref{lem:set-statistic-concentration} applies for every $k\geq k_0$
and $k/\log k<\ell\leq k$.

The actual algorithm first draws $N_{\mathrm P}$ and requests further
copies of $\widetilde X$ only if
$N_{\mathrm P}\leq N_{\max}=\ceil{100\lam}$.  Markov's inequality (\cref{fact:markov})
gives
\begin{equation*}
  \Pr\sbra*{N_{\mathrm P}>N_{\max}}
  \leq\frac{\lam}{N_{\max}}\leq\frac1{100}.
\end{equation*}
When $N_{\mathrm P}\leq N_{\max}$, the residual test follows the
decision rule analyzed in \cref{lem:set-statistic-concentration}.
When $N_{\mathrm P}>N_{\max}$, it returns \Accept.
The latter event has probability at most $1/100$, so it adds at most
$1/100$ to the NO error and no error in the YES case.
Thus the NO error is at most $1/20+1/100=3/50$, and the YES error
is at most $1/20$.

Finally, $\lam=\Theta\rbra{\ell/\log\ell}$ gives
$N_{\max}=O\rbra{\ell/\log\ell}$.  Since
$k/\log k<\ell\leq k$, this is also $O\rbra{k/\log k}$.
\end{proof}

\subsubsection{Proof of the upper bound}\label{subsec:cover-upper-independent-proof}
We now combine \cref{prop:set-amplification,lem:set-initial-union,prop:set-small-cover,prop:set-residual-algorithm} to prove \cref{thm:set-testing}.

\begin{proof}[Proof of \cref{thm:set-testing}]
\noindent\emph{Algorithm.}
If $k\geq\abs{\Omega}$, then $T=\Omega$ satisfies the YES condition for
every distribution, so accepting without samples is correct.
From now on, assume that $k<\abs{\Omega}$.
Run \textsc{UnionSizeTest} with the absolute integer $k_0$ from
\cref{prop:set-residual-algorithm}.
It forms each $\widetilde X$ as the union of a block of $L=\ceil{1/\delta}$
original samples.  Unions of disjoint blocks are independent copies of
$\widetilde X$, and by \cref{prop:set-amplification} they satisfy the
amplified YES and NO guarantees used below.

\noindent\emph{YES instances.}
In a YES instance, let $T_*$ be a witness. By \cref{prop:set-amplification}, every copy of $\widetilde X$ is contained in $T_*$ almost surely.
The algorithm for $k<k_0$ therefore accepts with probability one.  For
$k\geq k_0$, the initial union satisfies $S\subseteq T_*$ and hence
$\abs{S}\leq k$.  The branch $\ell\leq k/\log k$ accepts with probability
one by \cref{prop:set-small-cover} with $S'=S$, and the branch
$\ell>k/\log k$ accepts with probability at least $19/20$ by
\cref{prop:set-residual-algorithm}.  Thus the acceptance probability is
at least $19/20$ in every YES instance.

\noindent\emph{NO instances.}
In a NO instance with $k<k_0$, \cref{prop:set-small-cover} with $S'=\varnothing$ bounds the error by $1/20$.
Now suppose that $k\geq k_0$.  If the initial union has more
than $k$ elements, the algorithm rejects immediately.  Otherwise,
by \cref{lem:set-initial-union}, the event that $\abs{S}\leq k$ and $S$ is
not good has probability at most $1/20$.

Fix any good $S$.  By \cref{prop:set-small-cover} with $S'=S$, the conditional
error is at most $1/20$ when $\ell\leq k/\log k$.  By
\cref{prop:set-residual-algorithm}, it is at most $3/50$ when
$\ell>k/\log k$.  Consequently, the error probability is at most $1/20+3/50=11/100<1/3$.
Thus the algorithm succeeds with probability at least $2/3$ in both
cases.

\noindent\emph{Sample complexity.}
The sample counts of the tester show
that every execution uses $O\rbra{k/\log\rbra{ek}}$ copies of
$\widetilde X$.  For $k<k_0$, this bound follows from
$40\rbra{k+1}=O\rbra{k}$ and $\log\rbra{ek}\leq\log\rbra{ek_0}$, since $k_0$ is absolute.
For $k\geq k_0$, it follows from $m_0=O\rbra{k/\log k}$, together with $40\rbra{\ell+1}=O\rbra{k/\log k}$ when $\ell\leq k/\log k$ and with \cref{prop:set-residual-algorithm} when $\ell>k/\log k$.
Since each copy is the union of
$L=\ceil{1/\delta}\leq2/\delta$ original samples, the total sample
complexity is $O\rbra*{k/\rbra{\delta\log\rbra{ek}}}$.
\end{proof}

\subsection{Lower bound from independent samples}\label{subsec:cover-lower-independent}
We now show that the sample bound of \cref{thm:set-testing} is optimal.
The singleton-set specialization of \cref{prob:set-testing} is classical support size testing.
For an integer $s\geq1$ and $p\in\Delta\rbra{\sbra{M}}$, the largest probability mass carried by at most $s$ labels is
\begin{equation*}
  \mass_s\rbra*{p}=\max_{S\subseteq\sbra*{M}\colon\abs*{S}\le s}p\rbra*{S}.
\end{equation*}
The total-variation distance from $p$ to distributions supported on at most $s$ labels is $1-\mass_s\rbra{p}$.
Indeed, every distribution supported on a set $S$ has distance at least $p(\sbra M\setminus S)$ from $p$.
Conditioning $p$ on a set $S$ with $\abs S\le s$ and $p(S)=\mass_s(p)$ gives distance exactly $1-\mass_s(p)$.
Thus the NO condition below says that $p$ is $\delta$-far in total-variation distance from every distribution with support size at most $s$.

\begin{problem}[Support Size Testing]
\label{prob:classical-support-testing}
  Fix integers $M\geq2$ and $1\leq s\leq M$,
  and let $0<\delta<1$.
  Given independent samples from an unknown distribution
  $p \in \Delta\rbra{\sbra{M}}$, distinguish between the
  following two cases with probability at least $2/3$:
  \begin{itemize}
    \item (YES) $\abs*{\supp\rbra*{p}} \le s$.
    \item (NO) $\mass_s\rbra*{p} \le 1-\delta$.
  \end{itemize}
\end{problem}

Applied to singleton samples, \cref{thm:set-testing} gives the following upper bound.

\begin{corollary}[Classical support size upper bound]
\label{cor:classical-support-upper-bound}
There is a randomized tester for \cref{prob:classical-support-testing} that succeeds with probability at least $2/3$ and uses $O(s/(\delta\log(es)))$ samples, with an absolute implicit constant.
\end{corollary}

\begin{proof}
Map each sample $Y\sim p$ to the singleton $X=\cbra{Y}$.
If $\abs{\supp\rbra{p}}\le s$, then $X\subseteq\supp\rbra{p}$ almost surely.
If $\mass_s\rbra{p}\le1-\delta$, then $\Pr\sbra{X\nsubseteq T}=1-p\rbra{T}\geq\delta$ for every $T\subseteq\sbra{M}$ with $\abs{T}\leq s$.
Hence \cref{thm:set-testing} applies with $\Omega=\sbra{M}$ and $k=s$.
\end{proof}

Since singleton samples give instances of \cref{prob:set-testing}, it suffices to prove a lower bound for \cref{prob:classical-support-testing} in the independent-sample model of \cref{subsec:independent-samples}.

\begin{theorem}[Classical support size lower bound,
adapted from the proof of {\cite[Proposition~1]{WY19}}]
\label{thm:classical-support-lb}
Fix a sufficiently large integer $s$, $0<\delta<1/32$, and a constant $c>1/2$. Any randomized tester that, given independent samples from a distribution on $\sbra{2s}$, distinguishes the two cases of \cref{prob:classical-support-testing} with success probability at least $c$ requires $\Omega(s/(\delta\log s))$ samples, where the implicit constant depends only on $c$.
\end{theorem}

Together with \cref{cor:classical-support-upper-bound}, this gives the tight
bound $\Theta(s/(\delta\log s))$ for $M\geq2s$, $0<\delta<1/32$, and
sufficiently large $s$.

This bound is stated in~\cite{PH26,Wang26} as a consequence of~\cite{VV11,WY19}. For completeness, we prove \cref{thm:classical-support-lb} using the moment-matching approach of~\cite{WY19}.
We first construct priors with separated support sizes and close sample laws.
We then dilute them by a common distribution on a disjoint support.
For an integer $m\geq1$, a prior $\pi$ on $\Delta\rbra{\sbra{m}}$, and an integer $t\geq0$, write
$\calQ_{\pi}^{\rbra{t}}=\Ex_{p\sim\pi}\sbra*{p^{\ot t}}$,
where $p$ is drawn once before the $t$ independent samples.
For a random variable $V\geq0$, write $\Ex\sbra{\operatorname{Poi}\rbra{V}}$ for the Poisson mixture, that is, the law of a count $N$ that, conditional on $V$, has law $\operatorname{Poi}\rbra{V}$.

\begin{lemma}[Moment-matched unit-mean variables
  {\cite[proof of Proposition~1]{WY19}}]
\label{lem:unit-mean-moment-matching}
For every integer $d\geq1$, $1<a<2$, and $\lam>a$, there are random
variables $X_0,X_1$ supported on $\cbra{0}\cup[a,\lam]$ such that
$\Ex\sbra{X_0}=\Ex\sbra{X_1}=1$,
$\Ex\sbra{X_0^j}=\Ex\sbra{X_1^j}$ for $1\leq j\leq d$, and
\[
 \Pr\sbra*{X_1>0}-\Pr\sbra*{X_0>0}
 =\frac{\rbra*{1+\sqrt{a/\lam}}^2}{a}
  \rbra*{\frac{1-\sqrt{a/\lam}}{1+\sqrt{a/\lam}}}^d.
\]
\end{lemma}

\begin{lemma}[Moment-matched Poisson mixtures
  {\cite[Lemma~4]{WY19}}]
\label{lem:poisson-moment-matching}
Let $d\geq1$ be an integer and $\lam>0$.
If random variables $V_0,V_1$ take values in $[0,\lam]$ and satisfy
$\Ex\sbra{V_0^j}=\Ex\sbra{V_1^j}$ for $1\leq j\leq d$, then
\[
 \dtv{\Ex\sbra*{\operatorname{Poi}\rbra*{V_0}}}
     {\Ex\sbra*{\operatorname{Poi}\rbra*{V_1}}}
 \leq\rbra*{\frac{e\lam}{2d}}^d.
\]
\end{lemma}

\begin{lemma}[Classical hard priors]
\label{lem:classical-hard-priors}
There is an absolute constant $c_0>0$ such that, for every sufficiently large integer $m$, there are an integer $0\leq s_0<m$ and priors $\pi_0,\pi_1$ on $\Delta\rbra{\sbra{m}}$ with the following properties.
\begin{itemize}
  \item For $p\sim\pi_0$, almost surely
  $\abs{\supp\rbra{p}}\leq s_0$.
  \item For $p\sim\pi_1$, almost surely
  $\abs{\supp\rbra{p}}\geq s_0+m/4$ and
  $p\rbra{i}\geq1/m$ for every $i\in\supp\rbra{p}$.
\end{itemize}
For $t_\star=\floor{c_0m/\log m}$,
\begin{equation*}
  \dtv{\calQ_{\pi_0}^{\rbra{t_\star}}}
      {\calQ_{\pi_1}^{\rbra{t_\star}}}
  \leq\frac1{10}.
\end{equation*}
\end{lemma}

\begin{proof}
\noindent\emph{Moment-matched variables.}
Set $d=\floor{\ln m}$, $\lam=16\rbra*{\ln m}^2$,
and $\nu=\rbra*{\lam/m}^{1/4}$.
For sufficiently large $m$, we have $d\geq2$,
$0<\nu\leq1/5$, and $1+\nu<\lam$.
Apply \cref{lem:unit-mean-moment-matching} with $a=1+\nu$ to obtain
$X_0,X_1$. Put $\eta_m=\sqrt{a/\lam}$ and
$r_m=\rbra{1-\eta_m}/\rbra{1+\eta_m}$. Their nonzero-mass gap is
\begin{equation*}
  g_m=\Pr\sbra*{X_1>0}-\Pr\sbra*{X_0>0}
  =\frac{\rbra{1+\eta_m}^2}{a}r_m^d.
\end{equation*}
We have $1<a\leq6/5$, $d\eta_m\leq\sqrt a/4$,
and $1-r_m\leq2\eta_m$.
Bernoulli's inequality gives
\begin{equation*}
  g_m\geq\frac{r_m^d}{a}
  \geq\frac{1-2d\eta_m}{a}
  \geq\frac1a-\frac1{2\sqrt a}
  \geq\frac13.
\end{equation*}

\noindent\emph{Conditioning and normalization.}
For $b\in\cbra{0,1}$, take independent copies $X_{b,1},\ldots,X_{b,m}$ of $X_b$ and set
\begin{equation*}
 v^{\rbra{b}}=\frac1m\rbra*{X_{b,1},\ldots,X_{b,m}},
 \qquad Z_b=\abs{\supp\rbra*{v^{\rbra{b}}}},
 \qquad \mu_b=\Ex\sbra*{Z_b}.
\end{equation*}
Then $\mu_1-\mu_0=mg_m$. Define
\begin{equation*}
 \calE_b=\cbra*{
 \abs{\norm{v^{\rbra{b}}}_1-1}\leq\nu,
 \ \abs{Z_b-\mu_b}\leq m^{2/3}}.
\end{equation*}
Independence and $X_b^2\leq\lam X_b$ give
$\Ex\sbra*{\norm{v^{\rbra{b}}}_1}=1$,
$\operatorname{Var}\sbra*{\norm{v^{\rbra{b}}}_1}\leq\lam/m$,
and $\operatorname{Var}\sbra{Z_b}\leq m/4$, since $Z_b$ is binomial with $m$ trials.
Chebyshev's inequality (\cref{fact:chebyshev}) and a union bound therefore yield, for sufficiently large $m$,
\begin{equation*}
 1-\Pr\sbra*{\calE_b}
 \leq\frac{\lam}{m\nu^2}+m^{-1/3}
 =\nu^2+m^{-1/3}\leq\frac1{50}.
\end{equation*}
Let $\wh\pi_b$ be the law of $v^{\rbra{b}}$ conditioned on $\calE_b$, and let $\pi_b$ be the law of $p=v/\norm{v}_1$ for $v\sim\wh\pi_b$.
Since $1-\nu\leq\norm{v}_1\leq1+\nu$ and each nonzero $v_i\geq(1+\nu)/m$, normalization preserves support and gives $p\rbra{i}\geq1/m$ on that support.
Set $s_0=\floor{\mu_0+m^{2/3}}$. On $\calE_0$, we have $Z_0\leq s_0$. On $\calE_1$,
\begin{equation*}
 Z_1-s_0\geq mg_m-2m^{2/3}
 \geq\frac m3-2m^{2/3}\geq\frac m4
\end{equation*}
for sufficiently large $m$. Since $\Pr\sbra{\calE_1}>0$ and $Z_1\leq m$, we also have $0\leq s_0\leq3m/4<m$.

\noindent\emph{Poissonized sample laws.}
Take $c_0=10^{-5}$ and set $\tau=8t_\star>0$. Let $\calN_b$ be the law obtained by drawing $v\sim\wh\pi_b$ and then, conditional on $v$, independent counts $N_i\sim\operatorname{Poi}\rbra*{\tau v_i}$.
Without conditioning on $\calE_b$, the count law is
$\bigotimes_{i=1}^m\Ex\sbra*{\operatorname{Poi}\rbra*{\tau X_b/m}}$.
Conditioning changes it by at most $1-\Pr\sbra{\calE_b}\leq1/50$ in total variation.
The variables $\tau X_b/m$ lie in $[0,\tau\lam/m]$ and match moments
through degree $d$. By \cref{lem:poisson-moment-matching},
\begin{equation*}
 \dtv{\Ex\sbra*{\operatorname{Poi}\rbra*{\tau X_0/m}}}
     {\Ex\sbra*{\operatorname{Poi}\rbra*{\tau X_1/m}}}
 \leq\rbra*{\frac{e\tau\lam}{2md}}^d.
\end{equation*}
Since $\tau\leq8c_0m/\log m$ and $d\geq\ln m/2$ for $m\geq e^2$, we have $e\tau\lam/(2md)\leq128e\ln2\cdot c_0<e^{-4}$. Summing coordinate distances and both conditioning errors, and using $d\geq\ln m-1$, gives
\begin{equation*}
 \dtv{\calN_0}{\calN_1}
 \leq\frac2{50}+m\rbra*{\frac{e\tau\lam}{2md}}^d
 \leq\frac2{50}+\frac{e^4}{m^3}\leq\frac3{50}.
\end{equation*}

\noindent\emph{Fixed sample count.}
For sufficiently large $m$, $t_\star\geq\ln50$.
Conditional on $v$, $N_{\mathrm P}=\sum_iN_i$ is Poisson with mean $8t_\star\norm{v}_1$.
Using $\Ex\sbra*{2^{-N_{\mathrm P}}\mid v}=\exp\rbra*{-4t_\star\norm{v}_1}$ and $\norm{v}_1\geq1-\nu\geq1/2$, Markov's inequality (\cref{fact:markov}) gives
\begin{equation*}
 \Pr\sbra*{N_{\mathrm P}<t_\star\mid v}
 \leq2^{t_\star}\exp\rbra*{-4t_\star\norm{v}_1}
 \leq e^{-t_\star}\leq\frac1{50}.
\end{equation*}
Define a randomized map $\calK$ from histograms to $\sbra{m}^{t_\star}$ by uniformly ordering the histogram and taking its first $t_\star$ entries when $N_{\mathrm P}\geq t_\star$, and outputting a fixed sequence otherwise.
The map is independent of $b$. Given $v$ and $N_{\mathrm P}=h$, the histogram is $\operatorname{Multinomial}(h,p)$, so its uniform ordering has law $p^{\ot h}$.
Couple this ordering to the first $N_{\mathrm P}$ entries of a sequence of independent samples from $p$, drawn independently of $N_{\mathrm P}$ conditional on $v$. Then $\calK$ agrees with the first $t_\star$ samples unless $N_{\mathrm P}<t_\star$.
Averaging over $v\sim\wh\pi_b$ gives
\begin{equation*}
 \dtv{\calK\rbra*{\calN_b}}
     {\calQ_{\pi_b}^{\rbra{t_\star}}}\leq\frac1{50}.
\end{equation*}
Contraction under $\calK$ and the triangle inequality now yield
\begin{equation*}
 \dtv{\calQ_{\pi_0}^{\rbra{t_\star}}}
     {\calQ_{\pi_1}^{\rbra{t_\star}}}
 \leq\frac1{50}+\frac3{50}+\frac1{50}
 =\frac1{10}.\qedhere
\end{equation*}
\end{proof}

\begin{proof}[Proof of \cref{thm:classical-support-lb}]
All asymptotic bounds below are uniform in $\delta$.
Take $s_0,\pi_0,\pi_1$ from \cref{lem:classical-hard-priors} with $m=s$.
Set $\alpha=4\delta$, let $u$ be uniform on $\cbra{s+1,\ldots,2s-s_0}$, and, for $b\in\cbra{0,1}$, let $\wt\pi_b$ be the law of
\begin{equation*}
 \wt p=\alpha p+(1-\alpha)u,\qquad p\sim\pi_b,
\end{equation*}
where both distributions are extended by zero to $\sbra{2s}$.

\noindent\emph{YES instances.}
The supports of $p$ and $u$ are disjoint, so
$\abs{\supp\rbra{\wt p}}=\abs{\supp\rbra{p}}+s-s_0$.
Under $\wt\pi_0$, this is at most $s$, as required by
\cref{prob:classical-support-testing}.

\noindent\emph{NO instances.}
Under $\wt\pi_1$, the support size is at least $s+s/4$ and every nonzero atom has mass at least $\alpha/s$.
On $\sbra{s}$ this follows from \cref{lem:classical-hard-priors}. On $\supp(u)$ it follows from
$(1-\alpha)/(s-s_0)\geq\alpha/s$, since $s-s_0\leq s$ and $1-\alpha\geq\alpha$.
Every set of at most $s$ labels therefore misses at least $s/4$ atoms, giving
\begin{equation*}
 \mass_s\rbra{\wt p}\leq1-\frac{s}{4}\cdot\frac{\alpha}{s}=1-\delta.
\end{equation*}
Thus $\wt\pi_1$ consists of NO instances.

\noindent\emph{Indistinguishability.}
By Hoeffding's inequality (\cref{fact:hoeffding}), majority voting over a constant number of independent repetitions amplifies any fixed success probability greater than $1/2$ to $2/3$. It therefore suffices to consider testers with success probability at least $2/3$.
Let $t_\star>0$ be as in \cref{lem:classical-hard-priors} with $m=s$, and let $t$ be an integer with $0\leq t\leq t_\star/(20\alpha)$.
Map $\sbra{s}^{t_\star}$ to $\sbra{2s}^t$ by choosing $t$ independent $\mathrm{Bern}(\alpha)$ bits with sum $B$, filling the positions marked $1$ with the input samples in order, and filling the rest with independent samples from $u$. If $B>t_\star$, output a fixed sequence instead.
This map is independent of $p$ and $b$, and Markov's inequality (\cref{fact:markov}) gives
\begin{equation*}
 \Pr\sbra{B>t_\star}\leq\frac{t\alpha}{t_\star}\leq\frac1{20}.
\end{equation*}
Without the cutoff, independent samples from $p$ give output law $\wt p^{\ot t}$.
Using the same bits and samples couples the two procedures whenever $B\leq t_\star$.
Consequently, the image of $\calQ_{\pi_b}^{\rbra{t_\star}}$ is within $1/20$ of $\calQ_{\wt\pi_b}^{\rbra{t}}$.
Contraction and the triangle inequality give
\begin{equation*}
 \dtv{\calQ_{\wt\pi_0}^{\rbra{t}}}
     {\calQ_{\wt\pi_1}^{\rbra{t}}}
 \leq\frac1{20}+\frac1{10}+\frac1{20}=\frac15.
\end{equation*}
By \cref{subsec:independent-samples}, a tester using at most $t$ samples has average acceptance probabilities under $\wt\pi_0$ and $\wt\pi_1$ differing by at most $1/5$, whereas success probability $2/3$ on every promised input requires a difference of at least $1/3$.
Hence every such tester uses more than $t_\star/(20\alpha)$ samples.
Since $t_\star=\Theta(s/\log s)$ and $\alpha=4\delta$, this gives the lower bound $\Omega(s/(\delta\log s))$.
\end{proof}

\begin{corollary}[Lower bound for Union Size Testing from independent samples]\label{cor:cover-lower-independent}
Let $k$ be a sufficiently large integer, let $n\ge2k$, and let $0<\delta<1/32$.
Given independent samples from an unknown distribution $\calD$ on $2^{\sbra{n}}$, any tester for \cref{prob:set-testing} succeeding with probability at least $2/3$ requires $\Omega(k/(\delta\log k))$ samples.
\end{corollary}

\begin{proof}
Let $p$ be a distribution on $\sbra{2k}\subseteq\sbra{n}$, and feed the tester the singleton samples $X=\cbra{Y}$ with $Y\sim p$.

\noindent\emph{YES instances.}
If $\abs{\supp\rbra{p}}\leq k$, then $X\subseteq\supp\rbra{p}$ almost surely.

\noindent\emph{NO instances.}
If $\mass_k\rbra{p}\leq1-\delta$, then $\Pr\sbra{X\nsubseteq T}=1-p\rbra{T}\geq\delta$ for every $T\subseteq\sbra{n}$ with $\abs{T}\leq k$.

\noindent\emph{Completing the reduction.}
Thus the tester distinguishes the two cases of \cref{prob:classical-support-testing} on $\sbra{2k}$ with $s=k$, so the claimed sample bound follows from \cref{thm:classical-support-lb}.
\end{proof}

By \cref{thm:set-testing,cor:cover-lower-independent}, the sample complexity is $\Theta(k/(\delta\log k))$ for $n\geq2k$, $0<\delta<1/32$, and sufficiently large $k$.

\subsection{Upper bound from a reversible sampler}\label{subsec:cover-upper-reversible}
We now turn to the reversible-sampler model of \cref{subsec:reversible-sampler-access}, with the promise of \cref{prob:set-testing} on the ground set $\sbra{n}$.

\begin{theorem}[Upper bound for Union Size Testing from a reversible sampler]\label{thm:reversible-subset-upper}
Fix integers $1\le k<n$ and $0<\delta\le1$. Given ordinary queries to a reversible sampler $A$ for a distribution $\calD$ on $2^{\sbra{n}}$ and to its inverse $A^{\dagger}$, there is a quantum algorithm for \cref{prob:set-testing} on the ground set $\sbra{n}$ that succeeds with probability at least $9/10$ and uses $O(\sqrt{k/\delta})$ queries. The guarantee holds for all auxiliary outputs and unitary extensions of the sampler.
\end{theorem}

The tester first rejects distributions with $2k$ sufficiently heavy coordinates or with a dispersed tail.
In the remaining NO case, union sampling gives, with constant probability, either an oversized union or a set of observed coordinates whose minimum marginal determines the excess size of a single reciprocal gap test.
\Cref{subsec:reciprocal-gap,subsec:union-sampling} give the gap test and union sampling, and \cref{subsec:rank-finding,subsec:minimum-marginal} build the minimum-marginal estimator.
\Cref{subsec:cover-upper-proof} combines them to prove \cref{thm:reversible-subset-upper}.

Write
\begin{equation*}
 a_j=\Pr\sbra*{j\in X},\qquad p_\calD(S)=\Pr\sbra*{X\cap S\ne\varnothing}.
\end{equation*}
Also write $R_\calD(E)=\sum_{j\in E}a_j^{-1}$ when all marginals on $E$ are
positive. These are analysis quantities, whereas routines receive only explicit
bounds. Every subroutine call in this subsection uses independent randomness
and newly initialized work registers, so its guarantee holds conditionally on
all earlier calls.

Let $\Pi_S$ project onto the incidence strings of samples intersecting $S$.

\subsubsection{Reciprocal gap testing}\label{subsec:reciprocal-gap}

The gap test adapts the group-testing construction of~\cite{ABRW16}.

\begin{lemma}[Group-testing weights, adapted from
  {\cite[Claims~3.7 and~3.8]{ABRW16}}]
\label{lem:group-testing-weights}
For integers $n,k,d$ with $k,d\ge1$ and $n\ge2k+d$, there are nonnegative functions
$g_K,h_B:2^{\sbra{n}}\to\R$, indexed by $K,B\subseteq\sbra{n}$ with
$\abs{K}=k$ and $\abs{B}=k+d$, and a number $w>0$ with the following properties.
For every $S\subseteq\sbra{n}$, $g_K(S)=0$ if $S\cap K\ne\varnothing$,
and $h_B(S)=0$ if $\abs{S\cap B}\ne1$.
For every such $K,B$,
\[
 \sum_Sg_K(S)h_B(S)=1,\qquad
 \sum_Sg_K(S)^2=\sum_Sh_B(S)^2=w\le5\sqrt{\frac{k+d}{d}}.
\]
For every $j\in B$,
\[
 \sum_{S:S\cap B=\cbra{j}}h_B(S)^2=\frac{w}{k+d}.
\]
\end{lemma}

\begin{proof}
Put $m=n-k$. For $1\le s\le m-d+1$, choose the unique positive
$\alpha_s,\beta_s$ satisfying
\[
 \alpha_s\beta_s=\frac1{m\binom{m-1}{s-1}},\qquad
 \binom ms\alpha_s^2=(k+d)\binom{m-d}{s-1}\beta_s^2.
\]
For subsets with sizes in this range, set $g_K(S)=\alpha_{\abs{S}}$
when $S\cap K=\varnothing$, and set $h_B(S)=\beta_{\abs{S}}$ when
$\abs{S\cap B}=1$. Set all remaining values to zero.

For fixed $K,B$, put $\ell=\abs{B\setminus K}\ge d$.
Counting subsets disjoint from $K$ that meet $B$ once gives
\[
 \sum_Sg_K(S)h_B(S)
 =\frac\ell m\sum_{r=0}^{m-\ell}
    \frac{\binom{m-\ell}{r}}{\binom{m-1}{r}}
 =\ell\int_0^1(1-t)^{\ell-1}\,dt=1,
\]
where the middle equality uses $1/\rbra{m\binom{m-1}{r}}=\int_0^1t^r(1-t)^{m-1-r}\,dt$ and the binomial theorem.
The balancing condition makes both squared norms equal to
\[
 w=\sum_{s=1}^{m-d+1}
 \frac{\sqrt{\binom ms(k+d)\binom{m-d}{s-1}}}
      {m\binom{m-1}{s-1}}.
\]
Permutation symmetry within $B$ gives the stated weight for each $j\in B$.

For every $d\ge1$, bound the binomial ratio in each summand by
\[
 \frac{\binom{m-d}{s-1}}{\binom{m-1}{s-1}}
 \le\rbra*{1-\frac{s-1}{m-1}}^{d-1}.
\]
Separating $s=1$ and comparing the remaining decreasing terms with
their integral gives
\[
 w\le\sqrt{\frac{k+d}{d}}
    +2\sqrt{k+d}\int_0^{\pi/2}\cos^d t\,dt
 \le(1+\sqrt{2\pi})\sqrt{\frac{k+d}{d}}
 <5\sqrt{\frac{k+d}{d}},
\]
where $\cos t\le e^{-t^2/2}$ for $0\le t\le\pi/2$.
\end{proof}

The analysis of the gap test uses the following spectral bound.

\begin{lemma}[Effective spectral gap
  {\cite[Lemma~4.2]{LMRSS11}, \cite[Lemma~2.3]{ABRW16}}]
\label{lem:effective-gap}
Let $P,\Lambda$ be orthogonal projectors on the same finite-dimensional
Hilbert space and set $W=(2P-\Id)(2\Lambda-\Id)$.
For $0\le\theta\le\pi$, let $E_\theta$ project onto the eigenspaces of
$W$ with eigenphases in $[-\theta,\theta]$.
For every $z\in\ker\Lambda$,
\[
 \norm{E_\theta Pz}\le\frac\theta2\norm z.
\]
\end{lemma}

\begin{lemma}[Reciprocal gap test]\label{lem:gap-interface}
Let $A$ be a reversible sampler for a distribution $\calD$ on $2^{\sbra{n}}$.
For integers $1\le d\le k<n$, real $H_d>0$, and $0<\beta<1/3$, there is a procedure \GapTest$(A;n,k,d,H_d,\beta)$ that uses $O(\sqrt{H_d/d}\ln(1/\beta))$ ordinary queries to $A,A^{\dagger}$ on every input.
With error at most $\beta$, it accepts if some $K\subseteq\sbra{n}$ with $\abs{K}\le k$ has $\Pr\sbra{X\subseteq K}=1$, and rejects if some $B\subseteq\sbra{n}$ with $\abs{B}=k+d$ has positive marginals and $R_\calD(B)\le H_d$.
\end{lemma}

\begin{proof}[Proof of \cref{lem:gap-interface}]
Since $a_j\le1$, every $B$ in the rejection condition satisfies $\abs{B}=k+d\le n$ and $R_\calD(B)\ge k+d$.
Hence no such $B$ exists if $k+d>n$ or $H_d<k+d$, and the procedure returns \Accept\ in these cases.
Otherwise $H_d\ge k+d$, and we proceed as follows.
We adapt the two-reflection method in the proof of~\cite[Theorem~5.2]{ABRW16}.
Pad the ground set by never-present coordinates so that
$n_{\mathrm{pad}}\ge2k+d$, and choose $g_K,h_B,w$ from
\cref{lem:group-testing-weights} on $[n_{\mathrm{pad}}]$.

Put $Q_S=A^{\dagger}\Pi_SA$, and let $\calW$ be the space on which $A$ acts.
The walk below acts on $\C\ket t\oplus\bigoplus_{S\subseteq\sbra{n_{\mathrm{pad}}}}\ket S\ot\calW$, where the target state $\ket t$ is orthogonal to every copy $\ket S\ot\calW$.
Set $\gamma_0=1/(10\sqrt w)$, let $\psi_K=\ket t+\gamma_0\sum_Sg_K(S)\ket S\ket0$ for $\abs{K}=k$, and let $\Lambda$ be the orthogonal projector onto $\spn\cbra{\psi_K\colon\abs{K}=k}$.
The input-dependent projector and the walk are
\begin{equation*}
 \Gamma_A=\ket t\bra t+\bigoplus_S(\Id-Q_S),
 \qquad W_{\mathrm{gap}}=(2\Gamma_A-\Id)(2\Lambda-\Id).
\end{equation*}
For each subset label $S$,
\[
 \Id-2Q_S=A^{\dagger}(\Id-2\Pi_S)A.
\]
Since $\Lambda$ is known, each walk step uses two ordinary sampler queries.
Controlled walk steps apply $A$ and $A^{\dagger}$ unconditionally and control only the known reflections, acting as the identity on the target-state sector, so they use the same queries.

On a YES input, choose $K'\subseteq\sbra{n}$ with $\abs{K'}\le k$ and $\Pr\sbra{X\subseteq K'}=1$, and pad $K'$ to a set $K\subseteq\sbra{n_{\mathrm{pad}}}$ with $\abs{K}=k$.
Every term with $g_K(S)\ne0$ has $Q_S\ket0=0$.
Hence $\psi_K$ lies in the images of both projectors, so $W_{\mathrm{gap}}\psi_K=\psi_K$.
Its squared norm is $1+\gamma_0^2w=101/100$, so the initial state $\ket t$ has weight at least $100/101$ in the zero-phase space.

On a NO input choose, only for analysis, a set $B$ of $k+d$ coordinates with positive marginals and $R_\calD(B)\le H_d$. For nonzero $h_B(S)$, \cref{lem:group-testing-weights} gives $S\cap B=\cbra{j}$ for some $j\in B$, so $\norm{Q_S\ket0}^2=\norm{\Pi_SA\ket0}^2=p_\calD(S)\ge a_j>0$. For such $S$, put
\[
 \phi_S=Q_S\ket0/p_\calD(S),\qquad
 z_B=\ket t-\gamma_0^{-1}\sum_Sh_B(S)\ket S\phi_S.
\]
By \cref{lem:group-testing-weights} and $\braket{0}{\phi_S}=1$,
$\Lambda z_B=0$, while $(\Id-Q_S)\phi_S=0$ gives
$\Gamma_Az_B=\ket t$.
Using the coordinate weights in \cref{lem:group-testing-weights},
\begin{equation}\label{eq:gap-negative-witness}
 \norm{z_B}^2=1+\gamma_0^{-2}\sum_S\frac{h_B(S)^2}{p_\calD(S)}\le1+\frac{100w^2}{k+d}\sum_{j\in B}\frac1{a_j}\le1+2500\frac{H_d}{d}.
\end{equation}
Set $\theta=\sqrt{d/H_d}/250$, and let $E_\theta$ project onto the
eigenphases of $W_{\mathrm{gap}}$ in $[-\theta,\theta]$.
Since $\Lambda z_B=0$ and $\Gamma_Az_B=\ket t$,
\cref{lem:effective-gap} gives
\[
 \norm{E_\theta\ket t}^2\le\frac{\theta^2}{4}\norm{z_B}^2<\frac1{64},
\]
where the last inequality uses \cref{eq:gap-negative-witness} and $H_d\ge d$.
Use quantum phase estimation (\cref{lem:phase-estimation}) to estimate the eigenphase of
$W_{\mathrm{gap}}$ on $\ket t$ to circular accuracy $\theta/2$ and
accept if its circular distance from zero is at most $\theta/2$.
A constant number of repetitions followed by majority reduces the
error to $1/100$ on each eigencomponent of phase zero or phase
magnitude greater than $\theta$.
This accepts YES with probability at least $3/4$ and NO with
probability at most $1/4$.
By Hoeffding's inequality (\cref{fact:hoeffding}), the majority of $O(\ln(1/\beta))$ independent repetitions of the resulting test errs with probability at most $\beta$, using $O(\sqrt{H_d/d}\ln(1/\beta))$ ordinary queries.
\end{proof}

\subsubsection{Union sampling}\label{subsec:union-sampling}
For $K\subseteq\sbra{n}$, the probability that a sample is not contained in $K$ is $p_\calD(\overline K)=\Pr\sbra{X\nsubseteq K}$.

\begin{lemma}[Union sampling]\label{lem:union-interface}
Let $A$ be a reversible sampler for a distribution $\calD$ on $2^{\sbra{n}}$.
Given an integer $k\ge1$ and a real $H>0$, there is a procedure \UnionSample$(A;k,H)$ that returns the union $J$ of finitely many samples, each a set $X$ with $\calD(X)>0$, using at most $32\sqrt{(k+1)H}$ ordinary queries to $A,A^{\dagger}$.
For every fixed $K\subseteq\sbra{n}$ with $\abs{K}\le k$ and $p_\calD(\overline K)>0$,
\[
 \Pr\sbra*{J\subseteq K}\le\frac{1024}{H\,p_\calD(\overline K)}.
\]
For every fixed event $\calE\subseteq2^{\sbra{n}}$, the probability that at least one returned sample belongs to $\calE$ is at most $H\Pr\sbra{X\in\calE}$.
\end{lemma}

We prove \cref{lem:union-interface} using the randomized Grover schedule
of~\cite{BBHT98}.

\begin{proof}[Proof of \cref{lem:union-interface}]
Consider a known success predicate of probability $p$.
In block $r\ge0$, let $M_r=\ceil{(9/8)^r}$ and let $w_r=(2M_r+1)^2$ be its cost, a squared amplification length rather than a query count.
First make a direct sampling attempt.
If it fails, prepare the initial state again, choose $t$ uniformly from $\cbra{0,\ldots,M_r-1}$, apply $t$ Grover iterations, and measure success.
On success return the sampled subset. Otherwise proceed to the next block.
Count the full cost $w_r$ even if the direct attempt succeeds.

Every success returns a sample with the conditional law of $X$ given the predicate, so every returned sample $X$ satisfies $\calD(X)>0$.
Let $\theta=\arcsin\sqrt p$. By amplitude amplification~\cite[Eq.~(8)]{BHMT02}, $t$ Grover iterations succeed with probability $\sin^2((2t+1)\theta)$. If $0<p<1/2$ and $M_r\ge1/\sqrt p$, then $M_r\ge1/\sin(2\theta)$ because $\sin(2\theta)=2\sqrt{p(1-p)}>\sqrt p$, so the randomized attempt of block $r$ succeeds with probability at least $1/4$ by \cref{lem:randomized-grover}.
For $p\ge1/2$, the direct attempt already succeeds with probability at least $1/2$.
Hence, for $r_*=\ceil{\log_{9/8}(1/\sqrt p)}$, every block $r\ge r_*$ fails with probability at most $3/4$.
Summing the geometrically increasing costs against their reach
probabilities, using $(3/4)(9/8)^2<1$, gives
\begin{equation}\label{eq:union-cost}
 \Ex\sbra*{\sum_r w_r}\le\frac{1024}{p},
\end{equation}
where the sum runs over executed blocks.
The constant $1024$ follows from $w_r\le25(9/8)^{2r}$ by summing the two geometric series before and after $r_*$.
Block $u$ uses at most $2M_u<\sqrt{w_u}$ ordinary queries.
Since $2(9/8)^u\le\sqrt{w_u}\le5(9/8)^u$, we have $\sum_{u=0}^r\sqrt{w_u}\le45(9/8)^r\le32\sqrt{w_r}$.
Hence every deterministic prefix through block $r$ uses at most $32\sqrt{\sum_{u=0}^r w_u}$ queries, which we call the prefix bound.
For any $\calE\subseteq2^{\sbra{n}}$, block $r$ returns a sample in $\calE$ with probability at most $w_r\Pr\sbra{X\in\calE}$, which we call the event bound.
The inequality $\abs{\sin((2t+1)\theta)}\le(2t+1)\abs{\sin\theta}$ and the direct attempt give a factor at most $1+(2M_r-1)^2\le w_r$.
When $p=0$ there are no false successes.
The prefix and event bounds remain valid,
while \cref{eq:union-cost} is not invoked.

To run \UnionSample, initialize $J=\varnothing$, $h=0$, and $r=0$.
Repeatedly check whether $h+w_r>H$. If so, return $J$. Otherwise add
$w_r$ to $h$ and execute the whole block with success predicate $X\nsubseteq J$. On success, commit
$J\gets J\cup X$, return if $\abs{J}>k$, and otherwise reset $r=0$.
On failure, increment $r$. Thus only completed samples enter $J$,
and the bound $H$ forces termination even at uncovered probability zero.
A search stage runs from a reset of $r$ to the next commit or return, so at most $k+1$ stages occur.

Fix $K$ with $\abs{K}\le k$ and $p_\calD(\overline K)>0$, and run the uncapped process
until it first samples a set not contained in $K$. At each reached stage,
$J\subseteq K$ and $p_\calD(\overline J)\ge p_\calD(\overline K)$, so exact conditional sampling
gives probability $p_\calD(\overline K)/p_\calD(\overline J)$ of such a sample.
At most $k+1$ stages occur through the first such sample. Summing over
their initial unions $J$ gives
\[
 \Ex\sbra*{\sum_J\frac{p_\calD(\overline K)}{p_\calD(\overline J)}}=1.
\]
By \cref{eq:union-cost}, the expected total cost until this sample is at most
$1024/p_\calD(\overline K)$. If this total cost is at most $H$, the capped process
also samples a set not contained in $K$. Markov's inequality (\cref{fact:markov}) gives
$\Pr\sbra{J\subseteq K}\le1024/(H\,p_\calD(\overline K))$.

Conditioned on each history, the event bound bounds the probability of returning a sample in $\calE$ by the next cost times $\Pr\sbra{X\in\calE}$.
Summing over blocks gives $H\Pr\sbra{X\in\calE}$.
These bounds include the final sample even if it makes $\abs{J}>k$.

We finally bound the query count. If the stages have costs $w^{(1)},\ldots,w^{(s)}$, which sum to at most $H$, the prefix bound and Cauchy--Schwarz give
\[
 32\sum_{i=1}^s\sqrt{w^{(i)}}\le32\sqrt{(k+1)H}.
\]
\end{proof}

\subsubsection{Rank-sensitive finding}\label{subsec:rank-finding}

The minimum-marginal estimator below repeatedly needs a coordinate whose
marginal lies below a current threshold, at a cost governed by the
reciprocal mass of the candidate set rather than by the threshold.
We first give a coherent threshold test with this cost and then use it to
isolate such a coordinate inside a random subset.
In this subsubsection, $A$ is a finite-dimensional unitary with ordinary access to $A,A^{\dagger}$ and known initial state $\ket0$.
The index set $E$ is finite and labels known orthogonal projections $\Pi_j$, and $a_j=\norm{\Pi_jA\ket0}^2$.
For a reversible subset sampler, $\Pi_j$ projects onto the incidence strings of samples containing $j$, so $a_j=\Pr\sbra{j\in X}$.
No mutual orthogonality or commutativity is assumed.
\begin{lemma}[Robust reciprocal threshold test]\label{lem:threshold-or}
For $0<\tau\le1$, $R\ge1$, and $0<\beta<1/3$, there is a tester that uses $O(\sqrt R\ln(1/\beta))$ ordinary queries on every input and, with error at most $\beta$, returns \SMALL\ if some $a_j\le\tau/4$, and \LargeCase\ if every $a_j\ge\tau$ and $\sum_{j\in E}a_j^{-1}\le R$.
Moreover, for every family $\cbra{E_\iota}_{\iota\in\calI}$ of subsets of $E$ indexed by a finite set $\calI$, there is a unitary $V$ with $V\ket\iota\ket0=\ket\iota\ket{\xi_\iota}$ for every $\iota\in\calI$, such that measuring the answer register of $\ket{\xi_\iota}$ gives an answer with the above guarantee for the index set $E_\iota$. The unitary $V$, its inverse, and their controlled versions use $O(\sqrt R\ln(1/\beta))$ ordinary queries.
\end{lemma}
\begin{proof}
If $E=\varnothing$, return \LargeCase. Let $m=\abs{E}$. If $m>R$, the \LargeCase\ class is empty since $a_j\le1$, so return \SMALL. Assume $1\le m\le R$.
We build an operator $M_A$ with $\norm{M_A}_\infty\le1$ and a known unit
state $\chi$ such that $\chi$ has large weight on the small singular values
of $M_A$ in the \SMALL\ case and small weight in the \LargeCase\ case.
Phase estimation on a two-reflection walk then detects this weight.

Let $Q_j=A^{\dagger}\Pi_jA$, $v_j=Q_j\ket0$, and introduce a scalar state $\ket\omega$ orthogonal to the sampler workspace $\calW$. On $\calZ=\C\ket\omega\oplus\calW$, define
\begin{equation*}
 C_j=\ket{v_j}\bra\omega+\ket\omega\bra{v_j}.
\end{equation*}
Then $C_j\ket\omega=v_j$, $C_jv_j=a_j\ket\omega$, and $\norm{C_j}_\infty\le1$.
For known projections $\Pi_{\mathrm{in}}$ and $\Pi_{\mathrm{out}}$, called signal projections, a unitary $U$ encodes an operator $O$ with normalization $\alpha>0$ if $\Pi_{\mathrm{out}}U\Pi_{\mathrm{in}}=O/\alpha$.
This is the projected form of block encoding in~\cite{GSLW19}.
We verify the signal blocks of the linear combinations and products below directly.
Set $G_j=2Q_j-\Id$ and introduce a signal qubit $\mathsf a$.
The unitary
\[
 U_{Q_j}=(H_{\mathsf a}\ot\Id)
  \rbra*{\ket0\bra0_{\mathsf a}\ot\Id+\ket1\bra1_{\mathsf a}\ot G_j}
  (H_{\mathsf a}\ot\Id)
\]
satisfies $\bra0_{\mathsf a}U_{Q_j}\ket0_{\mathsf a}=Q_j$.
If $U_{O_1}^{(\mathsf a)}$ and $U_{O_2}^{(\mathsf b)}$ encode operators $O_1$ and $O_2$ with distinct signal registers $\mathsf a$ and $\mathsf b$, respectively, then
\[
 \bra{0}_{\mathsf a}\bra{0}_{\mathsf b}U_{O_1}^{(\mathsf a)}U_{O_2}^{(\mathsf b)}\ket{0}_{\mathsf a}\ket{0}_{\mathsf b}=O_1O_2.
\]
Thus composing $U_{Q_j}$ with a known encoding of $\ket0\bra0$,
using a separate signal register, encodes $Q_j\ket0\bra0$ with
normalization one. If $U_{\mathrm{enc}}$ is this encoding, the unitary
$\rbra*{\begin{smallmatrix}0&U_{\mathrm{enc}}\\ U_{\mathrm{enc}}^{\dagger}&0\end{smallmatrix}}$
encodes its Hermitian dilation.
The Hermitian dilation leaves the subspace
$\calW\oplus\spn\cbra{\ket0}$ invariant. Identifying $\ket0$ in
the second copy with $\ket\omega$, its restriction is exactly $C_j$.
The remaining states of that copy are excluded by the signal projections.
Since $G_j=A^{\dagger}(2\Pi_j-\Id)A$, its controlled version applies $A$ and $A^{\dagger}$ unconditionally and controls only $2\Pi_j-\Id$, so it uses the same two ordinary calls. This also gives
additional controls and inverses for the constructed encodings, each
using $O(1)$ ordinary queries.

We use the operators $C_j$ to construct a finite chain.
Take
\begin{equation*}
 \lambda=\sqrt{\tau/2},\qquad
 L=\ceil*{\log(10^8R)}.
\end{equation*}
The chain operator $D_A$ defined below has one chain for each $j\in E$ and maps a domain space to an output space. The domain contains a target state $\ket t$, one scalar root $\ket{j,0,\omega}$ for each $j\in E$, and copies of $\calZ$ at levels $1,\ldots,L$ of each chain. The output has $\calZ$ at levels $0,\ldots,L$. Embed each scalar root along $\omega$.
For a domain vector $u$, write $u_t$ for its coefficient of $\ket t$ and $u_{j,r}$ for its level-$r$ component in chain $j$, so that $u_{j,0}$ is the coefficient of $\ket{j,0,\omega}$, and set
\begin{equation*}
 (D_Au)_{j,r}=C_ju_{j,r}-\lambda u_{j,r+1}\quad(0\le r<L),\qquad(D_Au)_{j,L}=C_ju_{j,L},\qquad D_A\ket t=0.
\end{equation*}
Define
\begin{equation*}
 y=\ket t-\sum_j\ket{j,0,\omega},\qquad
 \Pi_\bdry=\Id-\frac{\ket y\bra y}{m+1},\qquad
 M_A=\frac{D_A\Pi_\bdry}{1+\lambda}.
\end{equation*}
Thus every $u$ in the range of $\Pi_\bdry$ satisfies the boundary condition $u_t=\sum_j u_{j,0}$. The known initial state is
\begin{equation*}
 \chi=\frac{\Pi_\bdry\ket t}{\norm{\Pi_\bdry\ket t}}
 =\sqrt{\frac m{m+1}}\ket t
 +\frac1{\sqrt{m(m+1)}}\sum_j\ket{j,0,\omega}.
\end{equation*}
Write $D_A=T_A-\lambda S$, where $T_A$ applies $C_j$ at the current level and $S$ is a known truncated shift.
Since $\norm{T_A}_\infty\le1$ and $\norm{S}_\infty\le1$, we have $\norm{D_A}_\infty\le1+\lambda$, and hence $\norm{M_A}_\infty\le1$.
Coherently selecting the encodings of $C_j$ gives an encoding $U_T$
of $T_A$ with normalization one. Let $U_S$ be a known encoding of $S$.
Use common padded signal spaces and a known qubit preparation
$P_\lambda\ket0=(\ket0+\sqrt\lambda\ket1)/\sqrt{1+\lambda}$.
Selecting $U_T$ or $-U_S$ on this qubit gives signal block
\[
 (\bra0P_\lambda^{\dagger}\ot\Id)
 \begin{pmatrix}T_A&0\\0&-S\end{pmatrix}
 (P_\lambda\ket0\ot\Id)
 =\frac{T_A-\lambda S}{1+\lambda}.
\]
Composing with a known encoding of $\Pi_\bdry$ on a separate signal
register gives a unitary $U_M$ with $M_A=\Pi_{\mathrm{out}}U_M\Pi_{\mathrm{in}}$,
where the known signal projections $\Pi_{\mathrm{in}}$ and $\Pi_{\mathrm{out}}$ select the stated domain and output
sectors, with scalar input roots at level zero and no target-output
sector. Other basis states are padding.
The normalization is $1+\lambda$ for $D_A\Pi_\bdry$, independent of $L$,
and $U_M$ uses $O(1)$ ordinary queries.

\noindent\emph{\SMALL.}
Choose a coordinate with $a_j\le\lambda^2/2$. Set $u_t=1$, $u_{j,r}=\lambda^{-r}C_j^r\omega$, and set the other chains to zero. Then $\Pi_\bdry u=u$ and
\[
 \norm u^2=1+\sum_{r=0}^{L}(a_j/\lambda^2)^r\le3,
 \qquad \norm{D_Au}^2=a_j(a_j/\lambda^2)^L\le2^{-L}.
\]
Let $F_s$ be the right singular-value projection of $M_A$ onto singular values at most $s$, including $\ker M_A$ (singular value zero).
Set $\theta=1/(100\sqrt R)$ and $\widehat u=u/\norm u$.
Every right singular vector of $M_A$ orthogonal to the range of $F_\theta$ has singular value greater than $\theta$, so
\[
 \norm{(\Id-F_\theta)\widehat u}\le\frac{\norm{M_A\widehat u}}{\theta}\le100\sqrt R\,2^{-L/2}\le\frac1{100}.
\]
Moreover, $\braket{\chi}{u}=\sqrt{(m+1)/m}\ge1$ and $\norm u^2\le3$ give $\abs{\braket{\chi}{\widehat u}}\ge1/\sqrt3$.
Consequently
\begin{equation}\label{eq:chain-small}
 \norm{F_\theta\chi}^2
 \ge(1/\sqrt3-1/100)^2.
\end{equation}

\noindent\emph{\LargeCase.}
For $j\in E$ and every level $0\le r\le L$, written as $r=2i$ or $r=2i+1$, define the output components
\begin{equation*}
 z_{j,2i}=\frac{\lambda^{2i}}{a_j^{i+1}}v_j,
 \qquad z_{j,2i+1}=\frac{\lambda^{2i+1}}{a_j^{i+1}}\omega.
\end{equation*}
By the definition of $D_A$, its adjoint satisfies $(D_A^{\dagger}z)_t=0$, $(D_A^{\dagger}z)_{j,0}=\braket{\omega}{C_jz_{j,0}}$, and $(D_A^{\dagger}z)_{j,r}=C_jz_{j,r}-\lambda z_{j,r-1}$ for $1\le r\le L$.
The components above satisfy $C_jz_{j,0}=\omega$ and $C_jz_{j,r}=\lambda z_{j,r-1}$, hence
\[
 D_A^{\dagger} z=\sum_j\ket{j,0,\omega},\qquad
 \norm z^2=\sum_j\frac1{a_j}\sum_{r=0}^L(\lambda^2/a_j)^r\le2R.
\]
Since $\Pi_\bdry\sum_j\ket{j,0,\omega}=\Pi_\bdry\ket t$ and $\norm{\Pi_\bdry\ket t}^2=m/(m+1)\ge1/2$, the vector $\widetilde z=(1+\lambda)z/\norm{\Pi_\bdry\ket t}$ satisfies
\[
 \chi=M_A^{\dagger}\widetilde z,\qquad \norm{\widetilde z}\le4\sqrt R,
\]
where the norm bound uses $\lambda<1$.
By the definition of $F_s$, $\norm{F_s\chi}\le4s\sqrt R$, and in particular
\begin{equation}\label{eq:chain-large}
 \norm{F_{2\theta}\chi}^2\le1/100.
\end{equation}

These spectral bounds yield a test based on two reflections.
Use
\begin{equation*}
 W_M=-(2\Pi_{\mathrm{in}}-\Id)U_M^{\dagger}(2\Pi_{\mathrm{out}}-\Id)U_M.
\end{equation*}
By Jordan's lemma applied to $\Pi_{\mathrm{in}}$ and $U_M^{\dagger}\Pi_{\mathrm{out}}U_M$, on the invariant plane of a singular value $\sigma$ of $M_A$ the eigenphases of $W_M$ are $\pm2\arcsin\sigma$, and right singular vectors with $\sigma=0$ have eigenphase zero.
Since $\theta\le1/100$ and $\arcsin x\le5x/4$ for $0\le x\le1/2$, singular values at most $\theta$ give eigenphase magnitude at most $5\theta/2$, and singular values above $2\theta$ give magnitude above $4\theta$.
Starting from $\chi$, apply quantum phase estimation (\cref{lem:phase-estimation}) to circular
accuracy $\theta/2$, accepting if the estimate has circular distance
at most $3\theta$ from zero.
A constant number of repetitions followed by majority accepts singular
values at most $\theta$ and rejects those above $2\theta$, with error
at most $10^{-3}$ on each eigencomponent.
By \cref{eq:chain-small,eq:chain-large}, the acceptance probability is at least $1/4$ in \SMALL\ and at most $1/64$ in \LargeCase.
Perform $\ceil{64\ln(2/\beta)}$ independent trials and return \SMALL\ if at least a $1/8$ fraction of them accept.
By Hoeffding's inequality (\cref{fact:hoeffding}), the error is at most $\beta$.
Each trial uses $O(1/\theta)=O(\sqrt R)$ ordinary queries.
For a family of index sets as in the statement, run the construction above coherently under the control of the label, with the labeled index set in place of $E$.
Its size $m$, the vector $y$, the projector $\Pi_\bdry$, the state $\chi$, the selection of $C_j$ over its elements, and the two initial cases are known operations controlled by the label.
On the labels of the two initial cases, apply every conjugated reflection with the identity projector, which uses the same calls to $A$ and $A^{\dagger}$ and acts as the identity, and write the prescribed answer.
Every label then uses the same sequence of calls to $A$ and $A^{\dagger}$, whose length depends only on $R$ and $\beta$.
Computing the trial outcomes and the final decision into the answer register without measurement gives $V$.
Reversing the circuit gives $V^{\dagger}$ at the same query cost. For controlled versions, apply each $A,A^{\dagger}$ pair unconditionally and control only the known intermediate reflection, so that the pair cancels on the inactive branch.
\end{proof}

For $0<\tau\le1$, define only for analysis
\[
 L_\tau=\cbra{j\in E\colon a_j<\tau},\quad
 S_\tau=\cbra{j\in E\colon a_j\le\tau/4},\quad N_\tau=\abs{L_\tau},\quad q=\abs{S_\tau}.
\]
\begin{lemma}[Rank-sensitive finding]\label{lem:rank-finder}
Let $E\ne\varnothing$, $R\ge1$, $0<\tau\le1$, and $0<\zeta\le1/100$.
There is a procedure \RankFind$(A,E;\tau,R,\zeta)$ that makes $O(\sqrt R\ln(e/\zeta))$ ordinary queries on every input and returns an index $j_{\mathrm{out}}\in E$ or $\bot$.
Suppose that $a_j>0$ for every $j\in E$ and $\sum_{j\in E}a_j^{-1}\le R$.
If $q\ge1$, then the expected cost is $O(\sqrt{R/q}\ln(eq/\zeta))$ and $\Pr\sbra{j_{\mathrm{out}}\in L_\tau}\ge9/10-\zeta$.
If $N_\tau>0$, then every $B\subseteq S_\tau$ satisfies
\[
 \Pr\sbra*{j_{\mathrm{out}}\in B}\ge
 \frac{\abs{B}}{5N_\tau}\Pr\sbra*{j_{\mathrm{out}}\in L_\tau}.
\]
\end{lemma}
\begin{proof}
All calls use newly initialized work registers, so the guarantees below hold conditional on the preceding history.
The proof first assumes a lower bound $q_0$ on $q$. Under this assumption, it builds a generator that gives every coordinate of $S_\tau$ comparable output mass, amplifies the generator with a verifier, and bounds the resulting output distribution. A geometric search over $q_0$ then removes the assumption.

\noindent\emph{Isolating generator.}
First suppose $q\ge q_0\ge1$ is promised.
Return failure if $R\tau<4q_0$.
Under the reciprocal promise, every $j\in S_\tau$ has $a_j^{-1}\ge4/\tau$ and every $j\in L_\tau$ has $a_j^{-1}>1/\tau$, so $R\tau\ge4q$ and $N_\tau<R\tau$.
Hence this early failure occurs only if $q<q_0$.
Otherwise set $\ell=\ceil{R\tau}$, so that $N_\tau\le\ell$, and include each coordinate of $E$ in a set $T$ independently with probability $1/\ell$.
Store the mask of $T$ in a work register and prepare the
uniform subset state by Hadamards only on coordinates in $T$.
For each $j\in S_\tau$,
\[
 \Pr\sbra*{T\cap L_\tau=\cbra{j}}=\frac1\ell\rbra*{1-\frac1\ell}^{N_\tau-1}\ge\frac1{e\ell}.
\]
Independently of this isolation event, Markov's inequality (\cref{fact:markov}) bounds the reciprocal mass of $T\setminus L_\tau$ by $100R/\ell$ with probability at least $99/100$.
Condition on both events. Apply \cref{lem:threshold-or} to the family of index sets $T'\subseteq T$, with threshold $\tau$, reciprocal-mass bound $100R/\ell$, and error $10^{-4}$. It gives a label-preserving unitary that distinguishes the sets $T'$ containing $j$ from those not containing it.
Computing its answer, applying a phase flip, and uncomputing the answer implement the phase oracle that multiplies $\ket{T'}$ by $-1$ exactly when $j\in T'$, up to error $2\sqrt{10^{-4}}$ in vector norm on the uniform superposition over $T'\subseteq T$.
Applying the same Hadamards returns $\cbra{j}$ with probability at
least $3/4$. Reversibly output $j$ only for a singleton
$\cbra{j}\subseteq T$, keeping all other registers unchanged. Otherwise output
failure.

Purifying the random mask gives a generator $G$ and its inverse
at cost $g=O(\sqrt{R/\ell})$.
Write $\pi_j$ for the probability that $G$ outputs $j$.
Since $(99/100)(3/4)/e\ge1/4$, the isolation bound, the reciprocal-mass event and the success probability $3/4$ give the lower bound in
\begin{equation}\label{eq:candidate-masses}
 \pi_j\ge\frac1{4\ell}\quad(j\in S_\tau),
 \qquad \pi_j\le\frac1\ell\quad(j\in E).
\end{equation}
The upper bound holds because every coordinate output belongs to $T$.

\noindent\emph{Verifier.}
By quantum amplitude estimation (\cref{lem:amplitude-estimation}) with threshold $\tau$ and error $1/10$, applied to the family $\cbra{\Pi_j}_{j\in E}$ extended by $\Pi_\bot=\Id$, we obtain a label-preserving verifier.
Take \SMALL\ as acceptance, except that the generator's failure label $\bot$ is always rejected.
It accepts $S_\tau$ with probability at least $9/10$ and accepts $E\setminus L_\tau$ with probability at most $1/10$.
It acts on the label and a newly initialized workspace, leaving the generator's work registers unchanged.
Its query cost is $v=O(\tau^{-1/2})=O(\sqrt{R/\ell})$.
Labels in $L_\tau\setminus S_\tau$ are acceptable intermediate outputs.
Call labels in $L_\tau$ legal, and identify the failure output with $\bot$.

\noindent\emph{Amplification.}
To amplify this generator, we adapt bounded-error amplitude amplification~\cite[p.~297]{HMW03} to strong labels $S_\tau$, intermediate labels $L_\tau\setminus S_\tau$, and bad labels $(E\setminus L_\tau)\cup\cbra{\bot}$.
Put $p_0=q_0/(4\ell)$, a lower bound on the strong mass when $q\ge q_0$, and fix $0<\beta<1/10$.
Let $\calV_r$ be a majority verifier of error $2^{-r-c}$ on strong and bad labels, where $c=O(\ln(e/\beta))$ is chosen so that $c\ge30$ and $64\cdot2^{1-c}\le\beta$.
Each $\calV_r$ runs the verifier on the label and a newly initialized workspace.
For $r\ge1$, it writes the AND of the old active flag and the new
verification outcome into a new flag, keeping the old flag and all
verification records. The verifier $\calV_0$ initializes the active flag.
Let $\calR_0=\Id-2\ket0\bra0$ on all registers.
Let $\calR_1$ apply phase $-1$ when the active flag is one and act as the
identity otherwise.
Initialize $G_1$ by applying $G$ and $\calV_0$, and set
\begin{equation*}
 G_{r+1}=\calV_r(-G_r\calR_0G_r^{\dagger}\calR_1)G_r.
\end{equation*}
Running the verifier on every label avoids controlled unknown queries.

Let $\alpha_r^2,\gamma_r^2,\varrho_r^2$ be the active strong, intermediate, and bad masses, with sum $p_r$. Amplification multiplies the active vector by $3-4p_r$, so
\[
 \alpha_{r+1}\ge\abs{3-4p_r}\sqrt{1-2^{-r-c}}\,\alpha_r,
 \qquad \varrho_{r+1}^2\le9\cdot2^{-r-c}\varrho_r^2,
 \qquad \sum_r\varrho_r^2\le2^{1-c}.
\]
Set $K=1+\ceil{\log_3(1/\sqrt{p_0})}$.
Suppose that $\sum_{r=1}^K(\alpha_r^2+\gamma_r^2)<1/16$.
Then $p_r<1/16+2^{1-c}<1/4$ for every $r\le K$, so every factor $3-4p_r=3(1-4p_r/3)$ is positive.
By \cref{eq:candidate-masses} and $q\ge q_0$, we have $\alpha_1^2\ge(1-2^{-c})p_0$, and the choice of $K$ gives $3^{K-1}\ge p_0^{-1/2}$.
Multiplying the recurrence and using $\prod_i(1-x_i)\ge1-\sum_ix_i$ therefore gives $\alpha_K\ge(1-2^{-c})\rbra{1-\tfrac43\rbra{\tfrac1{16}+2^{1-c}}}>1/2$, which contradicts $\alpha_K^2<1/16$.
Thus that sum is at least $1/16$.
Execute $G_1,\ldots,G_K$ independently, stopping at the first active outcome, and repeat this pass at most $64$ times.
A pass has no active outcome with probability at most $e^{-1/16}$. Since $c\ge30$, the total bad-output probability is at most $64\cdot2^{1-c}\le2^{-23}$.
Hence, as $e^{-4}+2^{-23}<1/10$, a legal label is returned with probability at least $9/10$.
The query cost $\kappa_r$ of $G_r$ satisfies the recurrence
\[
 \kappa_1=O(g+cv),\qquad \kappa_{r+1}=3\kappa_r+O((r+c)v),
\]
which gives total cost $O((g+v\ln(e/\beta))/\sqrt{p_0})=O(\sqrt{R/q_0}\ln(e/\beta))$. The bad-output bound holds even when $q<q_0$.

\noindent\emph{Output distribution.}
Let $s_{i,j}$ be the acceptance probability of the new
verification outcome of $\calV_i$ on label $j$. Since each verification uses a newly initialized workspace, this probability is
independent of the previous records. The active mass of $j$ after
$G_r$ is
\[
 D_r\pi_j\prod_{i=0}^{r-1}s_{i,j},\qquad
 D_r=\prod_{h=1}^{r-1}(3-4p_h)^2,\qquad D_1=1.
\]
The verifier product is at most one for every label and at least
$1-2^{1-c}\ge1-\beta$ for strong labels.
Let $D_*>0$ be the expected sum of the factors $D_r$ over the
trials actually executed across all $64$ passes.
Each execution starts from a newly prepared initial state, so its reach probability depends only on
earlier trials and is independent of its own output. Thus, for
$B\subseteq S_\tau$ and $B'\subseteq E$,
\begin{equation}\label{eq:intermediate-relative-mass}
 \Pr\sbra*{j_{\mathrm{out}}\in B}\ge(1-\beta)D_*\sum_{j\in B}\pi_j,
 \qquad
 \Pr\sbra*{j_{\mathrm{out}}\in B'}\le D_*\sum_{j\in B'}\pi_j.
\end{equation}
Combining this with \cref{eq:candidate-masses} and $\beta<1/10$ gives, when $N_\tau>0$,
\begin{equation}\label{eq:minorization}
 \Pr\sbra*{j_{\mathrm{out}}\in B}\ge\frac{\abs{B}}{5N_\tau}
       \Pr\sbra*{j_{\mathrm{out}}\in L_\tau}.
\end{equation}
If $R\tau\ge4q_0$, both \cref{eq:candidate-masses,eq:intermediate-relative-mass} remain valid even when $q<q_0$.
If $R\tau<4q_0$, the immediate failure makes \cref{eq:minorization} trivial.
Thus \cref{eq:minorization} holds for every $q_0\ge1$.

\noindent\emph{Unknown multiplicity.}
To remove the multiplicity promise, set $q_{\max}=\min\cbra{\abs{E},\floor{R\tau/4}}$ and return failure if it is zero.
Adapting the geometric search schedule of~\cite{HMW03}, try $q_0=\widehat q_i=2^{M-i}$, $0\le i\le M=\ceil{\log q_{\max}}$, with bad-output bounds $\beta_i=\zeta/(16\widehat q_i^2)$, stopping on a coordinate output.
Impossible guesses may return failure immediately.
The costs sum to $O(\sqrt R\ln(e/\zeta))$ on every input.
Under the reciprocal promise, $\sum_i\beta_i<\zeta$ bounds the probability of an output outside $L_\tau$, including when $q=0$.
If $q\ge1$, let $\widehat q_{i_0}$ be the first valid guess, so that $q/2<\widehat q_{i_0}\le q$. Every later guess is valid.
For $x\ge1$, put $\Gamma(x)=C_0\sqrt{R/x}\ln(16ex^2/\zeta)$, where the absolute constant $C_0$ is chosen so that $\Gamma(\widehat q_i)$ bounds the cost of the guess $\widehat q_i$.
The earlier guesses are $\widehat q_i=2^{i_0-i}\widehat q_{i_0}$, so their total cost is at most $\sum_{t\ge1}\Gamma(2^t\widehat q_{i_0})=O(\Gamma(\widehat q_{i_0}))$.
Each valid guess fails with conditional probability at most $1/10$, so the $t$th guess after $\widehat q_{i_0}$ is reached with probability at most $10^{-t}$.
Since the logarithm in $\Gamma$ is nondecreasing, that guess costs at most $\Gamma(2^{-t}\widehat q_{i_0})\le2^{t/2}\Gamma(\widehat q_{i_0})$.
Hence the expected cost is at most
\[
 O\rbra*{\Gamma(\widehat q_{i_0})}+\Gamma(\widehat q_{i_0})\sum_{t\ge0}10^{-t}2^{t/2}=O\rbra*{\sqrt{R/q}\ln(eq/\zeta)}.
\]
Excluding earlier bad stops leaves legal success at least $9/10-\zeta$.
Finally, weight \cref{eq:minorization} at each invocation by its probability of being reached and sum the disjoint output events.
\end{proof}

\subsubsection{Minimum-marginal estimation}\label{subsec:minimum-marginal}

The estimator alternates rank-sensitive finding with scalar threshold comparisons.

\begin{lemma}[Minimum-marginal estimation]\label{lem:min-interface}
Let $A$ be a reversible sampler for a distribution $\calD$ on $2^{\sbra{n}}$, let $\varnothing\ne E\subseteq\sbra{n}$ be known, and let $R_{\mathrm{in}}\ge1$ and $0<\beta<1/3$.
Put $R=2^{\ceil{\log R_{\mathrm{in}}}}$.
There is a procedure \MinMarg$(A,E;R_{\mathrm{in}},\beta)$ that, on every input, uses $O(\sqrt{R_{\mathrm{in}}}\ln(1/\beta))$ ordinary queries to $A,A^{\dagger}$ and returns a dyadic value $\widehat a\in[1/R,1]$.
If the marginals on $E$ are positive and $R_\calD(E)\le R_{\mathrm{in}}$, then with probability at least $1-\beta$,
\[
 a_*\le\widehat a\le256a_*,\qquad a_*=\min_{j\in E}a_j.
\]
\end{lemma}

\begin{proof}[Proof of \cref{lem:min-interface}]
Since $R_{\mathrm{in}}\le R<2R_{\mathrm{in}}$, bounds in terms of $R$ change costs only by an absolute factor.
Fix a sufficiently large absolute constant $C_1$ and then a sufficiently
small absolute constant $0<\nu<1/2$, as justified below.

\noindent\emph{Scalar estimation.}
We first define a procedure \ScalarEstimate$(A,\Pi;R,\nu)$ for a known projector $\Pi$.
Set $a=\norm{\Pi A\ket0}^2\in[0,1]$. The procedure scans the thresholds
$t_i=4^{-i}$, starting at $i=0$ and ending at the first threshold at
most $1/R$. At threshold $t_i$, set
\begin{equation}\label{eq:scalar-errors}
 \alpha_i=\nu\min\cbra{1,(Rt_i)^{-1}}.
\end{equation}
By quantum amplitude estimation (\cref{lem:amplitude-estimation}) for the single projector $\Pi$, with threshold $\tau=t_i$ and error $\beta=\alpha_i$, a fixed number of queries gives \LargeCase\ if $a\ge t_i$ and \SMALL\ if $a\le t_i/4$, with probability at least $1-\alpha_i$.
At the first \LargeCase\ answer return $b=\min\cbra{1,4t_i}$ clipped below at $1/R$. If no such answer occurs, return one.
On every input, this finite scan returns a dyadic number in $[1/R,1]$.

Suppose $a\ge1/R$ and every promised comparison is correct.
The last threshold is at most $a$, so the fallback is not used.
The first \LargeCase\ answer occurs at the first threshold at most $a$,
or at the preceding threshold, where either answer is allowed. Therefore
\begin{equation}\label{eq:scalar-good}
 a\le b\le16a.
\end{equation}
When every promised comparison is correct, the scan uses at most
$C_{\mathrm{sc}}a^{-1/2}\ln(eRa/\nu)$ queries, where
$C_{\mathrm{sc}}$ is an absolute constant.
This bound, and the cost $O(\sqrt R\ln(e/\nu))$ regardless of the answers, follow by summing the geometric threshold costs $c_i=O(t_i^{-1/2}\ln(e\max\cbra{Rt_i,1}/\nu))$.

\noindent\emph{Base routine.}
The base routine sets $v=1$, $\zeta=10^{-3}$, and $T_{\mathrm{cap}}=\ceil{100C_1\sqrt R\ln(e/\nu)}$.
While fewer than $T_{\mathrm{cap}}$ queries have been made, it repeats the following round.
If $v\le256/R$, it returns $v$.
Otherwise, it computes $j=\RankFind(A,E;v/64,R,\zeta)$ by \cref{lem:rank-finder}.
If $j\ne\bot$, it computes $b=\ScalarEstimate(A,\Pi_j;R,\nu)$ and sets $v=b$ when $b\le v/4$.
When the loop ends, or when the cap $T_{\mathrm{cap}}$ interrupts a call, it returns the current $v$.
Each call is interrupted immediately before its next query would exceed the cap $T_{\mathrm{cap}}$.
A call is not skipped according to its worst-case cost.
We count at least one query for every round, so the cap $T_{\mathrm{cap}}$ also bounds the number of rounds.

With $\tau=v/64$, the sets $S_\tau$ and $L_\tau$ of \cref{lem:rank-finder} consist of the coordinates with $a_j\le v/256$ and $a_j<v/64$, respectively. Call outputs in $L_\tau$ legal.
Each call uses independent randomness and newly initialized work registers, so the guarantees of \cref{lem:rank-finder} and of \ScalarEstimate\ hold conditionally on the history.
Moreover, $v$ changes only through a completed estimate with $b\le v/4$, so successive values decrease by a factor of at least four and remain in $[1/R,1]$.
A comparison whose input lies strictly between its two promised ranges may return either answer. The analysis below never corrects such comparisons.

\noindent\emph{Uncapped cost.}
We first bound the uncapped query cost.
Adapting the rank analysis of~\cite[Lemmas~1 and~2]{DH96}, temporarily omit the query cap and correct only completed promised scalar comparisons, leaving the finder and the unpromised comparisons unchanged. Then $v\ge a_*$ always. Define
\[
 q(v)=\abs{S_{v/64}},\qquad N(v)=N_{v/64}.
\]
The region $a_*\le v\le256a_*$ is correct and absorbing. Both
$q(v)=0$ and the explicit stopping rule place $v$ in this region. Before this region, every legal finder output causes an update by \cref{eq:scalar-good}, and every update from $v$ to a new value $v'$ satisfies
\begin{equation}\label{eq:rank-drop}
 N(v')\le q(v).
\end{equation}
Partition the execution before this region into phases $2^r\le q(v)<2^{r+1}$.
A legal output occurs after an expected constant number of calls.
Thereafter \cref{eq:rank-drop} ensures $N(v)<2^{r+1}$ for the rest of this phase.
Let $B_r$ be the globally smallest $h_r=\max\cbra{1,2^{r-1}}$ coordinates, with arbitrary tie breaking.
While the phase persists, $B_r\subseteq S_\tau$.
By \cref{lem:rank-finder}, each subsequent call has constant probability of outputting a member $j$ of $B_r$.
Its estimate gives
\[
 q(v')\le\abs{\cbra{i:a_i\le a_j/16}}<h_r,
\]
ending the phase, even in the presence of ties. Each phase therefore contains an expected constant number of calls. In phase $r$, \cref{lem:rank-finder} with $2^r\le q(v)<2^{r+1}$ and $\zeta=10^{-3}$ bounds the expected cost of each call by $O\rbra{\sqrt{R2^{-r}}\ln(e2^{r+1}/\zeta)}=O\rbra{(r+1)2^{-r/2}\sqrt R}$. Conditioned on the history at each call, these expected costs sum to
\begin{equation*}
 O(\sqrt R)\sum_{r\ge0}(r+1)2^{-r/2}=O(\sqrt R).
\end{equation*}

It remains to include the scalar-estimation costs.
The function $g(a)=a^{-1/2}\ln(eRa/\nu)$ is decreasing on $[1/R,1]$ when $\nu<1/e$.
In a phase with $q=q(v)\ge1$, a candidate that is not legal has $a_j\ge v/64$, while $Rv/64\ge4q$.
Its scalar cost is at most $O(\sqrt{R/q}\ln(eq/\nu))$, paid by the same phase sum up to a factor $\ln(e/\nu)$.
A legal candidate is accepted.
Its new value $v'$ satisfies $a_j\le v'\le16a_j$, so its cost is at most $4C_{\mathrm{sc}}(v')^{-1/2}\ln(eRv'/\nu)$.
Accepted values decrease by at least four and never fall below $1/R$.
Summing over the accepted values in reverse order gives the
pathwise bound
\[
 \sum_{v'}(v')^{-1/2}\ln(eRv'/\nu)
 =O(\sqrt R\ln(e/\nu)).
\]
Thus the corrected execution has expected query cost at most $C_1\sqrt R\ln(e/\nu)$ to reach the correct region.

\noindent\emph{Query cap and comparison errors.}
We now restore the query cap and account for errors in the scalar comparisons.
By the choice of $T_{\mathrm{cap}}$, Markov's inequality (\cref{fact:markov}) bounds failure to reach the correct region by $1/100$. Each threshold comparison has a predetermined query count
$c_i\ge t_i^{-1/2}$, padding if necessary, and reveals its decision
only upon completion. An interrupted comparison contributes no decision.
By \cref{eq:scalar-errors},
\[
 \alpha_i\le\frac{\nu}{\sqrt{Rt_i}}\le\frac{\nu c_i}{\sqrt R}.
\]
Given the history before a comparison, whether it can finish is
determined by the number of queries remaining under the cap. Each completed promised comparison
therefore has the stated conditional error bound. The sum of their costs
is at most $T_{\mathrm{cap}}$, including comparisons after the correct region
was first reached.
The probability of an error in any completed promised comparison
is therefore at most $\nu T_{\mathrm{cap}}/\sqrt R$.
Since $C_1$ does not depend on $\nu$, which enters only through $\ln(e/\nu)$, taking $\nu$ sufficiently small makes this bound less than $1/100$.
Until the first promised scalar error, the actual and corrected executions coincide.
Together with the cap failure bound, this shows that, under the promise of \cref{lem:min-interface}, the base routine returns a value in $[a_*,256a_*]$ with probability at least $9/10$.
Define \MinMarg$(A,E;R_{\mathrm{in}},\beta)$ as the median of an odd number $O(\ln(1/\beta))$ of independent runs of the base routine.
By Hoeffding's inequality (\cref{fact:hoeffding}) and the success probability $9/10$ of each run, this median lies in $[a_*,256a_*]$ with probability at least $1-\beta$.
It always lies in $[1/R,1]$, and its cost is $O(\sqrt{R_{\mathrm{in}}}\ln(1/\beta))$.
\end{proof}

\subsubsection{Proof of the upper bound}\label{subsec:cover-upper-proof}
The tester \textsc{ReversibleUnionSizeTest}$(A;n,k,\delta)$ combines amplitude amplification with the subroutines \GapTest, \UnionSample, and \MinMarg.
It sets $\gamma=(10^5\cdot1024)^{-1}$, $\vartheta=\gamma\delta/(4096k)$, and $\beta_{\mathrm{sub}}=1/100$.
\begin{enumerate}[label=(\alph*)]
\item If $n\geq2k$, it rejects if $\GapTest(A;n,k,k,2k/\vartheta,\beta_{\mathrm{sub}})$ rejects.
\item For $t=1,\ldots,128$, it draws a random set $S_t\subseteq\sbra{n}$ that contains each coordinate independently with probability $1/(2k)$, and runs amplitude amplification (\cref{lem:amplitude-amplification}) on $A$ and $\Pi_{S_t}$ with $\zeta=\gamma\delta/(4k)$ and $\beta=\beta_{\mathrm{sub}}$.
It rejects if at least $96$ of these runs have a successful measurement.
\item It sets $H=100\cdot1024/\delta$ and $J=\UnionSample(A;k,H)$.
It rejects if $\abs{J}>k$, and accepts if $J=\varnothing$.
\item Otherwise, it sets $R_*=200kH$ and $\widehat a=\MinMarg(A,J;R_*,\beta_{\mathrm{sub}})$.
It then sets
\[
 d=\min\cbra*{k,\max\cbra*{1,\floor{\gamma\delta/(8\widehat a)}}},\qquad H_d=\frac{4104kd}{\gamma\delta},
\]
and returns the answer of $\GapTest(A;n,k,d,H_d,\beta_{\mathrm{sub}})$.
\end{enumerate}
\Cref{alg:cover} summarizes this tester.

\begin{algorithm}[h]
\caption{\textsc{ReversibleUnionSizeTest}$(A;n,k,\delta)$}\label{alg:cover}
\begin{algorithmic}[1]
\Require Integers $1\le k<n$, $0<\delta\le1$, and ordinary queries to $A,A^{\dagger}$ for a reversible sampler $A$
\Ensure \Accept\ or \Reject
\State $\gamma\gets(10^5\cdot1024)^{-1}$
\State $\beta_{\mathrm{sub}}\gets1/100$
\State $\vartheta\gets\gamma\delta/(4096k)$
\If{$n\ge2k$}
 \If{$\GapTest(A;n,k,k,2k/\vartheta,\beta_{\mathrm{sub}})=\Reject$} \Comment{\cref{lem:gap-interface}}
  \State \Return \Reject
 \EndIf
\EndIf
\For{$t=1,\ldots,128$}
 \State $S_t\gets$ random subset of $\sbra{n}$ with $\Pr\sbra{i\in S_t}=1/(2k)$ independently over $i$
 \State Run amplitude amplification on $A,\Pi_{S_t}$ with $\zeta=\gamma\delta/(4k)$, $\beta=\beta_{\mathrm{sub}}$ \Comment{\cref{lem:amplitude-amplification}}
\EndFor
\State $N_{\mathrm{det}}\gets$ number of $t$ with a successful measurement
\If{$N_{\mathrm{det}}\ge96$} \State \Return \Reject \EndIf
\State $H\gets100\cdot1024/\delta$
\State $J\gets\UnionSample(A;k,H)$ \Comment{\cref{lem:union-interface}}
\If{$\abs{J}>k$} \State \Return \Reject \EndIf
\If{$J=\varnothing$} \State \Return \Accept \EndIf
\State $R_*\gets200kH$
\State $\widehat a\gets\MinMarg(A,J;R_*,\beta_{\mathrm{sub}})$ \Comment{\cref{lem:min-interface}}
\State $d\gets\min\cbra{k,\max\cbra{1,\floor{\gamma\delta/(8\widehat a)}}}$
\State $H_d\gets4104kd/(\gamma\delta)$
\State \Return $\GapTest(A;n,k,d,H_d,\beta_{\mathrm{sub}})$ \Comment{\cref{lem:gap-interface}}
\end{algorithmic}
\end{algorithm}

\begin{lemma}[Good-scale interval]\label{lem:good-scale}
Let $k,m$ be integers with $1\le k<m\le Ck$, where $C\ge1$, and let $a_1\ge\cdots\ge a_m>0$ satisfy $\sum_{j>k}a_j\ge\eta$ for a given $\eta>0$.
Fix $0<\gamma<1/4$ and put $c=\gamma/(2C)$.
There is $0<\lambda_{\mathrm{gs}}<a_{k+1}$ with the following two properties.
\begin{enumerate}[label=(\roman*)]
\item $\sum_{j:a_j\le\lambda_{\mathrm{gs}}}a_j\le\gamma\eta$.
\item For every $\kappa\ge1$ and $\lambda_{\mathrm{gs}}<\widehat a\le\kappa a_{k+1}$, the integer $d=\max\cbra{1,\floor{c\eta/\widehat a}}$ satisfies $1\le d\le\min\cbra{k,m-k}$ and
\begin{equation*}
 \frac1d\sum_{j=1}^{k+d}\frac1{a_j}
 \le\frac{(4C\kappa+2C)k}{\gamma\eta}.
\end{equation*}
\end{enumerate}
\end{lemma}
\begin{proof}
Let $D=\min\cbra{k,m-k}$ and $R_d=\sum_{i=1}^d1/a_{k+i}$.
Let $d_*$ be the first $d\le D$ with $R_d/d>k/(c\eta)$, or $D+1$ if none exists.
The $m-k$ tail entries have total mass at least $\eta$, so $a_{k+1}\ge\eta/(m-k)$.
Since $m-k<Ck$ and $\gamma<1/4$, this gives $R_1\le(m-k)/\eta<Ck/\eta<k/(c\eta)$.
Thus no crossing occurs at $d=1$, and $d_*\ge2$ because $D\ge1$.
Set $\lambda_{\mathrm{gs}}=c\eta/d_*$, and call a coordinate low-rate if its marginal is at most $\lambda_{\mathrm{gs}}$.

If the crossing exists, then $a_{k+d_*}<c\eta/k$.
The low-rate coordinates among the first $d_*-1$ tail entries have total mass at most $c\eta$, and all later tail entries have total mass at most $Cc\eta$.
Their sum is at most $\gamma\eta$.
If no crossing exists, the low-rate tail mass is at most $(m-k)\lambda_{\mathrm{gs}}\le Cc\eta\le\gamma\eta$, since $d_*=D+1$ and $m-k\le CD$.
In either case, the whole tail has greater mass, so $\lambda_{\mathrm{gs}}<a_{k+1}$.
The first $k$ entries consequently lie above $\lambda_{\mathrm{gs}}$, proving (i) for all coordinates.

For (ii), $\widehat a>\lambda_{\mathrm{gs}}=c\eta/d_*$ gives $\floor{c\eta/\widehat a}<d_*$.
Since $d_*\ge2$, this ensures $1\le d<d_*\le D+1$, and hence $R_d/d\le k/(c\eta)$.
Also $d\widehat a\ge c\eta/2$ and $a_{k+1}\ge\widehat a/\kappa$.
Therefore the first $k$ reciprocal terms divided by $d$ contribute at most $2\kappa k/(c\eta)$.
Adding the tail contribution proves (ii).
\end{proof}

\begin{proof}[Proof of \cref{thm:reversible-subset-upper}]
Run \textsc{ReversibleUnionSizeTest}$(A;n,k,\delta)$.

\noindent\emph{YES instances.}
If $X\subseteq K$ almost surely for a set $K$ with $\abs{K}\le k$, both gap tests have
their YES promise, regardless of the minimum estimate and clipped $d$.
By \cref{lem:union-interface}, every returned sample $X$ satisfies $\calD(X)>0$, hence $X\subseteq K$, so $\abs{J}\le k$. A run of amplitude amplification has a successful measurement only if its
random set meets $K$, which has probability at most $1/2$.
By Hoeffding's inequality (\cref{fact:hoeffding}), the probability of $96$ or more detections is at most
$2e^{-16}<1/100$.
Together with the two gap-test errors, false rejection has probability
at most $3/100$.

\noindent\emph{NO instances.}
Let $B=\cbra{j:a_j\ge\vartheta}$ and
$q_{\mathrm{out}}=\Pr\sbra{X\nsubseteq B}$.
If $\abs{B}\ge2k$, the first gap test rejects with probability at least
$99/100$.
If $q_{\mathrm{out}}\ge\gamma\delta$, assign each sampled subset not contained in $B$ to one of its coordinates outside $B$ by a fixed rule. For $j\notin B$, let $b_j$ be the probability that $X$ is assigned to $j$. These assigned masses
satisfy $\sum_{j\notin B}b_j=q_{\mathrm{out}}$ and
$b_j\le a_j<\vartheta$. For each test set $S_t$, put
$Z=\sum_{j\in S_t\setminus B}b_j$. Then $p_\calD(S_t)\ge Z$ and
\[
 \Ex\sbra*{Z}=\frac{q_{\mathrm{out}}}{2k},\qquad
 \frac{\operatorname{Var}\sbra*{Z}}{(\Ex\sbra*{Z})^2}\le\frac{2k\vartheta}{q_{\mathrm{out}}}\le\frac1{2048}.
\]
By Chebyshev's inequality (\cref{fact:chebyshev}), $p_\calD(S_t)\ge Z\ge\Ex\sbra{Z}/2\ge\gamma\delta/(4k)=\zeta$ with probability at least $1-1/512$.
On this event, amplitude amplification (\cref{lem:amplitude-amplification}) has a successful measurement with probability at least $1-\beta_{\mathrm{sub}}$, so each independent trial detects with probability at least $(1-1/512)(1-\beta_{\mathrm{sub}})\ge15/16$.
By Hoeffding's inequality (\cref{fact:hoeffding}), the probability of fewer than $96$ detections is at most
$2e^{-9}<1/100$.

It remains to consider the case $\abs{B}<2k$ and $q_{\mathrm{out}}<\gamma\delta$.
For every $k$-set $K$, the truncated sample $X\cap B$ is not contained in $K$ with probability
at least $(1-\gamma)\delta\ge\delta/2$, so $m=\abs{B}>k$.
List the coordinates of $B$ as $j_1,\ldots,j_m$ with $a_{j_1}\ge\cdots\ge a_{j_m}\ge\vartheta$, and let
$K_*=\cbra{j_1,\ldots,j_k}$. A union bound gives
\[
 \sum_{j\in B\setminus K_*}a_j\ge\Pr\sbra*{X\cap B\nsubseteq K_*}\ge\delta/2.
\]
Apply \cref{lem:good-scale} to $a_{j_1}\ge\cdots\ge a_{j_m}$ with $C=2$ and $\eta=\delta/2$, obtaining $\lambda_{\mathrm{gs}}$.
Since $p_\calD(\overline{K_*})\ge\delta$, \cref{lem:union-interface} gives
\[
 \Pr\sbra*{J\subseteq K_*}\le\frac{1024}{\delta H}=1/100.
\]
By $q_{\mathrm{out}}<\gamma\delta$ and \cref{lem:good-scale}(i), the probability that a sample leaves $B$ or contains a coordinate $j\in B$ with $a_j\le\lambda_{\mathrm{gs}}$ is at most
$\gamma\delta+\gamma\delta/2=3\gamma\delta/2$.
By \cref{lem:union-interface}, the chance that any such sample is returned is at most
\[
 H(\gamma+\gamma/2)\delta<1/100.
\]
Outside these exceptions, either the union is oversized and the tester rejects, or
\[
 J\subseteq B,\qquad \lambda_{\mathrm{gs}}<\min_{j\in J}a_j\le a_{j_{k+1}}.
\]
Also, \cref{lem:union-interface} gives
\[
 \Ex\sbra*{\sum_{j\in J\cap B}\frac1{a_j}}
 =\sum_{j\in B}\frac{\Pr\sbra*{j\in J}}{a_j}
 \le \abs{B}H=mH<2kH.
\]
With $R_*=200kH$, Markov's inequality (\cref{fact:markov}) gives
$\Pr\sbra{J\subseteq B\text{ and }R_\calD(J)>R_*}\le1/100$.
Conditional on $J\subseteq B$ and $R_\calD(J)\le R_*$, the call to \MinMarg\ succeeds with
probability at least $99/100$.
By \cref{lem:good-scale} with $\kappa=256$, clipping
of $d$ is unnecessary and the final gap test is in its NO promise.
Its error is at most $1/100$.
The five error events above, namely $J\subseteq K_*$, the return of a sample that leaves $B$ or contains a coordinate with marginal at most $\lambda_{\mathrm{gs}}$, the event that $J\subseteq B$ and $R_\calD(J)>R_*$, a failure of \MinMarg, and an error of the final gap test, have total probability at most
$1/20<1/10$. An empty union is already covered by the event $J\subseteq K_*$.

\noindent\emph{Query complexity.}
The parameter scales are
$\vartheta=\Theta(\delta/k)$, $H=\Theta(1/\delta)$,
$R_*=\Theta(k/\delta)$, and $H_d/d=\Theta(k/\delta)$ for every possible $d$.
By \cref{lem:amplitude-amplification,lem:gap-interface,lem:union-interface,lem:min-interface}, each stage therefore uses
$O(\sqrt{k/\delta})$ queries on every input, including off-promise
executions. There are only constantly many stages.
\end{proof}

\subsection{Lower bound from a reversible sampler}\label{subsec:cover-lower-reversible}
Product rotations with moment-matched priors give a matching lower bound for union size testing from a reversible sampler.

\begin{lemma}[Finite moment-matched priors]\label{lem:scalar-priors}
For $0<a<b$ and integer $D\ge1$, there are finitely supported probability measures $\mu_0,\mu_1$ on $\cbra{0}\cup[a,b]$ whose moments agree through degree $D$, with $\supp(\mu_1)\subseteq[a,b]$ and
\begin{equation*}
 \mu_0((0,\infty))\le\frac{D^2a}{b-a}.
\end{equation*}
\end{lemma}
\begin{proof}
Let $x_j=\cos(j\pi/D)$, $0\le j\le D$, and let $\ell_j$ be their Lagrange cardinal polynomials.
Put $t_0=(a+b)/(b-a)>1$, $p_j=(a+b-(b-a)x_j)/2$, and $w_j=\ell_j(t_0)$.
Since the nodes decrease and $t_0>1$, $\operatorname{sgn}(w_j)=(-1)^j$.
Interpolating the constant polynomial and $T_D$ gives
\[
 \sum_jw_j=1,\qquad \sum_j\abs{w_j}=T_D(t_0).
\]
Interpolation after the affine change of variable gives $P(0)=\sum_jw_jP(p_j)$ for every polynomial $P$ of degree at most $D$.
Let $\mu_0$ assign mass proportional to $1$ to the point $0$ and to $\abs{w_j}$ to $p_j$ for odd $j$, and let $\mu_1$ assign mass proportional to $w_j$ to $p_j$ for even $j$.
Both normalizing constants equal $(1+T_D(t_0))/2$, so the moments of $\mu_0$ and $\mu_1$ agree through degree $D$.
Choose $u\ge0$ with $\cosh u=t_0$.
For $x,y\ge0$, the addition formula gives $\tanh(x+y)\le\tanh x+\tanh y$, hence $\tanh(Du/2)\le D\tanh(u/2)$.
Using $T_D(\cosh u)=\cosh(Du)$, we obtain
\[
  \mu_0((0,\infty))=\frac{T_D(t_0)-1}{T_D(t_0)+1}=\tanh^2(Du/2)\le D^2\tanh^2(u/2)=D^2\frac{t_0-1}{t_0+1}\le\frac{D^2(t_0-1)}2=\frac{D^2a}{b-a}.
\]
\end{proof}

\begin{lemma}[Product-rotation priors]\label{lem:product-rotation-priors}
Let $k,q\ge1$ be integers, and let $0<\alpha\le1/2$ satisfy $2q\alpha\le1/8$.
For $\theta\in\R^{2k}$, let $U_\theta=\bigotimes_{j=1}^{2k}e^{i\theta_j\sigma_3}$ and $A_\theta=H^{\ot2k}U_\theta H^{\ot2k}$.
There are finitely supported probability measures $\pi_0,\pi_1$ on $\R^{2k}$, called the sparse and dense priors, with the following properties.
\begin{enumerate}[label=(\roman*)]
\item With probability at least $7/8$ under $\pi_0$, at most $k$ coordinates of $\theta$ are nonzero.
\item Almost surely under $\pi_1$, every coordinate satisfies $\alpha\le\abs{\theta_j}\le\pi/4$.
\item Every algorithm making at most $q$ queries to $U_\theta$, or at most $q$ queries to $A_\theta$, even with controlled queries, has the same mean acceptance probability under $\pi_0$ as under $\pi_1$.
\end{enumerate}
\end{lemma}

\begin{proof}
Let $a=(1-\cos\alpha)/2$ and $b=(1-\cos(\pi/4))/2$, and take $\mu_0,\mu_1$ from \cref{lem:scalar-priors} with $a$, $b$, and $D=2q$.
Let $\nu_0$ and $\nu_1$ be the laws of $\theta=\pm\arccos(1-2p)$, with an independent uniformly random sign, for $p\sim\mu_0$ and $p\sim\mu_1$, respectively.
Since $\Ex\sbra{e^{ir\theta}\mid p}=T_{\abs{r}}(1-2p)$ is a polynomial of degree $\abs{r}$ in $p$, the Fourier moments of $\nu_0$ and $\nu_1$ agree for $\abs{r}\leq2q$.
Atoms in $[a,b]$ give $\alpha\leq\abs{\theta}\leq\pi/4$, and the atom $0$ gives $\theta=0$.
By $\alpha\leq1/2$, we have $a\leq\alpha^2/4$ and $b-a=(\cos\alpha-\cos(\pi/4))/2\geq1/16$.
By \cref{lem:scalar-priors} and $2q\alpha\leq1/8$, $\nu_0\rbra{\R\setminus\cbra{0}}\leq(2q)^2a/(b-a)\leq4(2q\alpha)^2\leq1/16$.
Let $\pi_0=\nu_0^{\ot2k}$ and $\pi_1=\nu_1^{\ot2k}$.
Under $\pi_0$, the expected number of nonzero coordinates is at most $2k/16=k/8$, so Markov's inequality (\cref{fact:markov}) gives~(i).
Part~(ii) follows from the support of $\nu_1$.
For~(iii), we use the polynomial method of~\cite[Lemmas~4.1 and~4.2]{BBC+01} in a coordinatewise Fourier form.
Purify the initial state and refine each measurement record to include all Kraus labels and the final outcome.
For each complete record, known operations take linear combinations of the amplitudes of the unnormalized branch state.
A forward query changes each amplitude frequency by at most one in each coordinate. An inverse query negates that increment, and coherent choice among $\Id,U_\theta,U_\theta^{\dagger}$ preserves the same bound.
Each branch uses at most $q$ queries, so its amplitude frequencies lie in $[-q,q]^{2k}$, and the acceptance probability, a sum of squared norms of accepting branch states, has coordinate frequencies in $[-2q,2q]^{2k}$.
Known Hadamard conjugations preserve this bound, so it also holds for $A_\theta$.
Expanding the acceptance probability into monomials $\prod_je^{ir_j\theta_j}$ with $\abs{r_j}\le2q$, the independence of the coordinates and the matching Fourier moments of $\nu_0,\nu_1$ give~(iii).
\end{proof}

\begin{corollary}[Lower bound for Union Size Testing from a reversible sampler]\label{cor:cover-lower}
Let $k\ge1$, $n\ge2k$, and $0<\delta\le1/16$.
Given a reversible sampler $A$ for an unknown distribution $\calD$ on $2^{\sbra{n}}$, any tester for \cref{prob:set-testing} succeeding with probability at least $2/3$ requires $\Omega(\sqrt{k/\delta})$ queries, even with controlled queries.
\end{corollary}
\begin{proof}
Both promise classes are nonempty, so a zero-query tester is impossible.
Suppose that a tester makes at most $q$ queries, where $q\ge1$ is an integer with $q\le\sqrt{k/\delta}/32$.
Set $\alpha=2\sqrt{\delta/k}$, so that $\alpha\le1/2$ and $2q\alpha\le1/8$, and let $\pi_0,\pi_1$ be the sparse and dense priors of \cref{lem:product-rotation-priors}.
Use the samplers $A_\theta=H^{\ot2k}U_\theta H^{\ot2k}$ of \cref{lem:product-rotation-priors}, where $\theta\in\R^{2k}$ is drawn from $\pi_0$ or from $\pi_1$ and $U_\theta=\bigotimes_{j=1}^{2k}e^{i\theta_j\sigma_3}$, padding by identities if $n>2k$.
Since $H\sigma_3H=\sigma_1$, we have $A_\theta=\bigotimes_{j=1}^{2k}e^{i\theta_j\sigma_1}$, so these samplers need no workspace register and include coordinates independently with probabilities $\sin^2\theta_j$.

\noindent\emph{YES instances.}
By \cref{lem:product-rotation-priors}(i), at least $7/8$ of the mass of the sparse prior lies on inputs
with at most $k$ nonzero angles, which are YES instances.

\noindent\emph{NO instances.}
By \cref{lem:product-rotation-priors}(ii), on every input in the support of the dense prior and every $T\subseteq\sbra{n}$ with $\abs{T}\le k$,
\[
 \Pr\sbra*{X\nsubseteq T}
 =1-\prod_{j\notin T}\cos^2\theta_j
 \ge1-\cos^{2k}\alpha\ge1-e^{-4\delta}\ge\delta.
\]

\noindent\emph{Indistinguishability.}
The tester accepts with mean probability at least $(7/8)(2/3)>1/2$ under the sparse prior and at most $1/3$ under the dense prior.
By \cref{lem:product-rotation-priors}(iii), these two means are equal, a contradiction.
Hence every such tester makes more than $\sqrt{k/\delta}/32$ queries.
\end{proof}

By \cref{thm:reversible-subset-upper,cor:cover-lower}, the query complexity is $\Theta(\sqrt{k/\delta})$ for $n\geq2k$ and $0<\delta\le1/16$.

\section{Junta Boolean Testing}\label{sec:boolean}

This section studies the following problem in the query model of \cref{subsec:boolean-access}.

\begin{problem}[Junta Boolean Testing]\label{prob:boolean-junta-testing}
Fix integers $n\ge1$ and $0\le k\le n$, and fix $0<\eps\leq1$.
Given query access to an unknown Boolean function $f:\cbra{0,1}^n\to\cbra{0,1}$,
distinguish between the following two cases with probability at least $2/3$:
\begin{itemize}
  \item (YES) $f\in\calB_k$.
  \item (NO) $\distB(f,\calB_k)\ge\eps$.
\end{itemize}
\end{problem}

We prove the upper bound and then the matching lower bound for \cref{prob:boolean-junta-testing}.
The upper bound combines the reversible union size tester with partial support recovery. The lower bound encodes the moment-matched priors of \cref{lem:scalar-priors} in Boolean functions.

\subsection{Upper bound}\label{subsec:boolean-upper}

\begin{theorem}[Upper bound for Junta Boolean Testing]\label{thm:boolean-junta-upper-bound}
Fix integers $1\le k<n$ and $0<\eps\le1/8$. Given query access to a Boolean function $f:\cbra{0,1}^n\to\cbra{0,1}$, there is a quantum algorithm that distinguishes $f\in\calB_k$ from $\distB\rbra{f,\calB_k}\ge\eps$ with success probability at least $2/3$ and uses
\[
 O\!\rbra*{\min\cbra*{2^{n/2},
 \sqrt{\frac{1+\logplus(\eps2^k)}{\eps}}}}
\]
queries.
\end{theorem}

The tester has three branches.
If $2^{n/2}$ is the smaller term, it recovers the support exactly.
If $1+\logplus(\eps2^k)\ge k/32$, it applies \cref{thm:reversible-subset-upper} to Fourier sampling.
At smaller distances, partial support recovery
gives, with constant probability on a junta, either small residual
Fourier mass or few unrecovered relevant coordinates.
We test these alternatives separately. On a NO input, both tests have
their NO promises for every set returned by the recovery procedure.

\subsubsection{Partial support recovery}\label{subsec:partial-recovery}
Fix a known coordinate set $J\subseteq\sbra{n}$ and let $X,Y$ be
independent uniform inputs in $\cbra{0,1}^n$.
For $u\in\cbra{0,1}^{\abs{J}}$, put
$b_J(u)=\Ex\sbra{(-1)^{f(X)}\mid X_J=u}$.
Expectations over $u$ below use the uniform distribution.
Let $A_f$ be the Fourier sampler of \cref{lem:fourier-sampling}.
Let $S$ be the support of the outcome obtained by measuring
$A_f\ket0$ in the computational basis. Define
\[
 \nu_J=\Pr\sbra*{f(X)\ne f(Y)\mid X_J=Y_J},\qquad
 q_J=\Inf_{\overline J}\sbra{f}
     =\Pr\sbra*{S\nsubseteq J}.
\]
By Parseval's identity (\cref{fact:boolean-fourier}), $q_J=1-\Ex_u\sbra{b_J(u)^2}$, and by \cref{lem:boolean-influence-certificate}(i), $q_J\ge2\distB(f,\calB_J)$.
For $u\in\cbra{0,1}^{\abs{J}}$, let $\mu_u$ be the minority density of $f$ on the fiber $\cbra{x:x_J=u}$.

\begin{lemma}[Pair-fiber identities]\label{lem:pair-fiber}
For every $J\subseteq\sbra{n}$, we have $\nu_J=2^{-\abs{J}}\sum_u2\mu_u(1-\mu_u)$ and $q_J=2\nu_J$.
Moreover, $\mu_u\le2^{\abs{J}}\nu_J$ for every $u$, and $\nu_J\le\nu_{J'}$ whenever $J'\subseteq J$.
\end{lemma}

\begin{proof}
Averaging independent samples in each fiber gives the formula for $\nu_J$, and $\mu_u\le2\mu_u(1-\mu_u)$ gives the bound on $\mu_u$.
Since $q_J=1-\Ex_u\sbra{b_J(u)^2}$ and $b_J(u)^2=(1-2\mu_u)^2=1-4\mu_u(1-\mu_u)$, we get $q_J=2\nu_J$.
If $J'\subseteq J$, every outcome $S\nsubseteq J$ also satisfies $S\nsubseteq J'$. Hence $q_J\le q_{J'}$, and $q_J=2\nu_J$ gives $\nu_J\le\nu_{J'}$.
\end{proof}

\begin{fact}[Edge-isoperimetric inequality
  {\cite[Theorem~2.39]{OD14}}]
\label{fact:cube-edge-isoperimetry}
Let $h\geq1$ be an integer and let $\varnothing\ne B\subseteq\cbra{0,1}^h$
satisfy $\abs{B}\leq2^{h-1}$.
Let $\partial B$ be the set of unordered hypercube edges with one
endpoint in $B$ and the other outside $B$. Then
\[
 \abs{\partial B}\geq\abs{B}\log\rbra*{\frac{2^h}{\abs{B}}}.
\]
\end{fact}

We use Fourier-sampling support recovery~\cite[Lemma~4.4]{AS07} after random restrictions. For known $x$ and $J\subseteq\sbra{n}$, put
$D_J(x)=\cbra{j\notin J\colon f(x)\ne f(x\oplus e_j)}$ and
$N_k=\ceil{4k\ln(2k)}$.
For $R\subseteq\sbra{n}$ and $z\in\cbra{0,1}^n$, write $R\wedge z$
for the string agreeing with $z$ on $R$ and equal to zero elsewhere.
The subroutine \LocalRecover$(f,x,J;k)$ performs $N_k$ independent repetitions. Each repetition includes each coordinate of $\overline J$ in a new random set $R$ independently with probability $1/(2k)$ and Fourier-samples $g_R(z)=f(x\oplus(R\wedge z))$ with one query.
It returns the union $D\subseteq\sbra{n}\setminus J$ of the outcomes.

Pairs satisfying $X_J=Y_J$ are prepared by a known unitary. Two
queries to $O_f$ record $f(X)$ and $f(Y)$, making disagreement a
known success predicate. For a phase query to $g_R$, compute
$x\oplus(R\wedge z)$ in an additional register initialized to $\ket0$, apply $O_f$ with a
$\ket{-}$ target, and uncompute the register. Thus a restricted
Fourier sample and its inverse each use one query.

The procedure \PartialRecover$(f;k,\tau)$ starts from $J=\varnothing$ and uses the query cap $T_{\mathrm{rec}}=\ceil{C_{\mathrm{rec}}(\tau^{-1/2}+k^2\ln(ek))}$ for a sufficiently large absolute constant $C_{\mathrm{rec}}$.
While $\abs J<k$ and fewer than $T_{\mathrm{rec}}$ queries have been made, it performs a round.
A round fixes $J_0=J$, samples a uniform pair $(X,Y)$ with $X_{J_0}=Y_{J_0}$ conditioned on $f(X)\ne f(Y)$, and chooses $x$ uniformly from $\cbra{X,Y}$.
It then computes $D=\LocalRecover(f,x,J_0;k)$ and draws a Fourier sample $S$ of $(-1)^f$ conditioned on $S\nsubseteq J_0$.
Both conditioned samples use exact conditional sampling (\cref{lem:exact-conditional-sampling}).
If $J'=J_0\cup D\cup S$ has more than $k$ elements, the procedure returns \Overfull.
Otherwise, it sets $J=J'$.
When the loop ends, or when the cap interrupts a call, it returns the current $J$.
Each search and local-recovery call is interrupted immediately before
its next oracle query would exceed the cap $T_{\mathrm{rec}}$.
Only complete rounds update $J$. At zero success probability, the
corresponding search never succeeds and is eventually interrupted.

\begin{lemma}[Partial support recovery]\label{lem:partial-recovery}
Fix integers $1\le k\le n$ and $0<\tau\le1$.
For every Boolean function $f:\cbra{0,1}^n\to\cbra{0,1}$,
\PartialRecover$(f;k,\tau)$ uses at most
$\ceil{C_{\mathrm{rec}}(\tau^{-1/2}+k^2\ln(ek))}$ queries on every input,
for a sufficiently large absolute constant $C_{\mathrm{rec}}$.
It returns either a set $J\subseteq\sbra{n}$ with $\abs{J}\le k$ or \Overfull.

If $K=\supp(f)$ satisfies $\abs{K}\le k$, the procedure never returns \Overfull\ and always has $J\subseteq K$.
With probability at least $19/20$, its output also satisfies either
$\nu_J<\tau$ or $\abs{K\setminus J}\le H_\tau$, where
$H_\tau=\ceil{2\logplus(\tau2^k)}$.
\end{lemma}

\begin{proof}
On every Boolean input, each Fourier outcome of $g_R$ lies in $R$ and contains only relevant coordinates of $f$, so local recovery never introduces an irrelevant coordinate.
For a promised input, put $K=\supp(f)$. For any fixed $x,J$ and $j\in D_J(x)$,
\[
 \Pr\sbra*{R\cap(K\setminus J)=\cbra{j}}
 =\frac1{2k}\rbra*{1-\frac1{2k}}^{\abs{K\setminus J}-1}
 \ge\frac1{4k}.
\]
On this event $g_R(z)=f(x)\oplus z_j$, so the outcome is $\cbra{j}$.
The probability of missing $j$ in all $N_k$ repetitions is at most
$(1-1/(4k))^{N_k}\le1/(2k)$, so a union bound over at most $k$ sensitive coordinates gives $\Pr\sbra{D_J(x)\subseteq D}\ge1/2$.

For any fixed starting history with $J\subsetneq K$, its uncapped
searches have positive success probabilities $\nu_J$ and
$q_J=2\nu_J$. By exact conditional sampling (\cref{lem:exact-conditional-sampling}), they terminate almost
surely and return exact conditional samples, with expected query cost
$O(\nu_J^{-1/2})$. Local recovery uses $N_k$ queries.
Every completed round satisfies $J\subsetneq J'\subseteq K$, so $\abs{J'}\le k$ and the procedure never returns \Overfull. At $J=K$, we have $\nu_K=0$, so no further round completes.

\noindent\emph{Single-round progress.}
Call a round good if
$\abs{K\setminus J'}\le2\log(\nu_J2^k)$.
Condition on the fiber $u$ selected by the disagreeing-pair sampler,
and put $h=\abs{K\setminus J}$. The selected endpoint lies in a minority
level set $B$ with probability at least $1/2$ and is then uniform in $B$.
Identify the fiber with $\cbra{0,1}^h$ through the $h$ remaining relevant coordinates.
By the edge-isoperimetric inequality (\cref{fact:cube-edge-isoperimetry}),
$\Ex\sbra{\abs{D_J(x)}\mid u,x\in B}=\abs{\partial B}/\abs{B}\ge\log(1/\mu_u)$,
where $\mu_u=\abs{B}/2^h$. Hence
\cref{lem:pair-fiber} and $h+\abs{J}\le k$ give
\[
 \Ex\sbra*{h-\abs{D_J(x)}\mid u,x\in B}
 \le h+\log\mu_u\le\log(\nu_J2^k).
\]
Markov's inequality (\cref{fact:markov}) bounds the residual by twice the right side with
probability at least $1/2$. If that side is zero, the nonnegative
residual is zero almost surely. Local recovery finds all sensitive
coordinates with probability at least $1/2$, and the additional Fourier
sample can only help. Multiplying these three conditional bounds shows
that a round is good with probability at least $1/8$, conditional on
any starting history. This probability concerns the uncapped round,
without conditioning on completion before an external cap.

\noindent\emph{Total cost.}
Analyze the uncapped process until either stated alternative first holds.
Before then, $\nu_J\ge\tau>0$, $J\subsetneq K$, and each round adds a
coordinate, so there are at most $k$ such rounds. Group them by
$2^i\tau\le\nu_J<2^{i+1}\tau$, $i\ge0$. Monotonicity means each band
is visited only once. The number of rounds through the first good one, or stopping earlier
on exit from the band, is dominated by a geometric variable of mean $8$.
A good round in band $i$ leaves
\[
 \abs{K\setminus J'}\le2\log(\nu_J2^k)
 <2(i+1)+2\log(\tau2^k)\le H_\tau+2(i+1)
\]
unknown coordinates, allowing at most $2(i+1)$ further rounds before the goal holds.
The expected number of rounds in this band is therefore at most
$8+2(i+1)$. Conditioning on each round's starting history and using
the tower property bounds the expected number of search queries by
\[
 O(\tau^{-1/2})\sum_{i\ge0}(8+2(i+1))2^{-i/2}=O(\tau^{-1/2}).
\]
Local recovery adds at most $kN_k$ queries. Thus, if $Q^\star$ is the
number of oracle queries until the goal first holds,
$\Ex\sbra{Q^\star}\le C_1(\tau^{-1/2}+kN_k)$ for an absolute $C_1$.
Choose $C_{\mathrm{rec}}$ so that
$T_{\mathrm{rec}}\ge20C_1(\tau^{-1/2}+kN_k)$, using
$kN_k=O(k^2\ln(ek))$.
Couple the capped and uncapped processes using the same randomness.
On $Q^\star\le T_{\mathrm{rec}}$, the capped process completes all
rounds needed to reach the goal. Interruption returns the last
completed union, and the goal is absorbing because both $\nu_J$ and the
residual support decrease. Thus the failure probability is at most
$\Pr\sbra{Q^\star>T_{\mathrm{rec}}}\le1/20$ by Markov's inequality (\cref{fact:markov}).
The global cap bounds the number of oracle queries on every input,
including inputs outside the promise.
\end{proof}

\begin{corollary}[Exact support recovery]\label{cor:exact-support}
Fix integers $1\le k\le n$. Given query access to a Boolean function $f:\cbra{0,1}^n\to\cbra{0,1}$ promised to depend on at most $k$ coordinates, there is a quantum algorithm that outputs $\supp(f)$ with probability at least $2/3$ using $O(2^{k/2})$ queries.
\end{corollary}
\begin{proof}
Put $K=\supp(f)$ and take $\tau=2^{-k}$, so that the bound $H_\tau=\ceil{2\logplus(\tau2^k)}$ of \cref{lem:partial-recovery} equals $0$.
Run \PartialRecover$(f;k,\tau)$ and let $J\subseteq K$ be its output.
Whenever $J\subsetneq K$, write
$s=\abs{K}$ and $h=s-\abs{J}\ge1$. At least one fiber is nonconstant and has
minority density at least $2^{-h}$. By \cref{lem:pair-fiber},
\[
 \nu_J\ge2^{-\abs{J}}\,2^{1-h}(1-2^{-h})\ge2^{-s}\ge2^{-k}.
\]
Thus, under either alternative in \cref{lem:partial-recovery}, $J=K$,
which holds with probability at least $19/20$. The query bound is
$O(2^{k/2}+k^2\ln(ek))=O(2^{k/2})$, with an absolute constant for
integer $k\ge1$.
\end{proof}

\subsubsection{Proof of the upper bound}\label{subsec:boolean-tester}
For $J\subseteq\sbra{n}$, the projector $\Pi_{\overline J}$ of \cref{subsec:cover-upper-reversible} projects onto the Fourier outcomes $S$ that intersect $\overline J$, that is, $S\nsubseteq J$, so $\norm{\Pi_{\overline J}A_f\ket0}^2=q_J$.
Let $A_{f,J}$ apply $A_f$ and then compute $S\setminus J$ into an additional register initialized to $\ket0$, keeping $S$.
It is a reversible sampler for the law of $S\setminus J$, and both $A_{f,J}$ and $A_{f,J}^{\dagger}$ use one query.
\Cref{alg:boolean-tester} summarizes the tester in the following proof, which combines exact support recovery, union size testing, and partial support recovery.

\begin{algorithm}[h]
\caption{\textsc{BooleanJuntaTest}$(f;n,k,\eps)$}\label{alg:boolean-tester}
\begin{algorithmic}[1]
\Require Query access to $O_f$ on $n$ bits, integers $1\le k<n$, and $0<\eps\le1/8$
\Ensure \Accept\ or \Reject
\State $L\gets1+\logplus(\eps2^k)$
\If{$2^{n/2}\le\sqrt{L/\eps}$} \Comment{finite-domain branch}
  \State $K\gets$ exact support of $f$ with junta-size bound $n$ \Comment{\cref{cor:exact-support}}
  \State \Return \Accept\ iff $\abs{K}\le k$
\EndIf
\If{$L\ge k/32$} \Comment{sampling branch}
  \State \Return \textsc{ReversibleUnionSizeTest}$(A_f;n,k,2\eps)$ \Comment{\cref{thm:reversible-subset-upper}}
\EndIf
\State $J\gets\PartialRecover(f;k,\eps/64)$ \Comment{recovery branch, \cref{lem:partial-recovery}}
\If{$J=\Overfull$} \State \Return \Reject \EndIf
\State $h\gets\min\cbra{\ceil{2\logplus(\eps2^k/64)},k-\abs{J}}$
\If{quantum amplitude estimation on $A_f,\Pi_{\overline J}$ with $\tau=2\eps$, $\beta=1/100$ answers \SMALL}
  \State \Return \Accept \Comment{\cref{lem:amplitude-estimation}}
\EndIf
\If{$h\ge1$}
  \State \Return majority of $O(1)$ runs of \textsc{ReversibleUnionSizeTest}$(A_{f,J};n,h,2\eps)$
\EndIf
\State \Return \Accept\ iff amplitude amplification on $A_f,\Pi_{\overline J}$ with $\zeta=2\eps$, $\beta=1/100$ has no success
\end{algorithmic}
\end{algorithm}

\begin{proof}[Proof of \cref{thm:boolean-junta-upper-bound}]
\noindent\emph{Algorithm.} Set $L=1+\logplus(\eps2^k)$.
In the finite-domain branch, $2^{n/2}\le\sqrt{L/\eps}$, and the tester recovers the exact support using
\cref{cor:exact-support} with junta-size bound $n$ and accepts exactly when its
size is at most $k$.
Otherwise, in the sampling branch, $L\ge k/32$, and the tester applies \cref{thm:reversible-subset-upper} to the Fourier
sampler $A_f$ on $\sbra{n}$ with size bound $k$ and threshold $2\eps$.
Otherwise, in the recovery branch, the tester runs \PartialRecover$(f;k,\eps/64)$ and rejects
on \Overfull. If it returns a set $J$, the tester sets $h=\min\cbra{H_{\eps/64},k-\abs{J}}$, where $H_{\eps/64}=\ceil{2\logplus(\eps2^k/64)}$ is the bound of \cref{lem:partial-recovery} with $\tau=\eps/64$.
It accepts exactly when at least one of the following two tests, run with independent randomness, accepts.
\begin{enumerate}[label=(\roman*)]
\item The first test uses quantum amplitude estimation (\cref{lem:amplitude-estimation}) for $A_f$ and $\Pi_{\overline J}$ with $\tau=2\eps$ and $\beta=1/100$, accepting on the small-mass side.
\item For $h\ge1$, the second test applies \cref{thm:reversible-subset-upper} to $A_{f,J}$ on $\sbra{n}$ with size bound $h$ and threshold
$2\eps$, boosted to error $1/100$ by a majority of $O(1)$ independent runs and Hoeffding's inequality (\cref{fact:hoeffding}). For $h=0$, it instead uses
amplitude amplification (\cref{lem:amplitude-amplification}) for $A_f$ and $\Pi_{\overline J}$ with $\zeta=2\eps$ and $\beta=1/100$, accepting when no measurement succeeds.
\end{enumerate}

\noindent\emph{Correctness.}
Every Boolean function is an $n$-junta, so the finite-domain branch
decides membership in $\calB_k$ with probability at least $2/3$.
In the sampling branch, a Fourier sample $S$ of a $k$-junta $f$ satisfies $S\subseteq\supp(f)$.
If $\distB(f,\calB_k)\geq\eps$, then by \cref{lem:boolean-influence-certificate}(i), $\Pr\sbra{S\nsubseteq T}=q_T\geq2\eps$ for every $T\subseteq\sbra{n}$ with $\abs{T}\leq k$.
Thus the promises of \cref{prob:set-testing} hold with $\delta=2\eps$, and this branch has error at most $1/10$.

In the recovery branch, the first test distinguishes $q_J\le\eps/2$ from
$q_J\ge2\eps$. Both tests have error at most $1/100$ on their
respective promises.

\noindent\emph{YES instances.}
For a YES input, put $K=\supp(f)$. Then \PartialRecover\ never returns \Overfull\ and
returns $J\subseteq K$. By \cref{lem:partial-recovery}, with probability at least $19/20$, either
$q_J=2\nu_J<\eps/32$, giving the first test its YES promise, or
$\abs{K\setminus J}\le h$, giving the second test its YES promise.
Thus the error is at most $1/20+1/100$.

\noindent\emph{NO instances.}
On a NO input, for every returned $J$ and every $\abs{W}\le h$,
$\abs{J\cup W}\le k$, and \cref{lem:boolean-influence-certificate}(i) gives
\[
 q_{J\cup W}\ge2\distB(f,\calB_{J\cup W})\ge2\eps.
\]
Both tests have their NO promises, so their OR has error at most
$1/50$. Rejection on \Overfull\ is also correct.
In every branch the success probability is at least $2/3$.

\noindent\emph{Query complexity.}
Each call to $A_f$ or its inverse uses one query to $f$.
The finite-domain branch uses $O(2^{n/2})$ queries.
The sampling branch uses $O(\sqrt{k/\eps})=O(\sqrt{L/\eps})$
queries since $L\ge k/32$.

In the recovery branch, $L<k/32$ gives $\eps<2^{-31k/32}$.
Thus \cref{lem:partial-recovery} bounds recovery by
$O(\eps^{-1/2}+k^2\ln(ek))=O(\eps^{-1/2})$ queries.
The two tests use $O(\eps^{-1/2})$ and
$O(\sqrt{\max\cbra{1,h}/\eps})$ queries, respectively.
Since $h\le H_{\eps/64}=O(L)$ and $L\ge1$, the recovery branch uses
$O(\sqrt{L/\eps})$ queries.
The branch conditions therefore give a worst-case query bound
$O(\min\cbra{2^{n/2},\sqrt{L/\eps}})$ on every input.
\end{proof}

\subsection{Lower bound}\label{subsec:boolean-lower}

We first record the search lower bound used in the reductions below.

\begin{lemma}[Unique-search lower bound {\cite{BBBV97}}]
\label{lem:search-lower-bound}
For an integer $m\ge1$, distinguishing the zero function
$g:\cbra{0,1}^m\to\cbra{0,1}$ from one with a unique marked input,
with success probability at least $2/3$, requires
$\Omega(2^{m/2})$ Boolean oracle queries.
\end{lemma}

\begin{theorem}[Lower bound for Junta Boolean Testing]\label{thm:boolean-junta-lower-bound}
Fix integers $k\ge1$ and $n\ge2k$, and fix $0<\eps\le1/8$. Any tester that, given query access to a Boolean function $f:\cbra{0,1}^n\to\cbra{0,1}$, distinguishes $f\in\calB_k$ from $\distB\rbra{f,\calB_k}\ge\eps$ with success probability at least $2/3$ requires
\[
 \Omega\!\rbra*{\min\cbra*{2^{n/2},
 \sqrt{\frac{1+\logplus(\eps2^k)}{\eps}}}}
\]
queries.
\end{theorem}

We prove that the upper bound of \cref{thm:boolean-junta-upper-bound} is optimal by encoding the priors of \cref{lem:scalar-priors} in parity-address functions.
Counting address coordinates toward the junta size gives the logarithmic dependence.
A search reduction covers the remaining regimes of small distance and of small $n$.

\subsubsection{The parity-address construction}\label{subsec:parity-address}
We adapt the addressing construction in the proof of~\cite[Theorem~3.5]{AS07}, also used
in~\cite[Definition~6.2]{ABRW16}.
For integers $r\ge0$ and $L_0\ge1$, set $M=2^r$ and index the rows of
$X\in\cbra{0,1}^{M\times L_0}$ by $\cbra{0,1}^r$.
For $y\in\cbra{0,1}^{L_0}$ and $z\in\cbra{0,1}^r$, define
\begin{equation*}
 F_X(y,z)=\bigoplus_{i=1}^r z_i\ \oplus\!\bigoplus_{j=1}^{L_0}y_jX_{zj}.
\end{equation*}
Choose independent latent parameters $p_j\in[0,1]$ and, conditional on them, independent entries $X_{zj}\sim\mathrm{Bern}(p_j)$.

\begin{lemma}[Parity-address lower bound]\label{lem:address-lower}
Let $k,h,n$ be integers with $1\le h\le k$ and $n\ge k+h$, and let $0<\eps\le1/8$ satisfy $16h2^h\le\eps2^k$.
Every tester for \cref{prob:boolean-junta-testing} requires more than $\sqrt{h/\eps}/128$ queries in the worst case, even with controlled queries.
\end{lemma}
\begin{proof}
Both promise classes are nonempty, so a zero-query tester is impossible.
Set $r=k-h$, $M=2^r$, and $L_0=2h$, and pad $F_X$ by $n-k-h\ge0$ unused input coordinates to obtain a function on $n$ bits.
Then $M\ge16h/\eps$.
Suppose an integer $q\ge1$ is at most $\sqrt{h/\eps}/128$.
This assumption ensures that $a=8\eps/h\le1/2048$, so take $\mu_0,\mu_1$ from \cref{lem:scalar-priors} with $a$, $b=1/2$, and $D=2q$.
Draw the $p_j$ independently from $\mu_0$, or independently from $\mu_1$, and then generate $X$ as above.
By \cref{lem:scalar-priors}, the nonzero mass of $\mu_0$ is at most
\[
 \frac{D^2a}{b-a}\le128q^2\eps/h\le1/128.
\]

\noindent\emph{YES instances.}
When the $p_j$ are drawn from $\mu_0$, Markov's inequality (\cref{fact:markov}) shows that at most $h$ of them are nonzero
with probability at least $7/8$. Every other column of $X$ is zero,
so $F_X$ depends on at most $r+h=k$ variables.

\noindent\emph{NO instances.}
When the $p_j$ are drawn from $\mu_1$, every $p_j\ge a$.
Call a row \emph{uncovered} by a set of data coordinates if it has a parity variable outside that set.
For every set of at most $h$ data coordinates, each row is uncovered by it with conditional probability at least $1-e^{-ah}=1-e^{-8\eps}\ge4\eps$.
Conditional on the latent parameters, these events are independent across rows.
By the Chernoff bound (\cref{fact:chernoff}), the probability that fewer than
$2\eps M$ rows are uncovered by a fixed set is at most $e^{-M\eps/2}\le e^{-8h}$.
A union bound over at most $4^h$ data sets has failure at most $e^{-6h}$, uniformly in the latent parameters.

Across an address edge, $F_X$ changes by one plus a parity of the data bits, so the change is either identically one or balanced. Hence every address coordinate of $F_X$ has influence at least $1/2$.
By \cref{lem:boolean-influence-certificate}(ii), a candidate $k$-junta omitting an address bit has distance at least $1/4\ge\eps$.
Otherwise it depends on at most $h$ data bits.
Outside the exceptional event above, at least $2\eps M$ rows are uncovered by these data bits.
On each such row the candidate misses a parity variable, so it disagrees with $F_X$ on exactly half of the data assignments of that row.
Hence its distance from $F_X$ is at least $\eps$.
Thus, when the $p_j$ are drawn from $\mu_1$, $F_X$ is a NO instance with probability at least $1-e^{-6h}$.

\noindent\emph{Indistinguishability.}
We use the polynomial method of~\cite{BBC+01}.
After deferring intermediate measurements, the initial amplitudes have
degree zero in the truth-table signs. Known gates take linear
combinations of amplitudes, and each entry of a query matrix, including
a controlled query, is affine in those signs. Induction gives amplitude
degree at most $q$, so acceptance probability has degree at most $2q$.
This also covers randomized testers by averaging over their randomness.

For a monomial $\prod_{t=1}^{\ell}(-1)^{F_X(y^{(t)},z^{(t)})}$,
$\ell\le2q$, put
\[
 c_{zj}=\bigoplus_{t:z^{(t)}=z}y_j^{(t)},\qquad
 m_j=\sum_{z\in\cbra{0,1}^r}c_{zj}\le\ell,
\]
with empty XORs equal to zero. Conditional independence gives its
expectation as
\[
 (-1)^{\sum_{t=1}^{\ell}\sum_{i=1}^{r}z_i^{(t)}}
 \prod_{j=1}^{L_0}(1-2p_j)^{m_j}.
\]
Independence of the $p_j$ factors the average into univariate
expectations of degree at most $2q$. The matching scalar moments
therefore make the two average acceptance probabilities identical.
However, the YES and NO guarantees above give mean acceptance at least $(7/8)(2/3)>1/2$ when the $p_j$ are drawn from $\mu_0$
and at most $1/3+e^{-6h}<1/2$ when they are drawn from $\mu_1$, a contradiction.
\end{proof}

\subsubsection{Proof of the lower bound}\label{subsec:boolean-lower-proof}
\begin{proof}[Proof of \cref{thm:boolean-junta-lower-bound}]
Set
\[
 r=\min\cbra{n-1,\floor{\log(1/(2\eps))}},
 \qquad \ell=\max\cbra{1,k+1-r}.
\]
For $y\in\cbra{0,1}^\ell$, write
$\operatorname{Parity}_\ell(y)=\bigoplus_{j=1}^\ell y_j$.
Let $g:\cbra{0,1}^r\to\cbra{0,1}$ be an oracle promised to be zero
or to have a unique marked input. Define
$f_g(z,y)=g(z)\operatorname{Parity}_\ell(y)$ and pad by unused input
coordinates to obtain a function on $n$ bits.

\noindent\emph{YES instances.}
If $g=0$, then $f_g=0\in\calB_k$.

\noindent\emph{NO instances.}
If $g$ has a unique marked address $a$, then
$f_g(z,y)$ equals $\operatorname{Parity}_\ell(y)$ when $z=a$ and zero otherwise.
There are $r+\ell>k$ relevant coordinates, at most $n$, and each has influence $2^{-r}$.
Any $k$-junta misses one, so by \cref{lem:boolean-influence-certificate}(ii), it has distance at least $2^{-(r+1)}$.
Since $2^{-(r+1)}\ge\eps$, these are NO instances of
\cref{prob:boolean-junta-testing}.

\noindent\emph{Completing the reduction.}
A query to $f_g$ is simulated by computing $g(z)$ with one query,
conditionally adding $\operatorname{Parity}_\ell(y)$ to the output bit,
and uncomputing $g(z)$ with a second query.
A junta tester therefore distinguishes $g=0$ from a function with a
unique marked input using twice as many queries.
By the unique-search lower bound (\cref{lem:search-lower-bound}) with $m=r$, every junta tester requires
\begin{equation}\label{eq:boolean-search-lower}
 \Omega\!\rbra*{\min\cbra{2^{n/2},\eps^{-1/2}}}
\end{equation}
queries.

Now choose the largest $h\in\cbra{0,\ldots,k}$ satisfying $16h2^h\le\eps2^k$. Since $\eps\le1/8$, one has $h<k$. Maximality and $\log(h+1)\le h$ give
\[
 1+\logplus(\eps2^k)\le2h+6.
\]
For $h=0$, \cref{eq:boolean-search-lower} suffices since
$1+\logplus(\eps2^k)\le6$. For $h\ge1$, we have $1+\logplus(\eps2^k)\le2h+6\le8h$, and $n\ge2k$ allows
\cref{lem:address-lower}, giving
$\Omega(\sqrt{(1+\logplus(\eps2^k))/\eps})$ and hence its minimum with $2^{n/2}$.
\end{proof}

A direct search reduction also yields the matching lower bound for exact support recovery.
\begin{corollary}[Exact support recovery lower bound]\label{cor:exact-support-lower}
For $1\le k\le n$, recovering the exact support of a promised Boolean
$k$-junta with probability at least $2/3$ requires
$\Omega(2^{k/2})$ quantum queries, even when every relevant
coordinate is known to lie in $\sbra{k}$.
\end{corollary}
\begin{proof}
Let $g:\cbra{0,1}^k\to\cbra{0,1}$ be an oracle promised to be zero
or to have a unique marked input, and set $f(x)=g(x_{\sbra{k}})$.

\noindent\emph{Unique marked input.}
If $g$ has a unique marked input $a$, then
$f(x)=1$ exactly when $x_{\sbra{k}}=a$, so its exact support is $\sbra{k}$.

\noindent\emph{Zero function.}
If $g=0$, then $f=0$ has exact support $\varnothing$.
Both types of target functions are $k$-juntas.

\noindent\emph{Completing the reduction.}
One query to $f$ is simulated by one query to $g$ on the first $k$
input coordinates. An exact-support recoverer distinguishes the two
cases by declaring that a marked input exists precisely when the
returned support is $\sbra{k}$, with the same success probability.
By the unique-search lower bound (\cref{lem:search-lower-bound}) with $m=k$, $\Omega(2^{k/2})$ queries are required.
\end{proof}

\section{Junta Unitary Testing}\label{sec:unitary}

This section studies the following problem in the two unitary access models of \cref{subsec:unitary-access}.

\begin{problem}[Junta Unitary Testing]
\label{prob:unitary-junta-testing}
Fix integers $n\ge1$ and $0\le k\le n$, and fix $0<\eps\leq1$.
Given access to an unknown $n$-qubit unitary $U\in\mathrm U\rbra{N}$ in the standard or forward-only model,
distinguish between the following two cases with probability at least $2/3$:
\begin{itemize}
  \item (YES) $U\in\calU_k$.
  \item (NO) $\distU\rbra{U,\calU_k}\geq\eps$.
\end{itemize}
\end{problem}

We first treat \cref{prob:unitary-junta-testing} in the standard model and then in the forward-only model.
For each access model, we prove the upper bound before the matching lower bound.

\subsection{Upper bound in the standard model}\label{subsec:unitary-upper-inverse}

\begin{theorem}[Standard-model upper bound for Junta Unitary Testing]\label{thm:unitary-junta-upper-bound}
Fix integers $1\le k\le n$ and $0<\eps\le1$. Given access to an $n$-qubit unitary $U$ in the standard model, there is a quantum algorithm that distinguishes $U\in\calU_k$ from $\distU\rbra{U,\calU_k}\ge\eps$ with success probability at least $2/3$ and uses $O(\sqrt{k}/\eps)$ queries.
\end{theorem}

We apply \cref{thm:reversible-subset-upper} to the Pauli sampler with $\delta=\eps^2$.
\Cref{alg:unitary-inverse} summarizes the tester constructed in the following proof.

\begin{algorithm}[h]
\caption{\textsc{UnitaryJuntaTest}$\rbra{U;n,k,\eps}$}\label{alg:unitary-inverse}
\begin{algorithmic}[1]
\Require Access in the standard model to an $n$-qubit unitary $U$, integers $1\le k\le n$, and $0<\eps\le1$
\Ensure \Accept\ or \Reject
\If{$k=n$} \State \Return \Accept \EndIf
\State $A\gets A_U$ followed by writing $\supp\rbra{x}$ of its output label $x$ into an additional register \Comment{Pauli sampler, \cref{lem:pauli-sampling}}
\State \Return \textsc{ReversibleUnionSizeTest}$\rbra{A;n,k,\eps^2}$ \Comment{\cref{thm:reversible-subset-upper}}
\end{algorithmic}
\end{algorithm}

\begin{proof}[Proof of \cref{thm:unitary-junta-upper-bound}]
For $k=n$, accept without queries. Otherwise use the Pauli sampler $A_U$ of \cref{lem:pauli-sampling}, whose output register holds a label $x\in\Z_4^n$ with amplitude $\wh U\rbra{x}$.
Compute $X=\supp\rbra{x}$ into an additional register and keep $x$. For every $T\subseteq\sbra{n}$, its uncovered probability is
\[
 \Pr\sbra*{X\nsubseteq T}
 =\sum_{\supp(x)\nsubseteq T}\abs{\widehat U(x)}^2
 =\Inf_{\overline{T}}\sbra{U}.
\]
If $U\in\calU_T$ with $\abs{T}\le k$, then $\wh U\rbra{x}=0$ unless $\supp\rbra{x}\subseteq T$, since nonidentity Pauli matrices have trace zero. Hence $X\subseteq T$ almost surely. If
$\distU(U,\calU_k)\ge\eps$, \cref{lem:influence-certificate}(i) gives
the NO promise with $\delta=\eps^2$.
Thus \cref{thm:reversible-subset-upper} gives
$O(\sqrt{k}/\eps)$ queries and success probability at least $9/10>2/3$.
\end{proof}

\subsection{Lower bound in the standard model}\label{subsec:unitary-lower-inverse}

\begin{theorem}[Standard-model lower bound for Junta Unitary Testing]\label{thm:unitary-junta-lower-bound}
Fix integers $k\ge1$ and $n\ge2k$, and fix $0<\eps\le1/4$. Any tester that, given access to an $n$-qubit unitary $U$ in the standard model, distinguishes $U\in\calU_k$ from $\distU\rbra{U,\calU_k}\ge\eps$ with success probability at least $2/3$ requires $\Omega(\sqrt{k}/\eps)$ queries.
The bound holds even with controlled queries.
\end{theorem}

Product rotations and finite moment-matched priors show that the upper bound of \cref{thm:unitary-junta-upper-bound} is optimal.

\begin{proof}[Proof of \cref{thm:unitary-junta-lower-bound}]
Both promise classes are nonempty, so a zero-query tester is impossible.
Suppose that a tester makes at most $q$ queries, where $q\ge1$ is an integer with $q\le\sqrt{k}/(32\eps)$.
Set $\alpha=2\eps/\sqrt k$, so that $\alpha\le1/2$ and $2q\alpha\le1/8$, and let $\pi_0,\pi_1$ be the sparse and dense priors of \cref{lem:product-rotation-priors} on $2k$ coordinates.
The tester receives $U_\theta\ot I$, where $\theta\in\R^{2k}$ is drawn from $\pi_0$ or from $\pi_1$, $U_\theta$ of \cref{lem:product-rotation-priors} acts on the first $2k$ qubits, and we write $\theta_j=0$ for $j>2k$.

\noindent\emph{YES instances.}
By \cref{lem:product-rotation-priors}(i), at least $7/8$ of the mass of the sparse prior lies on inputs
with at most $k$ nonzero angles, which lie in $\calU_k$.

\noindent\emph{NO instances.}
For every $T\subseteq\sbra{n}$, the normalized trace against any
$V_T\ot I_{\overline T}$ factors into $\prod_{j\notin T}\cos\theta_j$ and a
normalized unitary overlap on $T$. The latter is at most one, attained
by the product unitary restricted to $T$. Thus
\begin{equation}\label{eq:product-distance}
 \distU(U_\theta,\calU_T)^2
 =1-\prod_{j\notin T}\abs{\cos\theta_j}.
\end{equation}
By \cref{lem:product-rotation-priors}(ii), in the dense prior every $k$-set omits at least $k$ angles of magnitude at least $\alpha$, and \cref{eq:product-distance} gives
\[
 \distU(U_\theta,\calU_k)^2
 \ge1-\cos^k\alpha\ge1-e^{-2\eps^2}\ge\eps^2.
\]

\noindent\emph{Indistinguishability.}
The tester accepts with mean probability at least $(7/8)(2/3)>1/2$ under the sparse prior and at most $1/3$ under the dense prior.
By \cref{lem:product-rotation-priors}(iii), these two means are equal, a contradiction.
Hence every such tester makes more than $\sqrt k/(32\eps)$ queries.
\end{proof}

\subsection{Upper bound in the forward-only model}\label{subsec:unitary-upper-forward}

The next procedure of~\cite{BY25} turns one channel query into a subset sample whose probability of leaving a set $T$ is comparable to the influence of $\overline T$.

\begin{lemma}[Influence sampling {\cite[Lemma~20]{BY25}}]
\label{lem:influence-sample-round}
Fix a channel $\Phi\in C\rbra{\calH}$.
The procedure \InfluenceSample$\rbra{\Phi}$ draws $V$ uniformly from $\cbra{I_2,H,e^{-i\pi\sigma_1/4}}$ and, independently, $z$ uniformly from $\cbra{0,1}^n$.
It applies $\Phi$ to the product state $\bigotimes_{i=1}^nV\ketbra{z_i}{z_i}V^{\dagger}$ and measures each qubit $i$ of the output in the basis $\cbra{V\ket0,V\ket1}$, obtaining the outcome $w_i\in\cbra{0,1}$ that corresponds to $V\ket{w_i}$.
It returns $X=\cbra{i\in\sbra{n}\colon w_i\neq z_i}$.
For every $T\subseteq\sbra{n}$, this output satisfies
\begin{equation*}
  \frac23\Inf_{\overline{T}}\sbra*{\Phi}\leq\Pr\sbra*{X\nsubseteq T}\leq\Inf_{\overline{T}}\sbra*{\Phi}.
\end{equation*}
\end{lemma}

\begin{theorem}[Forward-only upper bound for Junta Unitary Testing]
\label{thm:unitary-junta-upper-bound-forward}
Fix integers $1\le k\le n$ and $0<\eps\le1$. Given forward-only access to an $n$-qubit unitary $U$, there is a nonadaptive quantum algorithm that distinguishes $U\in\calU_k$ from $\distU\rbra{U,\calU_k}\ge\eps$ with success probability at least $2/3$ and uses $O(k/(\eps^2\log(ek)))$ queries.
Every query acts on a product input state without ancillas and is followed by single-qubit measurements.
\end{theorem}

Forward queries to $U$ yield independent subset samples through influence sampling of $\Phi_U$.
\Cref{thm:set-testing} then decides whether these samples can be covered by at most $k$ coordinates.
\Cref{alg:unitary-forward} summarizes the tester constructed in the following proof.

\begin{algorithm}[h]
\caption{\textsc{ForwardUnitaryJuntaTest}$\rbra{U;n,k,\eps}$}\label{alg:unitary-forward}
\begin{algorithmic}[1]
\Require Forward-only access to an $n$-qubit unitary $U$, integers $1\le k\le n$, and $0<\eps\le1$
\Ensure \Accept\ or \Reject
\State $\delta\gets2\eps^2/3$
\State $t_{\max}\gets$ largest number of samples used by \textsc{UnionSizeTest}$\rbra{\sbra{n},k,\delta}$ \Comment{$t_{\max}=O(k/(\delta\log(ek)))$}
\For{$t=1,\ldots,t_{\max}$}
  \State $X_t\gets\InfluenceSample\rbra{\Phi_U}$ \Comment{one forward query, \cref{lem:influence-sample-round}}
\EndFor
\State \Return \textsc{UnionSizeTest}$\rbra{\sbra{n},k,\delta}$ run on $X_1,\ldots,X_{t_{\max}}$ \Comment{\cref{thm:set-testing}}
\end{algorithmic}
\end{algorithm}

\begin{proof}[Proof of \cref{thm:unitary-junta-upper-bound-forward}]
Apply \cref{thm:set-testing} with $\Omega=\sbra{n}$ and $\delta=2\eps^2/3$, generating each requested sample $X$ by an independent run of \InfluenceSample$\rbra{\Phi_U}$ from \cref{lem:influence-sample-round}.
If $U\in\calU_{T_*}$ with $\abs{T_*}\leq k$, then $\Phi_U\in\calC_{T_*}$, so $\Inf_{\overline{T_*}}\sbra{\Phi_U}=0$ by \cref{lem:junta-fourier-support} and $X\subseteq T_*$ almost surely by \cref{lem:influence-sample-round}.
If $\distU\rbra{U,\calU_k}\geq\eps$, then by \cref{lem:influence-certificate}(i), $\Inf_{\overline T}\sbra{U}\geq\eps^2$ for every $\abs{T}\leq k$.
Since $\Inf_{\overline T}\sbra{\Phi_U}=\Inf_{\overline T}\sbra{U}$, by \cref{lem:influence-sample-round}, $\Pr\sbra{X\nsubseteq T}\geq\delta$.
Thus \cref{thm:set-testing} applies and succeeds with probability at least $2/3$ using $O(k/(\eps^2\log(ek)))$ samples.
Each run of \InfluenceSample\ applies one forward query to a product input state without ancillas, followed by single-qubit measurements, so the algorithm makes $O(k/(\eps^2\log(ek)))$ queries.
Since the sample bound of \cref{thm:set-testing} holds on every execution, all runs can be performed before their outputs are processed, so the algorithm is nonadaptive.
\end{proof}

\subsection{Lower bound in the forward-only model}\label{subsec:unitary-lower-forward}

\begin{theorem}[Forward-only lower bound for Junta Unitary Testing]
\label{thm:unitary-junta-lower-bound-forward}
Fix a sufficiently large integer $k$, an integer $n\ge2k$, and $2^{-k/16}\le\eps\le1/16$. Any tester that, given forward-only access to an $n$-qubit unitary $U$, distinguishes $U\in\calU_k$ from $\distU\rbra{U,\calU_k}\ge\eps$ with success probability at least $2/3$ requires $\Omega(k/(\eps^2\log k))$ queries.
\end{theorem}

We show that the upper bound of \cref{thm:unitary-junta-upper-bound-forward} is optimal even among adaptive quantum testers with arbitrary
ancillas, in the stated parameter range. The proof reduces \cref{prob:classical-support-testing} to \cref{prob:unitary-junta-testing}
through a phase-masked unitary ensemble and applies \cref{thm:classical-support-lb}.

\subsubsection{Construction of the hard ensemble}\label{subsec:hard-construction}
We construct a randomized instance map from support size testing
to testing junta unitaries.
In the construction, each address selects
a data qubit on which the oracle applies Pauli $Z$.
A phase mask contributes an additional phase depending only
on the address.

Fix integers $M\geq2$ and $r\geq1$.
The construction uses an $r$-qubit address register and an
$M$-qubit data register, so $n=r+M$.
We index the address qubits by $\sbra{r}$ and the $j$th data qubit
by $r+j$.
Let $\mathbb T=\cbra{\omega\in\C\colon\abs{\omega}=1}$ be the unit circle.
Set $R=2^r$ and identify binary addresses
$x\in\cbra{0,1}^r$ with labels in $\sbra{R}$
in lexicographic order.
For $j\in\sbra{M}$, define
\begin{equation*}
  Z_j
  = I^{\ot\rbra{j-1}}\ot Z\ot I^{\ot\rbra{M-j}},
\end{equation*}
which acts as Pauli $Z$ on the $j$th data qubit and
as the identity on the remaining data qubits.

\begin{definition}[Phase-masked unitary] \label{def:unitary-family}
  For an address function $f \colon \cbra{0,1}^r \to \sbra{M}$ and a phase mask $g = \rbra{g_x}_{x \in \cbra{0,1}^r} \in \mathbb T^R$, define
  \begin{equation*}
    U_{f,g} = \sum_{x \in \cbra*{0,1}^r} g_x \ketbra{x}{x} \ot Z_{f\rbra*{x}}.
\end{equation*}
Equivalently, its action on computational-basis states is
\begin{equation} \label{eq:unitary-def}
  U_{f,g}\ket{x,z} = g_x\rbra*{-1}^{z_{f\rbra*{x}}}\ket{x,z},
\end{equation}
for any $x \in \cbra{0,1}^r$ and $z \in \cbra{0,1}^M$.
\end{definition}

By \cref{eq:unitary-def}, $U_{f,g}$ is diagonal with
all diagonal entries of modulus one, and hence is unitary.

Write $\operatorname{Unif}\rbra{\mathbb T}$ for the uniform
(Haar) distribution on $\mathbb T$, that is, the law of
$e^{i\theta}$ for $\theta$ uniform on $\interval[open right]{0}{2\pi}$.
For an unknown distribution $p\in\Delta\rbra{\sbra{M}}$,
define $\U\rbra{p}$ to be the ensemble of unitaries
$U_{f,g}$ obtained by sampling
$f\rbra*{x}\overset{\mathrm{i.i.d.}}{\sim}p$ and
$g_x\overset{\mathrm{i.i.d.}}{\sim}
\operatorname{Unif}\rbra{\mathbb T}$
for every $x\in\cbra*{0,1}^r$, with the two families
independent of each other.
For all integers $a$ and $b$, a phase
$g_x\sim\operatorname{Unif}\rbra{\mathbb T}$ satisfies
\begin{equation}
\label{eq:haar-phase-orthogonality}
  \Ex\sbra*{g_x^{a}\,\overline{g_x}^{\,b}}
  =\frac1{2\pi}\int_0^{2\pi}e^{i\rbra{a-b}\theta}\,d\theta
  =\begin{cases}
    1,& a=b,\\
    0,& a\neq b.
  \end{cases}
\end{equation}

Using the address ordering above, we identify $f$ with
$\rbra{f\rbra{1},\ldots,f\rbra{R}}\in\sbra{M}^R$.
We write $f\sim p^{\ot R}$ to mean that these values
are sampled independently from $p$.

\begin{lemma}[Reduction of YES instances]
\label{lem:completeness}
Let $1\leq s\leq M$ be an integer.
If $p\in\Delta\rbra{\sbra{M}}$ satisfies
$\abs{\supp\rbra{p}}\leq s$,
then every unitary in the support of the ensemble
$\U\rbra{p}$ is an $\rbra{r+s}$-junta unitary on $n$ qubits.
\end{lemma}

\begin{proof}
Fix any realization $U_{f,g}$ from $\U\rbra{p}$.
By the sampling rule, $\im\rbra{f}\subseteq\supp\rbra{p}$.
Set $T=\sbra{r}\cup\cbra{r+j\colon j\in\im\rbra{f}}$.
Since $T$ contains all address qubits and every data qubit
selected by $f$, we have
$U_{f,g}=V_T\ot I_{\overline{T}}$
for some unitary $V_T$.
Since $\abs{T}=r+\abs{\im\rbra{f}}\leq r+s$,
$U_{f,g}$ is an $\rbra{r+s}$-junta unitary.
\end{proof}

For the NO instances, we fix an integer $1\leq s\leq M$ and bound
$\Inf_{\overline{T}}\sbra{U_{f,g}}$
for all $T\subseteq\sbra{n}$ with $\abs{T}\leq r+s$.
We distinguish two cases according to whether $T$
contains all address qubits.
The first case uses the support size gap of $p$,
while the second uses the random phases.
Both cases use the Pauli coefficients of $U_{f,g}$.

For a fixed address function $f$, define its empirical
distribution $p_f$ by
\begin{equation*}
  p_f\rbra{j}
  =\frac1R\abs*{\cbra*{
    x\in\cbra{0,1}^r\colon f\rbra{x}=j}},
\end{equation*}
for every $j\in\sbra{M}$.
Its support satisfies
$\supp\rbra{p_f}=\im\rbra{f}$.

For $j\in\sbra{M}$, let $e_j\in\cbra{0,1}^M$ be the $j$th standard basis vector.
For a binary vector $y$, let $3y$ have entry $3$ exactly where $y$ has entry $1$, so that $3y\in\Z_4^r$ for $y\in\cbra{0,1}^r$ and $3e_j\in\Z_4^M$ for $j\in\sbra{M}$.

\begin{lemma}[Pauli coefficients of phase-masked unitaries]\label{lem:hard-ensemble-pauli}
Let $f\colon\cbra{0,1}^r\to\sbra{M}$ and $g\in\mathbb T^R$.
Every nonzero Pauli coefficient of $U_{f,g}$ has a label $\rbra{3y,3e_j}$ with $y\in\cbra{0,1}^r$ and $j\in\sbra{M}$, and
\[
  \wh{U}_{f,g}\rbra{3y,3e_j}
  =
  \frac1R
  \sum_{\substack{x\in\cbra{0,1}^r\\f\rbra{x}=j}}
  g_x\rbra*{-1}^{x\cdot y}.
\]
Moreover, $\Inf_{r+j}\sbra{U_{f,g}}=p_f\rbra{j}$ for every $j\in\sbra{M}$.
For every $S\subseteq\sbra{M}$ and $T=\sbra{r}\cup\cbra{r+j\colon j\in S}$, we have $\Inf_{\overline{T}}\sbra{U_{f,g}}=1-p_f\rbra{S}$.
\end{lemma}

\begin{proof}
Since $U_{f,g}$ is diagonal, every Pauli string containing
$\sigma_1$ or $\sigma_2$ has zero coefficient.
Binary character orthogonality on the data register gives the remaining labels and the stated coefficients.
Each Pauli string with a nonzero coefficient acts
nontrivially on exactly one data qubit.
Since $\abs{g_x}=1$, Parseval's identity on the address register
(\cref{fact:boolean-fourier}) gives
\begin{equation*}
  \Inf_{r+j}\sbra{U_{f,g}}
  =
  \sum_{y\in\cbra{0,1}^r}
  \abs*{\wh{U}_{f,g}\rbra{3y,3e_j}}^2
  =
  \frac1R
  \sum_{\substack{x\in\cbra{0,1}^r\\f\rbra{x}=j}}
  \abs{g_x}^2
  =p_f\rbra{j}.
\end{equation*}
Since $T$ contains all address qubits, a Pauli string
with a nonzero coefficient acts nontrivially outside $T$
exactly when its data label is $3e_j$ with $j\notin S$.
Summing the data-qubit influences over $j\in\sbra{M}\setminus S$ gives $\Inf_{\overline{T}}\sbra{U_{f,g}}=1-p_f\rbra{S}$.
\end{proof}

\begin{lemma}[Reduction of NO instances]
\label{lem:hard-no-reduction}
Let $1\leq s\leq M$ be an integer.
Let $0<\eps\leq1/16$ and set
$\delta=\rbra{4/3}\rbra{2\eps^2-\eps^4}$.
If $p\in\Delta\rbra{\sbra{M}}$ satisfies
$\mass_s\rbra{p}\leq1-\delta$,
then, with probability at least
$1-16s/\rbra{R\delta^2}-2r e^{-R/16}$,
a draw $U_{f,g}\sim\U\rbra{p}$ satisfies
\begin{equation*}
  \distU\rbra{U_{f,g},\calU_{r+s}}\geq\eps,
  \qquad
  \distC\rbra{\Phi_{U_{f,g}},\calC_{r+s}}\geq\eps.
\end{equation*}
\end{lemma}

\begin{proof}
We show that, with the stated probability, $\Inf_{\overline T}\sbra{U_{f,g}}\ge2\eps^2-\eps^4$ for every $T\subseteq\sbra{n}$ with $\abs{T}\le r+s$, distinguishing two cases according to whether $T$ contains all address qubits.

\noindent\emph{Case 1: $T$ contains all address qubits.}
In a draw $U_{f,g}\sim\U\rbra{p}$, the values $f\rbra{x}$ are independent samples from $p$.
By Cauchy--Schwarz and the definition of $\mass_s$,
\[
 \abs*{\mass_s\rbra{p_f}-\mass_s\rbra{p}}^2
 \leq s\sum_{j=1}^M\rbra*{p_f\rbra{j}-p\rbra{j}}^2.
\]
Since $Rp_f\rbra{j}$ has distribution $\operatorname{Bin}\rbra{R,p\rbra{j}}$, the right side has expectation at most $s/R$ over $f$.
By Markov's inequality (\cref{fact:markov}), the event $\mass_s\rbra{p_f}\leq\mass_s\rbra{p}+\delta/4$
has probability at least $1-16s/\rbra{R\delta^2}$ over $f$.
On this event, every $S\subseteq\sbra{M}$ with
$\abs{S}\leq s$ satisfies
\begin{equation*}
  1-p_f\rbra{S}
  \geq 1-\mass_s\rbra{p_f}
  \geq 1-\mass_s\rbra{p}-\frac{\delta}{4}
   \geq\frac{3\delta}{4}.
\end{equation*}
For every $T\subseteq\sbra{n}$ with $\sbra{r}\subseteq T$ and $\abs{T}\leq r+s$, the set
$S=\cbra{j\in\sbra{M}\colon r+j\in T}$ has $\abs{S}\leq s$.
Since the event depends only on $f$,
\cref{lem:hard-ensemble-pauli} gives
$\Inf_{\overline{T}}\sbra{U_{f,g}}=1-p_f\rbra{S}\geq3\delta/4$
for every such $T$ and every phase mask $g$.

\noindent\emph{Case 2: $T$ omits an address qubit.}
A set $T$ that omits an address qubit may contain more than $s$ data qubits, so we bound the influence through an omitted address qubit.
Fix $f$ and an address coordinate $i\in\sbra{r}$, and write each address as $(u,b)$ by inserting $b\in\cbra{0,1}$ at coordinate $i$ into $u\in\cbra{0,1}^{r-1}$.
By \cref{lem:hard-ensemble-pauli} and Parseval's identity (\cref{fact:boolean-fourier}) on the remaining address coordinates,
\[
 \Inf_i\sbra{U_{f,g}}
 =\frac12-\frac1R
   \sum_{\substack{u\in\cbra{0,1}^{r-1}\\f(u,0)=f(u,1)}}
   \Real\rbra*{g_{(u,0)}\overline{g_{(u,1)}}}.
\]
Define $X_u=1/2$ when $f(u,0)\ne f(u,1)$ and
$X_u=\rbra*{1-\Real\rbra*{g_{(u,0)}\overline{g_{(u,1)}}}}/2$ otherwise.
Then $\Inf_i\sbra{U_{f,g}}=(2/R)\sum_u X_u$.
The variables $X_u$ lie in $[0,1]$ and are independent because distinct $u$ use disjoint phase pairs.
By \cref{eq:haar-phase-orthogonality} and independence of the two phases, $\Ex\sbra{g_{(u,0)}\overline{g_{(u,1)}}}=0$, so every $X_u$ has mean $1/2$.
By Hoeffding's inequality (\cref{fact:hoeffding}),
\[
 \Pr_g\sbra*{\Inf_i\sbra{U_{f,g}}<1/4}\leq2e^{-R/16}.
\]
A union bound shows that all $r$ address influences are at least $1/4$ with probability at least $1-2r e^{-R/16}$ over $g$, for every fixed $f$. On this event, every $T$ omitting an address coordinate $i$ satisfies
$\Inf_{\overline T}\sbra{U_{f,g}}\geq\Inf_i\sbra{U_{f,g}}\geq1/4$ by monotonicity.

\noindent\emph{Combining the cases.}
The failure bound of Case~2 holds for every $f$, so it also holds after averaging over $f$.
A union bound shows that both bounds hold with probability
at least $1-16s/\rbra{R\delta^2}-2r e^{-R/16}$
over $U_{f,g}\sim\U\rbra{p}$.
Since $\eps\le1/16$, we have $3\delta/4=2\eps^2-\eps^4<1/4$, so either case gives
\begin{equation*}
  \Inf_{\overline{T}}\sbra{U_{f,g}}
  \geq\frac{3\delta}{4}=2\eps^2-\eps^4
\end{equation*}
for every $T\subseteq\sbra{n}$ with $\abs{T}\leq r+s$.
By \cref{lem:influence-certificate}(ii),
$\distU\rbra{U_{f,g},\calU_{r+s}}\geq\eps$.
On the same event, \cref{lem:unitary-channel-influence-certificate} gives
\begin{equation*}
  \distC\rbra{\Phi_{U_{f,g}},\calC_{r+s}}^2
  \geq2\eps^2-\eps^4
    -\frac12\rbra{2\eps^2-\eps^4}^2
  \geq\eps^2,
\end{equation*}
where the last inequality uses $\eps\leq1/16$.
\end{proof}

\subsubsection{Exact classicalization of forward queries}\label{subsec:classicalization}
Let $\operatorname{Acc}_{\calA}(\Phi)$ denote the acceptance probability of an algorithm $\calA$ in the query model of \cref{subsec:channel-access} with oracle $\Phi$.
Pad an algorithm making at most $q$ queries with dummy queries so that it makes exactly $q$ queries.
The padded algorithm is described by an oracle-independent initial state $\rho_{\mathrm{init}}$, interleaving channels $A_0,\ldots,A_q$, and an accepting POVM element $0\leq\Lambda_{\calA}\leq I$.
The workspace includes private randomness and measurement records, so these maps allow adaptive control.
For an algorithm using only uncontrolled forward queries to a unitary $U$,
write $\operatorname{Acc}_{\calA}(U)=\operatorname{Acc}_{\calA}(\Phi_U)$.

We adapt the path expansion and histogram grouping in the proof of~\cite[Theorem~3.1]{TW25}, based on the compressed oracle method~\cite{Zha19} (see also~\cite{FLMNW26}).
Our address labels require an additional step, because each phase-averaged sector reads at most $q$ addresses, but different sectors read different addresses.
Symmetrization and pointwise normalization below convert this decomposition into a classical acceptance rule in $[0,1]$ independent of $p$.

\begin{theorem}[Exact forward-query classicalization]
\label{thm:classicalization}
Fix integers $M\geq2$, $r\geq1$, and $q\geq0$, and let $n=r+M$.
Let $\calA$ be a binary-output quantum algorithm,
independent of the distribution $p$ that defines the ensemble $\U\rbra{p}$, that makes at most $q$ queries with forward-only
access to an unknown $n$-qubit unitary.
There is a function
$F_{\calA}:\sbra{M}^q\to[0,1]$,
independent of $p$, such that for every
$p\in\Delta\rbra{\sbra{M}}$,
\begin{equation*}
  \Ex_{U_{f,g}\sim\U\rbra*{p}}
    \sbra*{\operatorname{Acc}_{\calA}\rbra*{U_{f,g}}}
  =\Ex_{Y\sim p^{\ot q}}
    \sbra*{F_{\calA}\rbra*{Y}}.
\end{equation*}
\end{theorem}

The ensemble is sampled once and held fixed throughout the computation.
Sampling $Y\sim p^{\ot q}$ and accepting with probability $F_{\calA}(Y)$ defines a classical $q$-sample tester for every $p$.
Only existence is needed, since the classical sample lower bound does not restrict computation.

For $q\geq1$, we group query paths by address histograms. For integers $m\geq1$ and $q\geq0$, define
\begin{equation*}
 \mathfrak H_{m,q}=\cbra*{h\in\cbra{0,\ldots,q}^m\colon\sum_{j=1}^m h_j=q}.
\end{equation*}
For $h\in\mathfrak H_{R,q}$, the set $\supp\rbra{h}$ consists of its visited addresses, so $\abs{\supp\rbra{h}}\leq q$.

\begin{lemma}[Phase-averaged path-sector decomposition]
\label{lem:path-sector}
Let $q\geq1$, and let $V_0,\ldots,V_q$ be unitaries that do not depend on the address function $f$, the phase mask $g$, or the distribution $p$ defining $\U\rbra{p}$.
For an address function $f$ and a phase mask $g$, let $\ket{\psi_{f,g}}=V_qU_{f,g}V_{q-1}\cdots V_1U_{f,g}V_0\ket0$, where each $U_{f,g}$ acts on the query register.
Then there are vectors $\ket{\phi_{f,h}}$, indexed by address functions $f$ and address histograms $h\in\mathfrak H_{R,q}$, such that for every $f$,
\begin{equation*}
  \Ex_{g\sim\operatorname{Unif}\rbra{\mathbb T^R}}\sbra*{
    \ketbra{\psi_{f,g}}{\psi_{f,g}}}
  =\sum_{h\in\mathfrak H_{R,q}}
    \ketbra{\phi_{f,h}}{\phi_{f,h}},
\end{equation*}
and each $\ket{\phi_{f,h}}$ is independent of $g$ and $p$ and depends on $f$ only through its restriction to $\supp\rbra{h}$.
\end{lemma}

\begin{proof}
Expanding each query using the address decomposition
of $U_{f,g}$ gives
\begin{equation*}
  \ket{\psi_{f,g}}
  =\sum_{\rbra{x_1,\ldots,x_q}\in\sbra{R}^q}
    \rbra*{\prod_{t=1}^q g_{x_t}}
    V_q\rbra*{\ketbra{x_q}{x_q}\ot Z_{f\rbra{x_q}}}
    V_{q-1}\cdots
    V_1\rbra*{\ketbra{x_1}{x_1}\ot Z_{f\rbra{x_1}}}
    V_0\ket{0}.
\end{equation*}
For an address sequence $\rbra{x_1,\ldots,x_q}$ with
address histogram $h$, define its phase factor by
\begin{equation*}
  \chi_h\rbra{g}
  =\prod_{x\in\sbra{R}}g_x^{h_x}
  =\prod_{t=1}^q g_{x_t}.
\end{equation*}

Group the address sequences by their address histogram $h$.
For each group, factor out the common phase
$\chi_h\rbra{g}$ and denote the sum of the remaining
vectors by $\ket{\phi_{f,h}}$.
This gives, for every phase mask $g\in\mathbb T^R$,
\begin{equation}
\label{eq:path-grouping}
  \ket{\psi_{f,g}}
  =\sum_{h\in\mathfrak H_{R,q}}
    \chi_h\rbra{g}\ket{\phi_{f,h}}.
\end{equation}
Every address sequence with address histogram $h$ visits
only addresses in $\supp\rbra{h}$.
After removing the phase factor, the corresponding vectors
depend only on the fixed circuit and the labels at these addresses.
Thus $\ket{\phi_{f,h}}$ is independent of $g$ and $p$
and depends on $f$ only through its restriction to $\supp\rbra{h}$.

Paths with the same address histogram may interfere.
We now average over $g$ uniform on $\mathbb T^R$ to remove
cross terms between distinct address histograms.
Let $h,h'\in\mathfrak H_{R,q}$.
By independence of the coordinates of $g$ and
\cref{eq:haar-phase-orthogonality} with exponents $h_x$ and $h'_x$,
\begin{equation*}
  \Ex_{g\sim\operatorname{Unif}\rbra{\mathbb T^R}}
    \sbra*{\chi_h\rbra{g}\overline{\chi_{h'}\rbra{g}}}
  =\prod_{x\in\sbra{R}}
    \Ex_{g_x\sim\operatorname{Unif}\rbra{\mathbb T}}
      \sbra*{g_x^{h_x}\,\overline{g_x}^{\,h'_x}}
  =\begin{cases}
    1,& h=h',\\
    0,& h\neq h'.
  \end{cases}
\end{equation*}
Expand $\ketbra{\psi_{f,g}}{\psi_{f,g}}$ using
\cref{eq:path-grouping} and average over $g$.
The phase-orthogonality identity above removes all terms with $h\neq h'$.
The terms with $h=h'$ give the claimed decomposition.
\end{proof}

The vectors $\ket{\phi_{f,h}}$ need not be normalized or mutually orthogonal.
Taking the trace in \cref{lem:path-sector} gives $\sum_h\norm{\phi_{f,h}}^2=1$ for each fixed $f$.
The next lemma will turn an identity valid in expectation for every distribution into pointwise normalization.

\begin{lemma}[Pointwise normalization]
\label{lem:pointwise-normalization}
Let $M,q\geq1$ be integers, and let
$W:\sbra{M}^q\to\R$ be invariant under coordinate permutations.
Suppose that, for every $p\in\Delta\rbra{\sbra{M}}$,
\begin{equation*}
  \Ex_{Y\sim p^{\ot q}}\sbra*{W\rbra{Y}}=1.
\end{equation*}
Then $W\rbra{y}=1$ for every $y\in\sbra{M}^q$.
\end{lemma}

\begin{proof}
Let $c=\rbra{c_1,\ldots,c_M}\in\mathfrak H_{M,q}$.
A tuple $y\in\sbra{M}^q$ has label histogram $c$ if each
label $j\in\sbra{M}$ occurs exactly $c_j$ times.
Any two tuples with the same label histogram differ
only by a coordinate permutation.
Since $W$ is invariant under such permutations, it has
a common value on these tuples, which we denote by $w_c$.

For $x=\rbra{x_1,\ldots,x_M}\in\R^M$, define
the homogeneous polynomial of degree $q$
\begin{equation*}
  P\rbra{x}=
  \sum_{c\in\mathfrak H_{M,q}}
    \frac{q!}{c_1!\cdots c_M!}\,
    w_c\prod_{j=1}^M x_j^{c_j}.
\end{equation*}
Grouping tuples by their label histograms gives
$P\rbra{p}
=\Ex_{Y\sim p^{\ot q}}\sbra{W\rbra{Y}}=1$
for every $p\in\Delta\rbra{\sbra{M}}$.
For $x\in\R^M$ with positive coordinates,
homogeneity therefore gives
\begin{equation*}
  P\rbra{x}
  =\rbra*{\sum_{j=1}^M x_j}^q
    P\rbra*{\frac{x}{\sum_{j=1}^M x_j}}
  =\rbra*{\sum_{j=1}^M x_j}^q.
\end{equation*}
The two polynomials agree on a nonempty open set,
so they are identical.
Comparing coefficients gives $w_c=1$ for every
$c\in\mathfrak H_{M,q}$.
Thus $W\rbra{y}=1$ for every $y\in\sbra{M}^q$.
\end{proof}

\begin{proof}[Proof of \cref{thm:classicalization}]
If $q=0$, $\calA$ makes no oracle queries,
so its acceptance probability is independent of $U_{f,g}$.
Define $F_{\calA}$ on the empty tuple to equal
this probability.
The required identity then holds for every $p$.
From now on, assume that $q\geq1$.

Query-register operators below are implicitly tensored with the workspace identity.
To pad $\calA$ to exactly $q$ queries, keep a classical stop flag and preserve its stored output after termination.
At each query slot, conditionally on the stop flag, swap $\mathsf Q$ with an additional register initialized to $\ket{0^n}$, apply the oracle unconditionally, undo the swap, and discard that register.
Active branches implement the original query, whereas on stopped branches trace preservation leaves the live registers and stored output unchanged.

Next, purify $\rho_{\mathrm{init}}$ and replace each known channel $A_t$ by an oracle-independent unitary dilation $V_t$ on a new environment register.
Prepare the added registers at the start and keep them throughout.
Their partial trace recovers the padded algorithm for every fixed oracle, so the final pure state is
\begin{equation*}
  \ket{\psi_{f,g}}
  =V_q U_{f,g}V_{q-1}\cdots V_1 U_{f,g}V_0\ket{0},
\end{equation*}
where $V_0$ includes initial-state preparation and all $V_t$ are independent of $f,g,p$. If $\mathsf E$ collects the added registers, then $\Lambda=\Lambda_{\calA}\ot I_{\mathsf E}$ satisfies $0\leq\Lambda\leq I$ and
\[
 \operatorname{Acc}_{\calA}(U_{f,g})
 =\bra{\psi_{f,g}}\Lambda\ket{\psi_{f,g}}.
\]
We apply \cref{lem:path-sector} to these unitaries $V_0,\ldots,V_q$ and let $\ket{\phi_{f,h}}$ be the resulting vectors.

For each $h\in\mathfrak H_{R,q}$, order the addresses in $\supp\rbra{h}$ as $x_{h,1}<\cdots<x_{h,\abs{\supp\rbra{h}}}$.
Given $y\in\sbra{M}^q$, define $\ket{\phi_h(y)}=\ket{\phi_{f,h}}$ for any $f$ satisfying $f(x_{h,i})=y_i$ for $1\leq i\leq\abs{\supp\rbra{h}}$.
By \cref{lem:path-sector}, this is well defined, depends only on the first $\abs{\supp\rbra{h}}\leq q$ entries of $y$, and is independent of $p$. Repeated visits to an address use the same label.
For a fixed $y$, different histograms may assign different entries to the same address, so the vectors need not come from a single $f$. We therefore symmetrize before applying \cref{lem:pointwise-normalization}.

Let $\mathfrak S_q$ be the symmetric group on $\sbra{q}$.
For $y\in\sbra{M}^q$ and $\pi\in\mathfrak S_q$, define
$y_\pi=
\rbra{y_{\pi\rbra{1}},\ldots,y_{\pi\rbra{q}}}$.
Averaging the histogram contributions over all coordinate
permutations, define
\begin{equation*}
  F_{\calA}\rbra{y}
  =\frac1{q!}
    \sum_{\pi\in\mathfrak S_q}\sum_{h\in\mathfrak H_{R,q}}
    \bra{\phi_h\rbra{y_\pi}}\Lambda
    \ket{\phi_h\rbra{y_\pi}}, \qquad
  W\rbra{y}
  =\frac1{q!}
    \sum_{\pi\in\mathfrak S_q}\sum_{h\in\mathfrak H_{R,q}}
    \braket{\phi_h\rbra{y_\pi}}{\phi_h\rbra{y_\pi}}.
\end{equation*}
Both functions are invariant under permutations of the entries
of $y$ and are independent of $p$.
Since $0\leq\Lambda\leq I$,
\begin{equation}
\label{eq:sample-weight-order}
  0\leq F_{\calA}\rbra{y}\leq W\rbra{y},
\end{equation}
for every $y \in \sbra{M}^q$.
Fix $p$, $h$, and $\pi$. For $f\sim p^{\ot R}$ and $Y\sim p^{\ot q}$, both $(f(x_{h,1}),\ldots,f(x_{h,\abs{\supp\rbra{h}}}))$ and $(Y_{\pi(1)},\ldots,Y_{\pi(\abs{\supp\rbra{h}})})$ consist of $\abs{\supp\rbra{h}}$ independent samples from $p$. Hence
\begin{equation*}
  \Ex_{Y\sim p^{\ot q}}\sbra*{
    \ketbra{\phi_h\rbra{Y_\pi}}{\phi_h\rbra{Y_\pi}}}
  =\Ex_{f\sim p^{\ot R}}\sbra*{
    \ketbra{\phi_{f,h}}{\phi_{f,h}}}.
\end{equation*}
Summing over $h$ and averaging over $\pi$ gives the ensemble-averaged state by \cref{lem:path-sector}.
Taking traces against $\Lambda$ and $I$ therefore gives, for every $p\in\Delta\rbra{\sbra{M}}$,
\begin{equation}
\label{eq:sample-acceptance}
  \Ex_{Y\sim p^{\ot q}}\sbra*{F_{\calA}\rbra*{Y}}
  =\Ex_{U_{f,g}\sim\U\rbra*{p}}
      \sbra*{\operatorname{Acc}_{\calA}\rbra*{U_{f,g}}},
  \qquad
  \Ex_{Y\sim p^{\ot q}}\sbra*{W\rbra*{Y}}=1.
\end{equation}

Since $W$ is independent of $p$ and permutation invariant, the second identity in \cref{eq:sample-acceptance} and \cref{lem:pointwise-normalization} give $W(y)=1$ for every $y$.
Then \cref{eq:sample-weight-order} gives $0\leq F_{\calA}(y)\leq1$.
Its independence of $p$ and the first identity in \cref{eq:sample-acceptance} complete the proof.
\end{proof}

\subsubsection{Proof of the lower bound}\label{subsec:unitary-lower-forward-proof}
We now choose the ensemble parameters and apply
\cref{thm:classicalization} to obtain a forward-query lower bound from
the classical sample lower bound.

\begin{lemma}[Ensemble parameters]
\label{lem:ensemble-parameters}
Fix a sufficiently large absolute constant $C$, let $k$ be sufficiently large depending on $C$,
and fix an integer $n\geq2k$ and $2^{-k/16}\leq\eps\leq1/16$.
Set
\[
  \delta=\frac43\rbra{2\eps^2-\eps^4},\quad
  r=\ceil*{\log\rbra{Ck/\delta^2}},\quad
  R=2^r,\quad s=k-r,\quad M=n-r.
\]
Then $\eps^2\leq\delta<1/32$, $k/2\leq s\leq k$, $M\geq2s$, and
\[
  \frac{16s}{R\delta^2}+2re^{-R/16}<\frac1{20}.
\]
\end{lemma}

\begin{proof}
Since $\eps\leq1/16$, we have $\eps^2\leq\delta\leq8\eps^2/3<1/32$.
By the definition of $r$, $\delta\geq\eps^2$, and $\log\rbra{1/\eps}\leq k/16$,
\[
  r\leq\log\rbra{Ck/\delta^2}+1\leq\frac{k}{4}+\log k+\log C+1\leq\frac{k}{2}
\]
for all sufficiently large $k$, uniformly over the stated range of $\eps$.
Hence $k/2\leq s\leq k$ and $M-2s=n-2k+r\geq0$.
By $s\leq k$, $r\leq k/2$, and $R\geq Ck/\delta^2$,
\[
  \frac{16s}{R\delta^2}+2re^{-R/16}
  \leq\frac{16}{C}+ke^{-Ck/\rbra{16\delta^2}}
  <\frac1{20}
\]
for sufficiently large $C$ and then sufficiently large $k$.
\end{proof}

\begin{lemma}[Classicalized acceptance gap]\label{lem:ensemble-to-samples}
Fix an integer $r\ge1$, a sufficiently large integer $s$, an integer $M\ge2s$, and $0<\delta<1/32$.
Let $\calA$ be a binary-output quantum algorithm that makes at most $q$ queries with forward-only access to an unknown $(r+M)$-qubit unitary.
Suppose that $\Ex_{U_{f,g}\sim\U\rbra{p}}\sbra{\operatorname{Acc}_{\calA}\rbra{U_{f,g}}}$ is at least $2/3$ for every $p\in\Delta\rbra{\sbra{M}}$ with $\abs{\supp\rbra{p}}\le s$ and at most $2/5$ for every $p\in\Delta\rbra{\sbra{M}}$ with $\mass_s\rbra{p}\le1-\delta$.
Then $q=\Omega\rbra{s/\rbra{\delta\log s}}$.
\end{lemma}

\begin{proof}
By \cref{thm:classicalization}, there is a function $F_{\calA}:\sbra{M}^q\to[0,1]$, independent of $p$, whose expectation under $Y\sim p^{\ot q}$ equals $\Ex_{U_{f,g}\sim\U\rbra{p}}\sbra{\operatorname{Acc}_{\calA}\rbra{U_{f,g}}}$ for every $p\in\Delta\rbra{\sbra{M}}$.
Sampling $Y\sim p^{\ot q}$ and accepting with probability $F_{\calA}\rbra{Y}$ therefore defines a classical tester that succeeds with probability at least $3/5$ on every promised distribution.
Since $M\geq2s$, this tester also works for distributions on $\sbra{2s}$, viewed as distributions on $\sbra{M}$ with zero mass outside $\sbra{2s}$.
This extension preserves both support size and $\mass_s$.
By \cref{thm:classical-support-lb} with success probability $3/5$, we obtain $q=\Omega\rbra{s/\rbra{\delta\log s}}$.
\end{proof}

\begin{proof}[Proof of \cref{thm:unitary-junta-lower-bound-forward}]
We reduce \cref{prob:classical-support-testing} to
\cref{prob:unitary-junta-testing} using the ensemble $\U\rbra{p}$.
Let $C,\delta,r,R,s,M$ be as in \cref{lem:ensemble-parameters}.
For $p\in\Delta\rbra{\sbra{M}}$, each unitary $U_{f,g}$ in the support of $\U\rbra{p}$ acts on $r+M=n$ qubits.

\noindent\emph{YES instances.}
If $\abs{\supp\rbra{p}}\leq s$, then by \cref{lem:completeness} with $r+s=k$,
every $U_{f,g}$ in the support of $\U\rbra{p}$ satisfies $U_{f,g}\in\calU_k$.

\noindent\emph{NO instances.}
If $\mass_s\rbra{p}\leq1-\delta$, then by \cref{lem:hard-no-reduction} with $r+s=k$
and by \cref{lem:ensemble-parameters}, with probability at least $19/20$ over $U_{f,g}\sim\U\rbra{p}$,
\[
  \distU\rbra{U_{f,g},\calU_k}\geq\eps.
\]

\noindent\emph{Completing the reduction.}
Let $\calA$ be a tester for \cref{prob:unitary-junta-testing} that has forward-only access,
makes at most $q$ queries, and succeeds with probability at least $2/3$ on every promised input.
By the YES and NO guarantees above, the mean acceptance probability of $\calA$ over $\U\rbra{p}$ is at least $2/3$ when $\abs{\supp\rbra{p}}\le s$ and at most $\frac{19}{20}\cdot\frac13+\frac1{20}<\frac25$ when $\mass_s\rbra{p}\le1-\delta$.
Hence \cref{lem:ensemble-to-samples} gives $q=\Omega\rbra{s/\rbra{\delta\log s}}=\Omega\rbra{k/\rbra{\eps^2\log k}}$, since $s=\Theta\rbra{k}$ and $\delta=\Theta\rbra{\eps^2}$.
\end{proof}

\section{Junta Channel Testing}\label{sec:channel}

This section studies the following problem with the channel access of \cref{subsec:channel-access}.

\begin{problem}[Junta Channel Testing]
\label{prob:channel-junta-testing}
Fix integers $n\ge1$ and $0\le k\le n$, and fix $0<\eps\leq1$.
Given channel access to an unknown $n$-qubit channel $\Phi\in C\rbra{\calH}$,
distinguish between the following two cases with probability at least $2/3$:
\begin{itemize}
  \item (YES) $\Phi\in\calC_k$.
  \item (NO) $\distC\rbra{\Phi,\calC_k}\geq\eps$.
\end{itemize}
\end{problem}

We prove the upper bound and then the matching lower bound for \cref{prob:channel-junta-testing}.

\subsection{Upper bound}\label{subsec:channel-upper}

\begin{theorem}[Upper bound for Junta Channel Testing]
\label{thm:channel-junta-upper-bound}
Fix integers $1\le k\le n$ and $0<\eps\le1$. Given channel access to an $n$-qubit channel $\Phi$, there is a nonadaptive quantum algorithm that distinguishes $\Phi\in\calC_k$ from $\distC\rbra{\Phi,\calC_k}\ge\eps$ with success probability at least $2/3$ and uses $O(k/(\eps^2\log(ek)))$ queries.
Every query acts on a product input state without ancillas and is followed by single-qubit measurements.
\end{theorem}

Channel queries yield independent subset samples through influence sampling.
\Cref{alg:channel} summarizes the tester constructed in the following proof.

\begin{algorithm}[h]
\caption{\textsc{ChannelJuntaTest}$\rbra{\Phi;n,k,\eps}$}\label{alg:channel}
\begin{algorithmic}[1]
\Require Channel access to an $n$-qubit channel $\Phi$, integers $1\le k\le n$, and $0<\eps\le1$
\Ensure \Accept\ or \Reject
\State $\delta\gets2\eps^2/3$
\State $t_{\max}\gets$ largest number of samples used by \textsc{UnionSizeTest}$\rbra{\sbra{n},k,\delta}$ \Comment{$t_{\max}=O(k/(\delta\log(ek)))$}
\For{$t=1,\ldots,t_{\max}$}
  \State $X_t\gets\InfluenceSample\rbra{\Phi}$ \Comment{one channel query, \cref{lem:influence-sample-round}}
\EndFor
\State \Return \textsc{UnionSizeTest}$\rbra{\sbra{n},k,\delta}$ run on $X_1,\ldots,X_{t_{\max}}$ \Comment{\cref{thm:set-testing}}
\end{algorithmic}
\end{algorithm}

\begin{proof}[Proof of \cref{thm:channel-junta-upper-bound}]
Apply \cref{thm:set-testing} with $\Omega=\sbra{n}$ and $\delta=2\eps^2/3$, generating each requested sample $X$ by an independent run of \InfluenceSample$\rbra{\Phi}$ from \cref{lem:influence-sample-round}.
If $\Phi\in\calC_{T_*}$ with $\abs{T_*}\leq k$, then
$\Inf_{\overline{T_*}}\sbra{\Phi}=0$ by \cref{lem:junta-fourier-support},
so $X\subseteq T_*$ almost surely by \cref{lem:influence-sample-round}.
If $\distC\rbra{\Phi,\calC_k}\geq\eps$, then by \cref{lem:channel-influence-certificate}, $\Inf_{\overline T}\sbra{\Phi}\geq\eps^2$ for every $\abs{T}\leq k$.
By \cref{lem:influence-sample-round}, $\Pr\sbra{X\nsubseteq T}\geq\delta$.
Thus \cref{thm:set-testing} applies and succeeds with probability at least $2/3$ using $O(k/(\delta\log(ek)))=O(k/(\eps^2\log(ek)))$ samples.
Each run of \InfluenceSample\ applies one channel query to a product input state without ancillas, followed by single-qubit measurements, so the algorithm makes $O(k/(\eps^2\log(ek)))$ queries.
Since the sample bound of \cref{thm:set-testing} holds on every execution, all runs can be performed before their outputs are processed, so the algorithm is nonadaptive.
\end{proof}

\subsection{Lower bound}\label{subsec:channel-lower}

\begin{theorem}[Lower bound for Junta Channel Testing]
\label{thm:channel-junta-lower-bound}
Fix a sufficiently large integer $k$, an integer $n\ge2k$, and $2^{-k/16}\le\eps\le1/16$. Any tester that, given channel access to an $n$-qubit channel $\Phi$, distinguishes $\Phi\in\calC_k$ from $\distC\rbra{\Phi,\calC_k}\ge\eps$ with success probability at least $2/3$ requires $\Omega(k/(\eps^2\log k))$ queries.
The bound holds even when $\Phi$ is promised to be a unitary channel.
\end{theorem}

The unitary ensemble used for \cref{thm:unitary-junta-lower-bound-forward} also gives the channel lower bound, with each channel query simulated by one forward query.
Channel soundness requires separating NO instances from all junta channels, including noisy ones.

\begin{proof}[Proof of \cref{thm:channel-junta-lower-bound}]
We reduce \cref{prob:classical-support-testing} to
\cref{prob:channel-junta-testing} through the induced channels of the ensemble $\U\rbra{p}$.
Let $C,\delta,r,R,s,M$ be as in \cref{lem:ensemble-parameters}.
For $p\in\Delta\rbra{\sbra{M}}$, each unitary $U_{f,g}$ in the support of $\U\rbra{p}$ acts on $r+M=n$ qubits.

\noindent\emph{YES instances.}
If $\abs{\supp\rbra{p}}\leq s$, then by \cref{lem:completeness} with $r+s=k$,
every $U_{f,g}$ in the support of $\U\rbra{p}$ satisfies $U_{f,g}\in\calU_k$, and hence $\Phi_{U_{f,g}}\in\calC_k$.

\noindent\emph{NO instances.}
If $\mass_s\rbra{p}\leq1-\delta$, then by \cref{lem:hard-no-reduction} with $r+s=k$
and by \cref{lem:ensemble-parameters}, with probability at least $19/20$ over $U_{f,g}\sim\U\rbra{p}$,
\[
  \distC\rbra{\Phi_{U_{f,g}},\calC_k}\geq\eps.
\]

\noindent\emph{Completing the reduction.}
Let $\calA$ be a tester for \cref{prob:channel-junta-testing} that makes at most $q$ queries
and succeeds with probability at least $2/3$ on every promised input.
Each query to $\Phi_{U_{f,g}}$, including one on an input entangled with the workspace,
is implemented by one forward query to $U_{f,g}$.
Run in this way, $\calA$ is a forward-only algorithm making at most $q$ queries to $U_{f,g}$.
By the YES and NO guarantees above, the mean acceptance probability of $\calA$ over $\U\rbra{p}$ is at least $2/3$ when $\abs{\supp\rbra{p}}\le s$ and at most $\frac{19}{20}\cdot\frac13+\frac1{20}<\frac25$ when $\mass_s\rbra{p}\le1-\delta$.
Hence \cref{lem:ensemble-to-samples} gives $q=\Omega\rbra{s/\rbra{\delta\log s}}=\Omega\rbra{k/\rbra{\eps^2\log k}}$, since $s=\Theta\rbra{k}$ and $\delta=\Theta\rbra{\eps^2}$.
\end{proof}

\ifanon\else
\section*{Acknowledgments}
J.~B.\ is supported by the Quantum Advantage Turbo Charger (QATCH)
programme, funded by UKRI (Grant No.~UKRI4257), and by the
Quantum Advantage Pathfinder (QAP) programme, funded by EPSRC (Grant No.~EP/X026167/1).
M.~G.\ is supported by the National Key Research and Development Program of China (Grant No.
2023YFA1009403).
P.~Y.\ is supported by the National Natural Science Foundation of China (Grant Nos. 62332009 and 12347104), the Quantum Science and Technology—National Science and Technology Major Project (Grant No. 2021ZD0302901), the NSFC/RGC Joint Research Scheme (Grant No. 12461160276), the Natural Science Foundation of Jiangsu Province (No.
BK20243060), the Fundamental and Interdisciplinary Disciplines Breakthrough Plan of the Ministry of Education of China (No.
JYB2025XDXM118), the ``111 Center'' (No.~B26023), and the Fundamental Research Funds for the Central Universities (Grant No.~2026300376).
\fi

\section*{AI Methodology Statement}
OpenAI GPT-6 Astra, OpenAI GPT-5.6 Sol, and Anthropic Claude Opus 5.5 were used extensively throughout the preparation of this manuscript, including for mathematical exploration, checking proof arguments, drafting, restructuring, and language polishing.
Portions of the manuscript text were redrafted or modified with AI assistance across all sections.
The authors verified the correctness and originality of all content, including references, and take full responsibility for the final manuscript.

\bibliographystyle{alphaurl}
\phantomsection
\addcontentsline{toc}{section}{References}
\bibliography{main}

\end{document}